\documentclass[11pt,letterpaper]{article}

\def\stocmode{0}

\def\jamesmode{0}
\def\fastmode{0}
\def\nnrmode{1}
\def\entmaxmode{0}
\def\showauthornotes{0}
\def\showtableofcontents{1}
\def\showkeys{0}
\def\showdraftbox{1}
\def\showcolorlinks{1}
\def\usemicrotype{1}
\def\showfixme{1}

\ifnum\fastmode=0
\usepackage{etex}
\fi

\ifnum\fastmode=0
\usepackage[l2tabu, orthodox]{nag}
\fi

\usepackage{xspace,enumerate}

\usepackage[dvipsnames]{xcolor}

\ifnum\fastmode=0
\usepackage[T1]{fontenc}
\usepackage[full]{textcomp}
\fi

\usepackage[american]{babel}
\usepackage{mathtools}

\usepackage{amsthm}

\newtheorem{theorem}{Theorem}[section]
\newtheorem*{theorem*}{Theorem}

\newtheorem{proposition}[theorem]{Proposition}
\newtheorem*{proposition*}{Proposition}
\newtheorem{lemma}[theorem]{Lemma}
\newtheorem*{lemma*}{Lemma}
\newtheorem{corollary}[theorem]{Corollary}
\newtheorem*{conjecture*}{Conjecture}
\newtheorem{fact}[theorem]{Fact}
\newtheorem*{fact*}{Fact}

\newtheorem*{hypothesis*}{Hypothesis}

\theoremstyle{definition}
\newtheorem{definition}[theorem]{Definition}

\theoremstyle{remark}

\newtheorem*{claim*}{Claim}
\newtheorem{remark}[theorem]{Remark}
\newtheorem*{remark*}{Remark}

\newtheorem*{observation*}{Observation}

\usepackage[
letterpaper,
top=0.7in,
bottom=0.9in,
left=1in,
right=1in]{geometry}

\ifnum\jamesmode=0
\usepackage{pxfonts} % T1, lining figures in math, osf in text
\usepackage{textcomp} % required for special glyphs
\usepackage{amsmath}
\usepackage[scr=rsfso]{mathalfa}% \mathscr is fancier than \mathcal
\usepackage{bm} % load after all math to give access to bold math
\fi

\ifnum\jamesmode=1
\usepackage{lmodern}
\usepackage{textcomp} % required for special glyphs
\usepackage{amsmath,amssymb}
\usepackage[scr=rsfso]{mathalfa}% \mathscr is fancier than \mathcal
\usepackage{bm}
\fi

\ifnum\showkeys=1
\usepackage[color]{showkeys}
\fi

\makeatletter
\@ifpackageloaded{hyperref}
{}
{\usepackage[pagebackref]{hyperref}}

\ifnum\showcolorlinks=1
\hypersetup{
pagebackref=true,
colorlinks=true,
urlcolor=blue,
linkcolor=blue,
citecolor=OliveGreen}
\fi

\ifnum\showcolorlinks=0
\hypersetup{
pagebackref=true,
colorlinks=false,
pdfborder={0 0 0}}
\fi

\usepackage{prettyref}

\newcommand{\savehyperref}[2]{\texorpdfstring{\hyperref[#1]{#2}}{#2}}

\newrefformat{eq}{\savehyperref{#1}{\textup{(\ref*{#1})}}}
\newrefformat{lem}{\savehyperref{#1}{Lemma~\ref*{#1}}}
\newrefformat{def}{\savehyperref{#1}{Definition~\ref*{#1}}}
\newrefformat{thm}{\savehyperref{#1}{Theorem~\ref*{#1}}}
\newrefformat{cor}{\savehyperref{#1}{Corollary~\ref*{#1}}}
\newrefformat{cha}{\savehyperref{#1}{Chapter~\ref*{#1}}}
\newrefformat{sec}{\savehyperref{#1}{Section~\ref*{#1}}}
\newrefformat{app}{\savehyperref{#1}{Appendix~\ref*{#1}}}
\newrefformat{tab}{\savehyperref{#1}{Table~\ref*{#1}}}
\newrefformat{fig}{\savehyperref{#1}{Figure~\ref*{#1}}}
\newrefformat{hyp}{\savehyperref{#1}{Hypothesis~\ref*{#1}}}
\newrefformat{alg}{\savehyperref{#1}{Algorithm~\ref*{#1}}}
\newrefformat{rem}{\savehyperref{#1}{Remark~\ref*{#1}}}
\newrefformat{item}{\savehyperref{#1}{Item~\ref*{#1}}}
\newrefformat{step}{\savehyperref{#1}{step~\ref*{#1}}}
\newrefformat{conj}{\savehyperref{#1}{Conjecture~\ref*{#1}}}
\newrefformat{fact}{\savehyperref{#1}{Fact~\ref*{#1}}}
\newrefformat{prop}{\savehyperref{#1}{Proposition~\ref*{#1}}}
\newrefformat{prob}{\savehyperref{#1}{Problem~\ref*{#1}}}
\newrefformat{claim}{\savehyperref{#1}{Claim~\ref*{#1}}}
\newrefformat{relax}{\savehyperref{#1}{Relaxation~\ref*{#1}}}
\newrefformat{red}{\savehyperref{#1}{Reduction~\ref*{#1}}}
\newrefformat{part}{\savehyperref{#1}{Part~\ref*{#1}}}

\newcommand{\Sref}[1]{\hyperref[#1]{\S\ref*{#1}}}

\usepackage{nicefrac}
\let\nfrac=\nicefrac

\ifnum\fastmode=0
\ifnum\usemicrotype=1
\usepackage{microtype}
\fi
\fi

\ifnum\showauthornotes=1
\newcommand{\Authornote}[2]{{\sffamily\small\color{red}{[#1: #2]}}}
\newcommand{\Authornotecolored}[3]{{\sffamily\small\color{#1}{[#2: #3]}}}
\newcommand{\Authorcomment}[2]{{\sffamily\small\color{gray}{[#1: #2]}}}
\newcommand{\Authorstartcomment}[1]{\sffamily\small\color{gray}[#1: }

\newcommand{\Authorfnote}[2]{\footnote{\color{red}{#1: #2}}}
\newcommand{\Authorfixme}[1]{\Authornote{#1}{\textbf{??}}}
\newcommand{\Authormarginmark}[1]{\marginpar{\textcolor{red}{\fbox{\Large #1:!}}}}
\else
\newcommand{\Authornote}[2]{}
\newcommand{\Authornotecolored}[3]{}
\newcommand{\Authorcomment}[2]{}
\newcommand{\Authorstartcomment}[1]{}

\newcommand{\Authorfnote}[2]{}
\newcommand{\Authorfixme}[1]{}
\newcommand{\Authormarginmark}[1]{}
\fi

\newcommand{\Dnote}{\Authornote{D}}

\newcommand{\Jcomment}{\Authorcomment{J}}
\newcommand{\Dcomment}{\Authorcomment{D}}

\newcommand{\Pnote}{\Authornote{P}}

\newcommand{\Jnote}{\Authornote{J}}

\ifnum\showfixme=0

\fi

\usepackage{boxedminipage}

\newcommand{\Paren}[1]{\left(#1\right)}

\newcommand{\Bigparen}[1]{\Big(#1\Big)}
\newcommand{\Brac}[1]{\left[#1\right]}

\newcommand{\Abs}[1]{\left\lvert#1\right\rvert}

\newcommand{\card}[1]{\lvert#1\rvert}

\newcommand{\set}[1]{\{#1\}}

\newcommand{\norm}[1]{\lVert#1\rVert}

\newcommand{\iprod}[1]{\langle#1\rangle}

\newcommand{\Esymb}{\mathbb{E}}
\newcommand{\Psymb}{\mathbb{P}}

\DeclareMathOperator*{\E}{\Esymb}

\DeclareMathOperator*{\ProbOp}{\Psymb}

\renewcommand{\Pr}{\ProbOp}

\newcommand{\suchthat}{\;\middle\vert\;}

\newcommand{\textparen}[1]{\text{(#1)}}

\ifx\because\undefined
\newcommand{\because}[1]{\textparen{because #1}}
\else
\renewcommand{\because}[1]{\textparen{because #1}}
\fi

\newcommand{\bits}{\{0,1\}}

\newcommand{\vbig}{\vphantom{\bigoplus}}

\newcommand{\defeq}{\stackrel{\mathrm{def}}=}

\newcommand{\seteq}{\mathrel{\mathop:}=}

\newcommand{\from}{\colon}

\newcommand{\mper}{\,.}
\newcommand{\mcom}{\,,}

\newcommand\bdot\bullet

\ifx\mathds\undefined % use double stroke fonts if available
\newcommand{\Ind}{\mathbb I}
\else
\newcommand{\Ind}{\mathds 1}
\fi

\DeclareMathOperator{\Tr}{Tr}

\DeclareMathOperator{\opt}{opt}

\let\varprod\undefined

\DeclareMathOperator*{\varprod}{{\textstyle \prod}}

\newcommand{\N}{\mathbb N}
\newcommand{\R}{\mathbb R}

\newcommand{\Rnn}{\R_+}

\newcommand{\problemmacro}[1]{\texorpdfstring{\textup{\textsc{#1}}}{#1}\xspace}

\newcommand{\pnum}[1]{{\footnotesize #1}}

\newcommand{\ug}{\problemmacro{UG}}

\newcommand{\maxcut}{\problemmacro{max cut}}

\newcommand{\maxthreesat}{\problemmacro{max \pnum{3}-sat}}

\newcommand{\cC}{\mathcal C}
\newcommand{\cD}{\mathcal D}

\newcommand{\cF}{\mathcal F}
\newcommand{\cG}{\mathcal G}

\newcommand{\cL}{\mathcal L}
\newcommand{\cM}{\mathcal M}

\newcommand{\cS}{\mathcal S}
\newcommand{\cT}{\mathcal T}

\newcommand{\scrP}{\mathscr P}

\newcommand{\bbR}{\mathbb R}

\newcommand{\bbE}{\mathbb E}

\renewcommand{\leq}{\leqslant}
\renewcommand{\le}{\leqslant}
\renewcommand{\geq}{\geqslant}
\renewcommand{\ge}{\geqslant}

\ifnum\showdraftbox=1

\else

\fi

\let\epsilon=\varepsilon

\numberwithin{equation}{section}

\newcommand\MYcurrentlabel{xxx}

\newcommand{\MYstore}[2]{%
  \global\expandafter \def \csname MYMEMORY #1 \endcsname{#2}%
}

\newcommand{\MYload}[1]{%
  \csname MYMEMORY #1 \endcsname%
}

\newcommand{\MYnewlabel}[1]{%
  \renewcommand\MYcurrentlabel{#1}%
  \MYoldlabel{#1}%
}

\newcommand{\MYdummylabel}[1]{}

\newcommand{\torestate}[1]{%
  \let\MYoldlabel\label%
  \let\label\MYnewlabel%
  #1%
  \MYstore{\MYcurrentlabel}{#1}%
  \let\label\MYoldlabel%
}

\newcommand{\restatetheorem}[1]{%
  \let\MYoldlabel\label
  \let\label\MYdummylabel
  \begin{theorem*}[Restatement of \prettyref{#1}]
    \MYload{#1}
  \end{theorem*}
  \let\label\MYoldlabel
}

\newcommand{\restatelemma}[1]{%
  \let\MYoldlabel\label
  \let\label\MYdummylabel
  \begin{lemma*}[Restatement of \prettyref{#1}]
    \MYload{#1}
  \end{lemma*}
  \let\label\MYoldlabel
}

\newcommand{\restateprop}[1]{%
  \let\MYoldlabel\label
  \let\label\MYdummylabel
  \begin{proposition*}[Restatement of \prettyref{#1}]
    \MYload{#1}
  \end{proposition*}
  \let\label\MYoldlabel
}

\newcommand{\restatefact}[1]{%
  \let\MYoldlabel\label
  \let\label\MYdummylabel
  \begin{fact*}[Restatement of \prettyref{#1}]
    \MYload{#1}
  \end{fact*}
  \let\label\MYoldlabel
}

\newcommand{\restate}[1]{%
  \let\MYoldlabel\label
  \let\label\MYdummylabel
  \MYload{#1}
  \let\label\MYoldlabel
}

\newcommand{\addreferencesection}{
  \phantomsection
\ifnum\stocmode=0
  \addcontentsline{toc}{section}{References}
\else
  \addcontentsline{toc}{section}{References \hspace*{1in} --------- End of extended abstract ---------}
\fi

}

\newcommand{\sse}{\subseteq}

\newcommand{\e}{\epsilon}
\newcommand{\eps}{\epsilon}

\let\origparagraph\paragraph
\renewcommand{\paragraph}[1]{\vspace*{-7pt}\origparagraph{#1.}}

\allowdisplaybreaks

\usepackage{paralist}

\usepackage{comment}

\let\pref=\prettyref

\newcommand*{\vertiii}[1]{{\left\vert\kern-0.25ex\left\vert\kern-0.25ex\left\vert #1
    \right\vert\kern-0.25ex\right\vert\kern-0.25ex\right\vert}}

\newcommand*{\qe}[2]{S(#1\,\|\, #2)}
\newcommand*{\ce}[2]{D(#1\,\|\, #2)}

\renewcommand{\Ind}{\bm 1}
\ifnum\jamesmode=0
\fi
\let\1\Ind

\newcommand*{\hi}{\mathsf{high}}
\newcommand*{\low}{\mathsf{low}}

\newcommand*{\inst}{\Im}

\newcommand*{\gtwo}{\gamma_{2\to 2}^{\mathsf{psd}}}
\newcommand*{\psdrank}{\mathrm{rk}_{\mathsf{psd}}}
\newcommand*{\apsdrank}{\rho_{\mathsf{psd}}}
\newcommand*{\nnrank}{\mathrm{rk}_{+}}
\newcommand*{\sosdeg}{\deg_{\mathsf{sos}}}
\newcommand*{\juntadeg}{\deg_{\mathsf{J}}}

\newcommand{\polytope}[1]{\text{\textup{\textsc{#1}}}}

\newcommand*{\corr}{\polytope{corr}}
\newcommand*{\tsp}{\polytope{tsp}}
\newcommand*{\cut}{\polytope{cut}}
\newcommand*{\stab}{\polytope{stab}}

\newcommand*{\maxpi}{\problemmacro{max-$\Pi$}}
\newcommand*{\maxscrP}{\problemmacro{max-$\scrP$}}

\newcommand*{\sleq}{\preceq}
\newcommand*{\sgeq}{\succeq}

\DeclareMathOperator{\Id}{\mathrm{Id}}

\newcommand*{\uId}{\mathcal{U}}

\newcommand*{\annotaterel}[2]{\stackrel{\scriptscriptstyle\text{#1}}{#2}}

\title{{\bf Lower bounds on the size of semidefinite
\\ programming relaxations}
\ifnum\stocmode=1
\\ {\small\sf [extended abstract]}
\fi
}

\author{%
James R. Lee\thanks{University of Washington.}
\and Prasad Raghavendra\thanks{UC Berkeley.}
\and David Steurer\thanks{Cornell University.}
}

\date{}

\begin{document}

\maketitle
%\draftbox
\thispagestyle{empty}

\begin{abstract}

We introduce a method for proving lower bounds on the efficacy of semidefinite programming (SDP) relaxations
for combinatorial problems.  In particular, we show that
the cut, TSP, and stable set polytopes on $n$-vertex graphs
are not the linear image of the feasible region of any SDP (i.e., any spectrahedron) of dimension less than $2^{n^{\delta}}$,
for some constant $\delta > 0$.
This result yields the first super-polynomial lower bounds
on the semidefinite extension complexity of any explicit family of polytopes.

Our results follow from a general
technique for proving lower bounds on the positive semidefinite rank
of a matrix.  To this end, we establish a close connection between arbitrary SDPs and those arising from
the sum-of-squares SDP hierarchy.  For approximating maximum constraint satisfaction problems,
we prove that SDPs of  polynomial-size are equivalent in power to those
arising from degree-$O(1)$ sum-of-squares relaxations.
This result implies, for instance, that no family of polynomial-size SDP relaxations
can achieve better than a $7/8$-approximation for \textsc{max {\footnotesize 3}-sat}.

\end{abstract}

\clearpage

% tableofcontents added for better navigability of the document
\ifnum\showtableofcontents=1
{
\tableofcontents
\thispagestyle{empty}
 }
\fi

\clearpage

\section{Introduction}

Convex characterizations and relaxations of combinatorial problems have been
a consistent, powerful theme in the theory of algorithms since its inception.
Linear and semidefinite programming relaxations have been particularly useful
for the efficient computation of approximate solutions to NP-hard problems (see,
for instance, the books \cite{WS11,VaziraniBook01}).
In some sense, semidefinite programs (SDPs) can be seen as combining
the rich expressiveness of linear programs with the global geometric power of spectral methods.
For many fundamental combinatorial problems, this provides a genuinely new
structural and computational perspective \cite{GW95,KMS98,ARV09}.
Indeed, for an array of optimization problems, the best-known approximation algorithms can
only be achieved via SDP relaxations.

It has long been known that integrality gaps for linear programs (LPs) can often
lead to gadgets for NP-hardness of approximation reductions (see, e.g., \cite{DBLP:conf/stoc/LundY93,kcenter05,HK03}).
Furthermore, assuming the Unique Games Conjecture \cite{DBLP:conf/stoc/Khot02a},
it is known that integrality gaps for SDPs can
be translated directly into hardness of approximation results \cite{DBLP:conf/focs/KhotKMO04,Austrin10,Raghavendra08}.
All of this suggests that the computational model underlying LPs and SDPs is remarkably powerful.

Thus it is a natural (albeit ambitious) goal to
characterize the computational power of this model.  If $\mathrm{P} \neq \mathrm{NP}$, we do not expect
to find polynomial-size families of SDPs that yield arbitrarily good approximations
to NP-hard problems.  (See \cite{Rothvoss2013,DBLP:conf/esa/BrietDP13} for a discussion
of how this follows formally from the assumption $\mathrm{NP} \nsubseteq \mathrm{P}/\mathrm{poly}$.)

In the setting of linear programs (LPs), the search for a model and characterization
began in a remarkable work of Yannakakis \cite{Yannakakis91}.
He proved that the TSP and matching polytopes do not admit \emph{symmetric} linear programming
formulations of size $2^{o(n)}$, where $n$ is the number of vertices
in the underlying graph.  In the process, he laid the structural framework (in terms
of nonnegative factorizations) that would underlie all future work in the subject.
It took over 20 years before Fiorini, Massar, Pokutta, Tiwary, and de Wolf \cite{FMPTW12}
were able to remove the symmetry assumption and obtain a lower bound of $2^{\Omega(\sqrt{n})}$ on the
size of any LP formulation.  Soon afterward, Rothvo\ss\ \cite{DBLP:conf/stoc/Rothvoss14}
gave a lower bound of $2^{\Omega(n)}$ on the size of any LP formulation for the matching polytope
(and also TSP), completing Yannakakis' vision.

\Dnote{IMPORTANT: as Prasad noted it would be good to highlight the computational model proposed in the work of Yannakakis.
In contrast to other studied computational models it captures the best known algorithms for many basic problems.
We can also say that nonnegative rank and psd rank lower bounds can in principle separate P and NP (if I remember the discussion with Prasad correctly, this claim might even be buried somewhere in Yannakakis's paper.).
We see this result as support for the richness of the computational model.}

\Jnote{Point out similarity to ``dense model theorems,'' related issues.}

Despite the progress in understanding the power of LP formulations,
it remained a mystery whether there were similar strong lower bounds in the setting of SDPs.
An analogous positive semidefinite factorization framework
was provided in \cite{FMPTW12,GouveiaParriloThomas2011}.
Following the LP methods of \cite{DBLP:conf/focs/ChanLRS13},
the papers  \cite{DBLP:conf/coco/LeeRST14,FSP14} proved
exponential
lower bounds  on the size of \emph{symmetric} SDP formulations for NP-hard constraint satisfaction problems (CSPs).

In the present work, we prove strong lower bounds
on the size of general SDP formulations for the cut, TSP, and stable set polytopes.
Moreover, we show that polynomial-size SDP relaxations cannot achieve
arbitrarily good approximations for many NP-hard constraint satisfaction problems.
For instance, no polynomial-size family of relaxations
can achieve better than a $7/8$-approximation for \maxthreesat.  More
generally, we show that the low-degree sum-of-squares SDP
relaxations yield the best approximation among all polynomial-sized
families of relaxations for max-CSPs.

This is achieved by relating arbitrary SDP formulations to those coming from
the sum-of-squares SDP hierarchy\footnote{This hierarchy
is also frequently referred to as the Lasserre SDP hierarchy.} \cite{Lasserre01,Parrilo00,Shor87}, analogous to
our previous work with Chan relating LP formulations to the Sherali--Adams hierarchy \cite{DBLP:conf/focs/ChanLRS13}.
The SDP setting poses a number of significant challenges.
At a very high level, our approach can be summarized as follows:  Given an arbitrary SDP formulation
of small size, we use methods from quantum entropy maximization and online convex optimization (often going by
the name ``matrix multiplicative weights update'') to \emph{learn} an approximate low-degree sum-of-squares
formulation on a subset of the input variables.
In the next section, we present a formal overview of our results, and a discussion
of the connections to quantum information theory, real algebraic geometry, and proof complexity.

\paragraph{Organization}  The results of this work fall along two
broad themes, lower bounds on spectrahedral lifts of specific
polytopes and lower bounds on SDP relaxations for constraint
satisfaction problems.
Both sets of results are consequences
of a general method for proving lower bounds
on positive semidefinite rank.
For the convenience of the reader, we have
organized the two themes in two self-contained trajectories.  Thus, lower
bounds on spectrahedral lifts can be accessed via
\pref{sec:spectr-lifts-polyt}, \pref{sec:sos-degree},
\pref{sec:proofoverview}, \pref{sec:theproof} and
\pref{sec:corr}.  The lower bounds for constraint
satisfaction problems can be reached through \pref{sec:optimization},
\pref{sec:sos-degree}, \pref{sec:proofoverview},
\pref{sec:theproof} and \pref{sec:maxcsp}.

We also present general results on approximating density operators
against families of linear tests through quantum learning in
\pref{sec:learning}.  Finally, in \pref{sec:nnr}, we exhibit applications of our
techniques to non-negative rank; in particular, this is used
to give a simple, self-contained proof of a lower bound on the nonnegative rank of the unique disjointness
matrix.

\Pnote{update the reference to nnr and approximation sections better}

\subsection{Spectrahedral lifts of polytopes}
\label{sec:spectr-lifts-polyt}
Polytopes are an appealing and useful way to encode many combinatorial optimization problems.
For example, the traveling salesman problem on $n$ cities is equivalent to optimizing linear functions over the traveling salesman polytope, i.e., the convex hull of characteristic vectors $\Ind_{C}\in \{0,1\}^{\binom n 2} \subseteq \mathbb R^{\binom n 2}$ of $n$-vertex Hamiltonian cycles $C$ (viewed as edge sets).
If a polytope admits polynomial-size LP or SDP formulations, then we can optimize linear functions over the polytope in polynomial time (exactly for LP formulations and up to arbitrary accuracy in the case of SDP formulations).
Indeed, a large number of efficient, exact algorithms for combinatorial
optimization problems can be explained by small LP or SDP formulations of the underlying polytope.
(For approximation algorithms, the characterization in terms of compact formulations of polytopes is not as direct \cite{DBLP:conf/focs/BraunFPS12}.
In \pref{sec:semid-relax-constr}, we will give a direct characterization for approximation algorithms in terms of the original combinatorial problem.)
% For the perspective of this paper,
% the reader might think of the example $n = {N \choose 2}$ and where $P$ is the
% convex hull of all tours on the $N$-vertex complete graph, thought of as
% vectors in $\{0,1\}^{N \choose 2}$.
% (Optimizing linear functions over this polytope is equivalent to solving the traveling salesman problem.)
% There are many equivalent notions of
% what it means for $P$ to have a small SDP formulation.  Perhaps the most natural
% one is that there should be a low-dimensional SDP that captures $P$.
% We first discuss exact characterizations.
% In the next section, we address the approximation setting.

\Dnote{should decide on psd vs PSD}

\paragraph{Positive semidefinite lifts}
Fix a polytope $P \subseteq \mathbb R^n$ (e.g., the traveling salesman polytope described above).
We are interested in the question of whether there exists a low-dimensional SDP that captures $P$.
Let $\cS_+^k$ denote
the cone of symmetric, $k \times k$ positive semidefinite matrices embedded
naturally in $\mathbb R^{k \times k}$.  If there exists
an affine subspace $\cL \subseteq \mathbb R^{k \times k}$ and a linear
map $\pi : \mathbb R^{k \times k} \to \mathbb R^n$ such
that
\[
P = \pi(\cS_+^k \cap \cL)\,,
\]
one says that $P$ admits a \emph{positive-semidefinite (psd) lift of size $k$.} (This terminology is taken from \cite{FGPRT14}.)
We remark that the intersection of a PSD cone with an affine subspace
is often referred to as a \emph{spectrahedron.}

The point is that in order to optimize a linear function $\ell\from \R^n\to \R$ over the polytope $P$, it is enough to optimize the linear function $\ell\circ \pi\from \R^{k\times k}\to \R$ over the set $\cS^k_+ \cap \cL$ instead,
% one can optimize over $\cS_+^k \cap \cL$ instead
% of over $P$, i.e. for any vector $v \in \mathbb R^n$, we have
\begin{equation*}\label{eq:sdp-opt}
%\min_{x \in P}\, \langle v,x\rangle = \min_{y \in \cS_+^k \cap \cL} \langle v,\pi(y)\rangle\,.
\min_{x \in P}\, \ell(x) = \min_{y \in \cS_+^k \cap \cL}  \ell\circ \pi (y)\mper
\end{equation*}
Here, the optimization problem on the right is a semidefinite programming problem in $k$-by-$k$ matrices.
This idea also goes under the name of a \emph{semidefinite extended formulation} \cite{FMPTW12}.

\Dnote{TODO:
psd rank of polytopes is undefined at this point.
need to define something.}

\paragraph{The positive-semidefinite rank of explicit polytopes}
We define the \emph{positive-semidefinite (psd) rank} of a polytope $P$, denoted $\psdrank(P)$, to be the smallest number $k$ such that there exists a psd lift of size $k$.
(Our use of the word ``rank'' will make sense soon---see \pref{sec:sos-degree}.)
Bri\"et, Dadush, and Pokutta
\cite{DBLP:conf/esa/BrietDP13}
showed (via a counting argument)
that there exist 0/1 polytopes in $\R^n$ with exponential psd rank.
 %which is analogous to Shannon's result that there exist boolean functions with exponential circuit complexity.
% Streamlining---this is a bit of a weird reference.  For experts, it's completely unneeded,
% and for non-experts, it's confusing.
In this work, we prove the first super-polynomial lower bounds on the psd rank of explicit polytopes:
The correlation polytope  $\corr_n \subseteq \mathbb R^{n^2}$ is given by
\[
\corr_n = \mathrm{conv}\left(\{ xx^T : x \in \{0,1\}^n \}\right).
\]
In \pref{sec:psd-corr}, we show the following strong lower bound on its psd rank.
% Use of the word "exponential" is loaded as many people don't think of e^{n^{eps}} as exponential
%(The technical ingredients of this proof are in \pref{sec:theproof}, \pref{sec:single-test}, and \pref{sec:psd-corr}.)
\begin{theorem}\label{thm:corr-lb}
For every $n \geq 1$, we have
\[
\psdrank(\corr_n) \geq 2^{\Omega(n^{2/13})}\,.
\]
\end{theorem}

The importance of the correlation polytope $\corr_n$ lies in the fact that
a number of interesting polytopes from combinatorial optimization contain
a face that linearly projects to $\corr_n$ .  We first define a few different families
of polytopes and then recall their relation to $\corr_n$.

For $n \geq 1$,
let $K_n = ([n], {[n] \choose 2})$ be the complete graph on $n$ vertices.
For a set $S \subseteq [n]$, we use $\partial S \subseteq {[n] \choose 2}$ to denote the set of edges
with one endpoint in $S$ and the other in $\bar S$, and we use the notation $\1_{\partial S} \in \mathbb R^{n \choose 2}$
to denote the characteristic vector of $S$.  The \emph{cut polytope on $n$ vertices} is defined by
\[
\cut_n = \mathrm{conv} \left(\{ \1_{\partial S} : S \subseteq [n] \} \right).
\]
Similarly, if $\tau$ is a salesman tour of $K_n$ (i.e., a Hamiltonian cycle), we use $\1_{E(\tau)} \in \mathbb R^{n \choose 2}$ to
denote the corresponding indicator of the edges contained in $\tau$.  In that case, the \emph{TSP polytope} is given by
\[
\tsp_n = \mathrm{conv} \left(\{ \1_{E(\tau)} : \text{ $\tau$ is a Hamiltonian cycle} \} \right).
\]
Finally, consider an arbitrary $n$-vertex graph $G=([n], E)$.  We recall that a subset of vertices $S \subseteq [n]$
is an \emph{independent set} (also called a stable set) if there are no edges between vertices in $S$.  The \emph{stable set
polytope of $G$} is given by
\[
\stab_n(G) = \mathrm{conv} \left( \{\1_S \in \mathbb R^{n} : S \textrm{ is an independent set in $G$}\}\right).
\]
By results of \cite{DeSimone90} and \cite{FMPTW12} (see \pref{prop:polytope-relations}), \pref{thm:corr-lb} directly implies the following lower bounds on the  psd rank of the cut, TSP, and stable set polytopes.
\begin{corollary}
The following lower bounds hold for every $n \geq 1$,
\begin{align*}
& \psdrank(\cut_n) \geq 2^{\Omega(n^{2/13})}, \\
& \psdrank(\tsp_n) \geq 2^{\Omega(n^{1/13})}, \\
\max_{\textrm{$n$-vertex $G$}} & \psdrank(\stab_n(G))  \geq 2^{\Omega(n^{1/13})}.
\end{align*}
\end{corollary}

\subsection{Semidefinite relaxations and constraint satisfaction}
\label{sec:optimization}
\label{sec:semid-relax-constr}
% use of "we will"--we *are* doing it.
% better to use present tense
% also better not to say "in this section, ..."
% because... yeah, we are in this section
We now formalize a computational model of semidefinite relaxations for combinatorial optimization problems and prove strong lower bounds for it.
Unlike the polytope setting in the previous section, this model also allows us to capture approximation algorithms directly.
\Dnote{TODO: say that reader can skip this section if they are only interested in the polytope results}

% There is an another compelling motivation
% for studying psd rank coming from the optimization perspective.
Consider the following general optimization problem:%
\footnote{In this section, we restrict our discussion to optimization problems over the discrete cube.
Some of our results also apply to other problems, e.g., the traveling salesman problem (albeit only for exact algorithms).}
Given a low-degree function $f\from \{0,1\}^n\to \R$, represented by its coefficients as a multilinear polynomial,
\begin{equation}
\label{eq:hypercube-optimization}
  \begin{gathered}
    \text{maximize }f(x)\\
    \text{subject to } x\in \{0,1\}^n\mper
  \end{gathered}
\end{equation}
Many basic optimization problems are special cases of this general problem, corresponding to functions $f$ of a particular form:
For the problem of finding the maximum cut in a graph $G$ with $n$ vertices, the function $f$ outputs on input $x\in \{0,1\}^n$ the number of edges in $G$ that cross the bipartition represented by $x$, i.e., $f(x)$ is the number of  edges $\{i,j\}\in E(G)$ with $x_i\neq x_j$.
Similarly, for \maxthreesat on a 3CNF formula $\varphi$ with $n$ variables, $f(x)$ is the number of clauses in $\varphi$ satisfied by the assignment $x$.
More generally, for any $k$-ary boolean constraint satisfaction problem, the function $f$ counts the number of satisfied constraints.
Note that in these examples, the functions have at most degree $2$, degree $3$, and degree $k$, respectively.

Algorithms with provable guarantees for these kinds of problems---either implicitly or explicitly---certify upper bounds on the optimal value of instances.
(Indeed, for solving the decision version of these optimization problems, it is enough to provide such certificates.)
It turns out that the best-known algorithms for these problems are captured by certificates of a particularly simple form, namely sums of squares of low-degree polynomials.
The following upper bounds on problems of the form \pref{eq:hypercube-optimization} are equivalent to the relaxations obtained by the sum-of-squares SDP hierarchy \cite{Lasserre01,Parrilo00,Shor87}.
For $f : \{0,1\}^n \to \R$, we use $\deg(f)$ to denote the degree of the unique
multilinear real polynomial agreeing with $f$ on $\{0,1\}^n$; see \pref{sec:prelims}.

\newcommand{\sosub}{\overline{\textup{sos}}}

\Jnote{Need to be careful here about what degree means.}

\begin{definition}
  The \emph{degree-$d$ sum-of-squares upper bound} for a function $f \from \{0,1\}^n\to \R$, denoted $\sosub_d(f)$, is the smallest number $c\in \R$ such that $c-f$ is a sum of squares of functions of degree at most $d/2$, i.e., there exists functions $g_1,\ldots,g_t\from \{0,1\}^n\to \R$ for some $t\in \N$ with $\deg(g_1),\ldots,\deg(g_t)\le d/2$ such that the following identity between functions on the discrete cube holds:
  \begin{displaymath}
    c- f = g_1^2 + \cdots g_t ^2\mper
  \end{displaymath}
\end{definition}
Every function $f$ satisfies $\sosub_d(f)\ge \max (f)$ since sums of squares of real-valued functions are nonnegative pointwise.
For $d \geq 1$,
the problem of computing $\sosub_d(f)$ for a given function $f\from \{0,1\}^n\to \R$ (of degree at most $d$) is a semidefinite program of size
at most $1+ n^{d/2}$ (see, e.g., Theorem \ref{thm:gtwo-pseudodist}).%
\footnote{Moreover, for every $d\in \N$, there exists an $n^{O(d)}$-time algorithm based on the ellipsoid method that, given $f$, $c$, and $\e>0$, distinguishes between the cases $\sosub_d(f)\ge c$ and $\sosub_d(f)\le c-\e$ (assuming the binary encoding of  $f$, $c$, and $\e$ is bounded by $n^{O(d)}$).}
%\Jnote{Is the preceding obvious?  Or are you stating a non-trivial thing as fact...?
%Please explain (or remove it). I don't know why we would want to delve into issues
%of bit complexity in the introduction.}

The $\sosub_d$ upper bound is equivalent to the degree-$d$
sum-of-squares (also known as the level-$d/2$ Lasserre) SDP bound, and for small values
of $d$, these upper bounds underlie the best-known approximation algorithms for several
optimization problems.
For example, the Goemans--Williamson algorithm for \maxcut is based on the upper bound $\sosub_2$.
If we let $\alpha_{GW}\approx 0.878$ be the approximation ratio of
this algorithm, then every graph $G$ satisfies $\max (f_G) \ge
\alpha_{GW}\cdot \sosub_2(f_G)$ where the function $f_G$ measures cuts in $G$, i.e., $f_G(x)\coloneq \sum_{ij\in E(G)}(x_i-x_j)^2$.

A natural generalization of low-degree sum-of-squares certificates is
obtained by summing squares of functions in a low-dimensional
subspace.  We can formulate this generalization as a {\it non-uniform}
model of computation that captures general semidefinite programming
relaxations.  First, we make the following definition for a subspace
of functions.

\begin{definition}
\label{def:subspace-u-sos-ub}
  For a subspace $U$ of real-valued functions on $\{0,1\}^n$, the \emph{subspace-$U$ sum-of-squares upper bound} for a function $f\from \{0,1\}^n\to \R$, denoted $\sosub_U(f)$, is the smallest number $c\in \R$ such that $c-f$ is a sum of squares of functions from $U$, i.e., there exist a $g_1,\ldots,g_t\in U$ such that $c-f=g_1^2+\cdots + g_t^2$ is an identity of functions on $\{0,1\}^n$.
\end{definition}

Here, the subspace $U$ can be thought of as ``non-uniform advice'' to an
algorithm, where its dimension $\dim(U)$ is the size of advice.
%(The subspace of functions of degree at most $d/2$ has dimension $O(n^{d/2})$.)
In fact, if we fix this advice $U$, the problem of computing $\sosub_U(f)$ for a given function $f$ has a semidefinite programming formulation of size $\dim(U)$.%
\footnote{Under mild conditions on the subspace $U$, there exists a boolean circuit of size $(\dim U)^{O(1)}$ that given a constant-degree function $f$, and number $c\in \R$ and $\e>0$, distinguishes between the cases $\sosub_U(f)\ge c$ and $\sosub_U(f)\le c-\e$ (assuming the bit encoding length of $f$, $c$, and $\e$ is bounded by $(\dim U)^{O(1)}$.).
Note that since we will prove lower bounds against this model, the possibility that some subspaces might not correspond to small circuits does not weaken our results.
}
Moreover, it turns out that the generalization captures, in a certain precise
sense, all possible semidefinite programming relaxations for
\pref{eq:hypercube-optimization}.
The dimension of the subspace corresponds to the size of the SDP.
%(at least if we use the natural encoding of the objective function).
See \pref{sec:sdp-model} for a detailed discussion of the model.

In this work, we exhibit unconditional lower bounds in
this powerful non-uniform model of computation.  For example, we show that
the \maxthreesat problem cannot be approximated to a factor better
than $7/8$ using a polynomial-size family of SDP relaxations.  Formally, we
show the following lower bound for \maxthreesat.
\begin{theorem} \label{thm:maxthreesat}
  For every $s>7/8$, there exists a constant $\alpha>0$ such that for every $n\in \N$ and every linear subspace $U$ of functions $f\from \{0,1\}^n\to \R$ with
  \begin{displaymath}
    \dim U \le n^{\alpha \frac{\log n}{\log \log n}}\mcom
  \end{displaymath}
  there exists a \maxthreesat instance $\inst$ on $n$ variables such that $\max(\inst)\le s$ but $\sosub_U(\inst)=1$ (i.e., $U$ fails to achieve a factor-$s$ approximation for \maxthreesat).
\end{theorem}

Our main result is a characterization of
an optimal semidefinite
programming relaxation for the class of constraint satisfaction
problems among all families of SDP relaxations of similar size.
Roughly speaking, we show that the $O(1)$-degree sum-of-squares relaxations
are {\it optimal} among all polynomial-sized SDP relaxations for constraint
satisfaction problems.
Towards stating our main result, we define the class of constraint
satisfaction problems.  For the sake of clarity, we restrict
ourselves to boolean constraint satisfaction problems although the
results hold in greater generality.
%
%%% begin: definition and discussion of boolean CSPs

For a finite collection $\scrP$ of $k$-ary predicates $P\from \{0,1\}^k\to \{0,1\}$, we let \maxscrP denote the following optimization
problem:
An instance $\inst$ consists of boolean variables $X_1,\ldots,X_n$ and
a collection of $\scrP$-constraints $P_1(X)=1,\ldots,P_M(X)=1$ over
these variables.
A $\scrP$-constraint is a predicate $P_0\from \{0,1\}^n\to \{0,1\}$ such
that $P_0(X)=P(X_{i_1},\ldots,X_{i_k})$ for some $P\in\scrP$ and \emph{distinct} indices
$i_1,\ldots,i_k\in[n]$.
The objective is to find an assignment $x\in \{0,1\}^n$ that satisfies as
many of the constraints as possible, that is, which maximizes
\begin{displaymath}
  \inst(x)\defeq \frac{1}{M} \sum_{i=1}^M P_i(x)\mper
\end{displaymath}

We denote the optimal value of an assignment for $\inst$ as
$\opt(\inst)=\max_{x\in\{0,1\}^n}\inst(x)$.
For example, \maxcut corresponds to the case where $\scrP$ consists of
the binary inequality predicate.
For \maxthreesat, $\scrP$ contains all eight 3-literal disjunctions,
e.g., $X_1\vee \bar X_2\vee \bar X_3$.

%%% end: definition and discussion of boolean CSPs

Next, we discuss how to compare the quality of upper bound
certificates of the form $\sosub_U$.
Let $\Pi$ be a boolean CSP and let $\Pi_n$ be the restriction of $\Pi$
to instances with $n$ boolean variables.
As discussed before, the problem $\Pi_n$ could for example be \maxcut on graphs with $n$ vertices or \maxthreesat on formulas with $n$ variables.
We say that a subspace $U\subseteq \R^{\{0,1\}^n}$ achieves a \emph{$(c,s)$-approximation for $\Pi_n$} if every instance $\inst\in \Pi_n$ satisfies
\begin{equation}
\label{eq:c-s-approximation}
  \max(\inst)\le s \quad
  \Longrightarrow\quad \sosub_U(\inst)\le c \mper
\end{equation}
%where $\inst\from \{0,1\}^n\to \R$ is the objective function of instance $\inst$.
In other words, the upper bound $\sosub_U$ allows us to distinguish\footnotemark\ between the cases $\max(\inst)\le s$ and $\max(\inst)>c$ for all instances $\inst\in \Pi_n$.
\footnotetext{In order to distinguish between the cases
	$\max(\inst)\le s$ and $\max(\inst)>c$ it is enough to
	check whether $\inst$ satisfies $\sosub_U(\inst)\le c$.
In the case $\max(\inst)\le s$, we know that $\sosub_U(\inst)\le c$ by \pref{eq:c-s-approximation}.
On the other hand, in the case $\max(\inst)>c$, we know that $\sosub_U(\inst)>c$ because $\sosub_U(\inst)$ is always an upper bound on $\max(\inst)$.}

We prove the following theorem, which shows that for every
boolean CSP, the approximation guarantees obtained by the degree-$d$ sum-of-squares upper bound (also
known as the the level-$d/2$ Lasserre SDP relaxation) are optimal among all semidefinite programming relaxations of size at
most $n^{cd}$ for some universal constant $c > 0$.
\begin{theorem}\label{thm:csp-intro}
  Let $\Pi$ be boolean constraint satisfaction problem and let $\Pi_n$ be the set of instances of $\Pi$ on $n$ variables.
  Suppose that for some $m,d\in \N$, the subspace of degree-$d$ functions $f\from \{0,1\}^m\to\R$ fails to achieve a $(c,s)$-approximation for $\Pi_m$ (in the sense of \pref{eq:c-s-approximation}).
  Then there exists a number $\alpha=\alpha(\Pi_m,c,s)>0$ such that for all $n\in \N$, every subspace $U$ of functions $f\from \{0,1\}^n\to\R$ with $\dim U \le \alpha \cdot (n/\log n)^{d/4}$ fails to achieve a $(c,s)$-approximation for $\Pi_n$.
\end{theorem}

The theorem has several immediate concrete consequences for specific
boolean CSPs.  First, we know that $O(1)$-degree sos upper bounds do
not achieve an approximation ratio better than $7/8$ for \maxthreesat
\cite{Grigoriev2001, schoenebeck2008linear}, therefore \pref{thm:csp-intro}
implies that polynomial-size SDP relaxations for \maxthreesat cannot achieve an approximation ratio better than $7/8$.
In fact, a quantitatively stronger version of the above theorem yields \pref{thm:maxthreesat}.

Another concrete consequence of this theorem is that if there exists a polynomial-size family
of semidefinite programming relaxations for \maxcut that achieves an approximation ratio better than $\alpha_{\mathrm{GW}}$, then also a $O(1)$-degree sum-of-squares upper bound achieves such a ratio.
This assertion is especially significant in light of the notorious Unique Games
Conjecture one of whose implications is that it is NP-hard to
approximate \maxcut to a ratio strictly better than $\alpha_{GW}$.
%
%\pref{thm:csp-intro} implies that if the sum-of-squares SDP hierarchy does not give a polynomial-time algorithm to refute the Unique Games %Conjecture, then no polynomial-size SDP relaxation (in the sense of \pref{def:subspace-u-sos-ub}) can refute the Unique Games Conjecture.

\Jcomment{I comment the reference to the UGC.  It does not seem (to me) to follow from \pref{thm:csp-intro}.}

%Assuming the Unique Games Conjecture,
% \Jnote{Insert a reference to some UGC related material here---not Khot's paper, something that explains the max-cut connection.
% Did Luca write a survey?}
%In particular, assuming $\mathrm{P}\neq \mathrm {NP}$, the Unique Games Conjecture predicts, for all constants $\e>0$ and $d\ge 1$, the existence of an infinite family of graphs such that $\max (f_G) \le (\alpha_{GW}+\e)\cdot \sosub_d(f_G)$ for graphs $G$ in the family.
%(For small enough $\e>0$, such families of graphs are not known except for $d=1$.)

\Dnote{say for what problems the quantitative thing holds.}

%We are interested in the following question:
%
%\begin{quote}
%  \emph{Among all possible subspaces whose dimension grows at most polynomially with the input size $n$, which ones achieve the best possible approximation guarantees for \maxcut, \maxthreesat, and other constraint satisfaction problems?}
%\end{quote}

\Jnote{Should we delay this discussion to the results section?  There are many things that are undefined
here, and I wouldn't be able to understand them as a reader.  For instance, a ``problem $\Pi_n$''  What is a problem?}
%We say that a problem $\Pi$ is \emph{closed under adding variables} if for all $n,m\in \N$ with $n\ge m$ and for every instance $\inst\in \Pi$ on $m$ variables and every subset $S\subseteq [n]$ with $\lvert  S \rvert=m$, there exists an instance $\inst'\in \Pi$ on $n$ variables such that $f_{\inst'}(x)=\inst(x_S)$ for all $x\in\{0,1\}^n$, where $x_S$ is the projection of $x$ on the coordinates in $S$, i.e., $x_S=(x_{s_1},\ldots,x_{s_m})$ with $S=\{s_1,\ldots,s_m\}$ and $s_1<\cdots<s_m$.
%Examples of problems closed under adding variables are \maxcut, \maxthreesat, and, more generally, every boolean constraint satisfaction problem.

%(Here, we restrict ourselves to instances with $n$ variables because for a subspace $U\subseteq \R^{\{0,1\}^n}$, the upper bound $\sosub_U(f)$ %is defined only for function $f$ in $n$ Boolean variables.)

%For certain boolean constraint satisfaction problems, we obtain quantitatively stronger bounds.

\Dnote{explain csps here. say what it means to have polynomial degree integrality gaps. cite the known examples. state the result that we get matching quasi polynomial lower bounds.}
\Pnote{If we state the above theorem for general problems, how about deferring the definition of csps and just directly
skipping to ``A concrete consequence}

\subsection{Positive semidefinite rank and sum-of-squares degree}
\label{sec:sos-degree}

%\paragraph{Positive semidefinite factorizations}

In order to prove our results on spectrahedral lifts and semidefinite relaxations, the
factorization perspective will be essential.
In the  LP setting, the characterization of polyhedral lifts and LP relaxations in terms of nonnegative factorizations is a significant contribution of Yannakakis \cite{Yannakakis91}.
In the SDP setting, the analogous characterization is in terms of positive semidefinite factorizations
 \cite{FMPTW12,GouveiaParriloThomas2011}.

\begin{definition}[PSD rank]
Let $M \in \mathbb R_+^{p \times q}$ be a matrix with non-negative entries.  We say that $M$ admits
a \emph{rank-$r$ psd factorization} if there exist positive semidefinite matrices $\{A_i : i \in [p]\},
\{B_j : j \in [q]\} \subseteq \cS_+^r$ such that $M_{i,j}= \Tr(A_i B_j)$ for all $i \in [p], j \in [q]$.
We define $\psdrank(M)$ to be the smallest $r$ such that $M$ admits a rank-$r$ psd factorization.
We refer to this value as the \emph{psd rank of $M$.}
\end{definition}

\Dnote{try to explain here how psd factorizations are qualitatively different from nonnegative factorizations}
Nonnegative factorizations correspond to the special case that the matrices $\{A_i\}$ and $\{B_j\}$ are restricted to be diagonal.
A rank-$r$ nonnegative factorization can equivalently be viewed as a sum of $r$ rank-$1$ nonnegative factorizations (nonnegative rectangles).
Indeed, this viewpoint is crucial for all lower bounds on nonnegative factorization.
In contrast, rank-$r$ psd factorizations do not seem to admit a good characterization in terms of rank-$1$ psd factorizations.
This difference captures one of the main difficulties of proving psd rank lower bounds.

\paragraph{Main theorem}

Consider a nonnegative function $f\from \{0,1\}^n \to \R_+$ on the $n$-dimensional discrete cube.
We say that $f$ has a \emph{sum-of-squares (sos) certificate of degree $d$}
if there exist functions $g_1, \ldots, g_k\from \{0,1\}^n\to \R$
such that $\deg(g_1),\ldots,\deg(g_k) \leq d/2$, and $f(x)=\sum_{i=1}^k g_i(x)^2$
for all $x \in \{0,1\}^n$.
(Here, the $\deg(g)$ denotes the degree of $g$ as a multilinear polynomial.
We refer to Section \ref{sec:prelims} for the precise definition.)
We then define the \emph{sos degree of $f$}, denoted $\sosdeg(f)$, to be the minimal
$d$ such that $f$ has a degree-$d$ sos certificate.
%\Dnote{do we actually use the $\{-1,1\}^n$ domain anywhere?
%the following sentence is actually a bit confusing.
%maybe remove it?}
%For reasons of clarity, we will also want to define $\sosdeg(f)$ for functions $f : \{-1,1\}^n \to \Rnn$,
%but observe that since the two domains are the same under a linear change of variables,
%the sos degree does not change.

This notion is closely related\footnote{For the sake of simplicity,
we have only defined this notion for functions on the discrete cube.
In more general settings, one has to be a bit more careful;
we refer to \cite{GV02}.}
 to (a special case of) the Positivstellensatz
proof system of Grigoriev and Vorobjov \cite{GV02}.
We refer to the surveys \cite{LaurentSurvey09, DBLP:journals/corr/BarakS14} and the introduction of \cite{OZ12}
for a review of such proof systems and their relationship to semidefinite programming.

With this notion in place, we can now present a representative theorem
that embodies our approach.
For a point $x \in \bbR^n$ and a subset $S \subseteq [n]$, we
denote by $x_S \in \bbR^{|S|}$ the vector $x_S = (x_{i_1}, x_{i_2}, \ldots, x_{i_{|S|}})$
where $S = \{i_1, i_2, \ldots, i_{|S|}\}$ and $i_1 < i_2 < \cdots < i_{|S|}$.
For a function $f : \{0,1\}^m \to \Rnn$ and a number $n \geq m$,
we define the following central object:
The $\binom{n}{m} \times 2^n$-dimensional real matrix $M_n^f$ is given by
\begin{equation}\label{eq:restriction}
M^f_n(S,x) \defeq f(x_S)\,,
\end{equation}
where $S \subseteq [n]$ runs over all subsets of size $m$ and $x \in \{0,1\}^n$.

\begin{theorem}[Sum-of-squares degree vs. psd rank]
\label{thm:sos-vs-psd}
For every $m \geq 1$ and $f \from \{0,1\}^m \to \Rnn$,
there exists a constant $C > 0$ such that the following holds.
%Fix $m \geq 1$, and consider $f\from \{0,1\}^m\to \Rnn$.
For $n \geq 2m$,
if $d+2 = \sosdeg(f)$, then
\[
1+n^{1+ d/2} \geq \psdrank(M^f_n) \geq C \left(\frac{n}{\log n}\right)^{d/4}\,.
\]
%Here, $M_n^f$ is interpreted as an ${n \choose m}$-by-$ 2^n$ nonnegative matrix in the natural way.
\end{theorem}

\Pnote{ should $\sosdeg(f) > d$, so that we can find a degree $d$
pseudodistribution? there is an off by one error}

\begin{remark}
The reader might observe that the matrix in \eqref{eq:restriction} looks very similar to the
``pattern matrices'' defined by Sherstov \cite{Sherstov2011}.
This comparison is not unfounded;
some high-level aspects of our proof are quite similar to his.
Random restrictions are a powerful tool for analyzing functions over the discrete cube.
We refer to \cite[Ch. 4]{OD14} for a discussion of their utility in the
context of discrete Fourier analysis.  They were also an important tool
in the work \cite{DBLP:conf/focs/ChanLRS13} on lower bounds for LPs.
Accordingly, one hopes that our methods may have additional applications in
communication complexity.
This would not be surprising, as there is a model of quantum communication
that exactly captures psd rank (see \cite{FMPTW12}).
\end{remark}

\paragraph{Connection to spectrahedral lifts of polytopes}

The connection to psd lifts proceeds as follows.
Let $\{x_1, x_2, \ldots, x_{v}\} \subseteq P$ be such that $P = \mathrm{conv}(V)$ is the convex hull of $V$, and also fix
a representation
\[
P = \left\{ \vphantom{\bigoplus} x \in \mathbb R^n : \langle a_i, x\rangle \leq b_i \,\, \forall i \in [m]\right\}\,.
\]
The \emph{slack matrix $S$} associated to $P$ (and our chosen representation) is the matrix $S \in \Rnn^{m \times v}$ given by
$S_{i,j} = b_i - \langle a_i, x_j\rangle$.   It is not difficult to see that $\psdrank(S)$ does not depend on
the choice of representation.  It turns out that the psd rank of $S$ is precisely the minimum size
of a psd lift of $P$.

\begin{proposition}[\cite{FMPTW12,GouveiaParriloThomas2011}]
\label{prop:gpt}
For every $n,k \geq 1$, every polytope $P \subseteq \mathbb R^n$ and every slack matrix $S$ associated to $P$, it
holds that $\psdrank(S) \leq k$ if and only if $P$ admits a psd lift of size $k$.
\end{proposition}

Thus our goal in this paper becomes one of proving lower bounds on psd rank.
With this notation, we have a precise way to characterize the lack of previous progress:  Before this work,
there was no reasonable method available to prove lower bounds on the psd rank of explicit matrices.
The characterization of \pref{prop:gpt}
explains our abuse of notation in
\pref{thm:corr-lb}, writing $\psdrank(P)$ to denote the psd rank
of any slack matrix associated to a polytope $P$.

\pref{thm:sos-vs-psd} is already enough to show that $\psdrank(\corr_n)$ must grow faster than
any polynomial in $n$, as we will argue momentarily.  In \pref{sec:proof-of-main-theorem}, we present a more refined version
(using a more robust version of sos degree) that will allow us to achieve a lower bound of the form $\psdrank(\corr_n) \geq 2^{\Omega(n^{\delta})}$
for some $\delta > 0$.

%Using a stronger version of \pref{thm:sos-vs-psd}, we prove the following lower bound
%in Section \ref{sec:psd-corr}.
%
\Jnote{Should mention here or earlier some known lower bounds.  There are definitely some of the polynomial regime.
Also good to mention that there were psd rank upper bounds for some hard slack matrices for LP.}

\Dnote{would be good to add here the intuitive view on psd rank.
right now it is sort of mysterious how it is connected to semidefinite programming.
instead it should be obvious.
maybe the view in terms of sum of squares of functions in a subspace is useful.
i think it is also important to say that one should think of the rows of the matrix as instances and the columns as solutions.}

%In other words, if $d=\sosdeg(f)$, then $\psdrank(M_n)$ grows asymptotically (in the exponent) at least like $n^{d/4}$.
%In fact, one can achieve a lower bound that grows like
%This theorem is already enough to show that the correlation polytope does not admit lifts of polynomial size.

\Jnote{Mention Nisan-Wigderson generator as also being of the same form.}

\Dnote{the following discussion about the relationship between $M^f_n$ and (slack matrices) of polytopes is a bit more complicated than necessary.
if $f$ is quadratic then the matrix corresponds to the correlation polytope.
}
\Dnote{IMPORTANT: this polytope view seems quite restrictive (since it doesn't distinguish what kind of function $f$, e.g., max cut, knapsack, or something else).
it would be good to explain these things better.
otherwise people wonder why we have this general theorem about $M^f_n$ when it only gives lower bounds for the correlation polytope.}

Given \pref{thm:sos-vs-psd}, in order to prove a lower bound on $\psdrank(\corr_n)$, we should find,
for every $d \geq 1$, a number $m$ and a function $f : \{0,1\}^m \to \Rnn$ such that $\sosdeg(f) \geq d$
and such that $M_n^f$ is a submatrix of some slack matrix associated to $\corr_n$.
To this end, it helps to observe the following (we recall the proof in Section \ref{sec:corr}).

\begin{proposition}
\label{prop:quadric-intro}
If $f \from \{0,1\}^m \to \Rnn$ is a nonnegative quadratic function over $\{0,1\}^m$,
then for any $n \geq m$, the matrix $M_n^f$ is a submatrix of some slack matrix associated to $\corr_n$.
\end{proposition}

\begin{comment}
The correlation polytope is also referred to as the \emph{Boolean quadric polytope} \cite{Padberg89} for
the following reason.
Suppose that $\ell_{y,a}(x) = a-\sum_{i=1}^n y_i x_i$ is a linear
form in the variables $\{x_i\}$.
If, for some $a, b \in \mathbb R$ and $y \in \mathbb R^n$, we have
\[
\ell_{y,a}(x)^2 \geq b\qquad \forall x \in \{0,1\}^n,
\]
then there is a corresponding valid linear inequality for $\corr_n$ in the sense that,
for $x \in \{0,1\}^n$,
\begin{equation}\label{eq:corr-ineq}
b \leq \ell_{y,a}(x)^2 = (a - y^T x)^2 = a^2 + (y^T x)^2 - 2 a y^T x = a^2 + \langle yy^T, xx^T\rangle - 2 a \langle \mathrm{diag}(y), xx^T\rangle\,,
\end{equation}
where we have used the notation $\mathrm{diag}(y)$ for the diagonal matrix with the entries of $y$ on the diagonal, and
$\langle A,B\rangle = \Tr(A^T B)$ denotes the Frobenius inner product.
Note our use of the fact that $\langle \mathrm{diag}(y), \mathrm{diag}(x)\rangle = \langle \mathrm{diag}(y), xx^T\rangle$
which is true for $x \in \{0,1\}^n$.

If \eqref{eq:corr-ineq} holds for all $x \in \{0,1\}^n$, then by convexity we have, for every $Z \in \corr_n$,
\[
a^2 + \langle yy^T, Z\rangle - 2 a \langle \mathrm{diag}(y), Z\rangle - b\geq 0\,.
\]

Thus if $f : \{0,1\}^m \to \Rnn$ is given by $f(x)=(a-\sum_{i=1}^m x_i)^2-b$, then $M_n^f$ is a submatrix
of some slack matrix associated to $\corr_n$ since
\[
M_n^f(S,x) = f(x_S) = \ell_{\1_S,a}(x)^2-b\,,
\]
where $\1_S$ is the indicator vector with ones in the positions corresponding to $S$.
\end{comment}

Given the preceding proposition, the following result of Grigoriev on the Knapsack tautologies
completes our quest for a lower bound.
\begin{theorem}[\cite{DBLP:journals/cc/Grigoriev01}]
\label{thm:grigor}
For every odd integer $m \geq 1$,
the function $f : \{0,1\}^m \to \Rnn$ given by
\begin{equation}\label{eq:knapsack}
f(x) = \left(\frac{m}{2}-\sum_{i=1}^m x_i\right)^2 - \frac14
\end{equation}
has $\sosdeg(f) \geq m+1$.
\end{theorem}

Note that since $m/2$ is not an integer, \eqref{eq:knapsack}
is nonnegative for all $x \in \{0,1\}^m$.
It turns out that in order to prove stronger lower bounds for $\corr_n$,
we will require a lower bound on the \emph{approximate} sos degree of $f$.
Thus the Knapsack tautologies \eqref{eq:knapsack} will be studied carefully
in \pref{sec:psd-corr}.
In \pref{sec:overview}, we discuss the proof of \pref{thm:sos-vs-psd}
in some detail.  Then in \pref{sec:theproof}, we present a quantitatively
stronger theorem and its proof.
%The next section is concerned with a formal statement of our results.

\paragraph{Connection to semidefinite relaxations and constraint satisfaction}

Fix now numbers $k,n \geq 1$ and a boolean CSP $\Pi$.  Fix a
pair of constants $0 \leq s \leq c \leq 1$.  Suppose our goal is to
show a lower bound on the size of SDP relaxations that yield a
$(c,s)$-approximation on instances of size $n$.
It turns out that this task reduces to proving a lower bound on the
positive semidefinite rank of an explicit matrix $M$ indexed by problem instances and problem solutions (points on the discrete cube in our case).
\begin{proposition}
\label{prop:intro-csp2}
  For any boolean CSP $\Pi_n$ and any constants $0 \leq s<c \leq 1$, let $U$ be a subspace of minimal dimension that achieves a $(c,s)$-approximation for $\Pi_n$.
  Denote the set of instances
  \[\Pi_n^{\le s}=\{\inst \mid \max(\inst)\le s\}\,.\]
  Let $M\from \Pi_n^{\le s}\times \{0,1\}^n\to \R$ denote the matrix
  \begin{displaymath}
    M(\inst, x ) = c - \inst(x)\mper
  \end{displaymath}
  Then, $ \psdrank(M)^2\geq \dim(U) \geq \psdrank(M)$.
\end{proposition}

Before describing the proof of this proposition, observe that together with our main theorem the proposition implies \pref{thm:csp-intro} (optimality of degree-$d$ sum-of-squares for approximating boolean CSPs):
If $\inst_0$ is a \maxscrP instance on $m$ variables with $\max(\inst_0)\le s$ and $\sosub_d(\inst_0)>c$, then $f=c-\inst_0$ has sos degree larger than $d$.
Our main theorem gives a lower bound on the psd rank of the matrix $M^f_n$.
Since this matrix is a submatrix of the matrix in \pref{prop:intro-csp2}, our psd rank lower bound implies a lower bound on the minimum dimension of a subspace achieving a $(c,s)$-approximation for $\maxscrP_n$.

%In the remainder of the section, we will present a proof of the above
%theorem.
\begin{proof}[Proof of \pref{prop:intro-csp2}]
Set $r = \dim(U)$.  Fix a basis $q_1,\ldots, q_r : \bits^n \to \R$
for the subspace $U$.  Define the function $Q : \bits^n \to \cS_r^+$
by setting $(Q(x))_{ij} \seteq q_i(x) q_j(x)$ for all $i,j \in [r]$.
Notice that for any $q \in U$, we can write $q = \sum_{i=1}^r \lambda_i q_i$
and thus $q(x)^2 = \Tr(\Lambda Q(x))$ where $\Lambda \in \cS_r^+$ is
defined by $\Lambda_{ij} \seteq \lambda_i \lambda_j$.

Since $U$ achieves a $(c,s)$-approximation for $\Pi_n$, for every
instance $\inst \in \Pi_n^{\leq s}$ we will have
$\sosub_{U}(\inst) \leq c$.  By definition of
$\sosub_{U}(\inst)$, this implies that $c - \inst = \sum_i
g_i^2$ for $g_i \in U$.  By expressing each $g_i^2$ as $g_i^2 =
\Tr(\Lambda_i Q(x))$
for some $\Lambda_i \in \cS_r^+$ we get,
$$ M(\inst, x) = c - \inst(x) = \sum_i \Tr(\Lambda_i Q(x)) = \Tr\left(
\Lambda_{\inst} Q(x)\right) \mper$$
This yields an explicit psd factorization of $M$ certifying that
$\psdrank(M) \leq \dim(U)$.

Conversely, by definition of $\psdrank(M)$ there exists positive
semidefinite matrices $\{\Lambda_{\inst} : \inst \in \Pi_n^{\leq s}\}$, $\{ Q(x) : x \in \bits^n \} \sse \cS_r^+$ such that $M(\inst,x) =
\Tr(\Lambda_{\inst} Q(x))$.  Denote by $R(x) \seteq Q(x)^{1/2}$
the positive semidefinite square root, and consider
the subspace $\tilde{U}
\seteq \mathrm{span}\{ (R(x))_{ij} \} \sse \R^{\bits^n}$.  Clearly,
the dimension of $\tilde{U}$ is at most $\psdrank(M)^2$.  Further, for each instance
$\inst \in \Pi_n^{\leq s}$, we can write
$$ c - \inst = M_{\inst, x} = \Tr(\Lambda_{\inst} Q(x)) =
\Tr(\Lambda_{\inst} R(x)^2) = \norm{\sqrt{\Lambda_\inst} R(x)}_F^2
\mper$$
Observe that $\norm{\sqrt{\Lambda_{\inst}}R(x)}_F^2$ is a sum of
squares of functions from the subspace $\tilde{U}$\footnote{Here,$\|\cdot\|_F$ denotes
the Frobenius norm.}.  Therefore we have
$\sosub_{\tilde{U}}(\inst) \leq c$, showing that $\sosub_{\tilde{U}}$ yields a
$(c,s)$-approximation to $\Pi_n$.  Since $U$ is the minimal subspace
yielding a $(c,s)$-approximation, we have $\dim(U) \leq
\dim(\tilde{U}) \leq \psdrank(M)^2$.
\end{proof}

\section{Proof overview and setup} \label{sec:proofoverview}

\subsection{Preliminaries}
\label{sec:prelims}

We write $[n] \defeq \{1,2,\ldots,n\}$ for $n \in \mathbb N$.
We will often use the notation $\E_x$ to denote a uniform averaging operator
where $x$ assumes values over a finite set.  For instance, if $x \in \{-1,1\}^n$, then
$\E_x f(x) = 2^{-n} \sum_{x \in \{-1,1\}^n} f(x)$.  The domain of the operator
should always be readily apparent from context.
We also use asymptotic notation:  For two expressions $A$ and $B$, we write $A \leq O(B)$ if
there exists a universal constant $C$ such that $A \leq C \cdot B$.
We also sometimes write $A \lesssim B$ to denote $A \leq O(B)$.
The notation $A \geq \Omega(B)$ similarly denotes $B \lesssim A$, and
the notations $A = \Theta(B)$ and $A \asymp B$ are both used to denote
the conjunction of $A \lesssim B$ and $B \lesssim A$.
For a real number $x > 0$, we use $\log x$ to denote the natural logarithm of $x$.

\medskip
\noindent
{\bf Inner product spaces and norms.}
Let $H$ denote a finite-dimensional vector space over $\R$ equipped with an inner product $\langle \cdot,\cdot \rangle$
and the induced Euclidean norm $|\cdot|$.
All vector spaces we consider here will be of this kind.
We use $\cM(H)$ to denote the set of self-adjoint linear operators on $H$, and
$\cD(H)\subseteq \cM(H)$ for the set of density operators on $H$, i.e.,
those positive semidefinite operators with trace one.
We will use the standard Loewner ordering $\succeq$ on $\cM(H)$.

If $H$ comes equipped with a canonical (ordered) orthonormal basis (as will always be the case throughout), we represent linear operators on $H$ by matrices with rows and columns indexed by the basis elements.
In this case, $\cM(H)$ consists of symmetric matrices and $\cD(H)$ consists of symmetric, positive semidefinite matrices whose diagonal entries summing to one.
If $A \in \cM(H)$ is positive semidefinite, we use $A^{1/2}$ to denote the positive semidefinite square root of $A$.

Given a linear operator $A : H \to H$, we define the operator, trace, and Frobenius norms, respectively:
\begin{align*}
\lVert A\rVert &= \max_{x \neq 0} \frac{|Ax|}{|x|} \\
\lVert A\rVert_* &= \Tr(\sqrt{A^T A}) \\
\lVert A\rVert_F &= \sqrt{\Tr(A^T A)}\,.
\end{align*}
Recall $\Tr(A^T B) \leq \|A\| \cdot \|B\|_{*}$ and the Cauchy-Schwarz inequality $\Tr(A^T B) \leq \|A\|_F \|B\|_F$.
For a matrix $M$, we use $\|M\|_{\infty}$ to denote the maximum absolute value of an entry in $M$.

\medskip
\noindent
{\bf Fourier analysis and degree over the discrete cube.}
We use $L^2(\{-1,1\}^n)$ to denote the Hilbert space of real-valued functions $f : \{-1,1\}^n \to \mathbb R$.
This space is equipped with the natural inner product under the uniform measure: $\langle f,g\rangle = \E_x f(x) g(x)$.
We recall the Fourier basis:  For $S \subseteq [n]$, one has $\chi_S(x) = \prod_{i \in S} x_i$.
The functions $\{\chi_S : S \subseteq [n]\}$ form an orthonormal basis for $L^2(\{-1,1\}^n)$.
We can decompose $f$ in the Fourier basis as $f = \sum_{S \subseteq [n]} \hat f(S) \chi_S$.

We will use $\deg(f)$ to denote the degree of $f$ as a multi-linear polynomial on the discrete cube:
$\deg(f) \defeq \max \{ |S| : \hat f(S) \neq 0 \}$.
Note that by identifying $\{0,1\}$ and $\{-1,1\}$, we can define
$\deg(f)$ for functions $f : \{0,1\}^n \to \mathbb R$ as well.
(Since the change of domains is given by the linear map $x \mapsto 2x-1$,
the degree of polynomial representations do not change.)

If we are given a matrix-valued function $M : \{-1,1\}^n \to \mathbb R^{p \times q}$,
we can decompose $M$ as $M = \sum_{S \subseteq [n]} \hat M_S \chi_S$ where $(\hat M_S)_{ij} = \widehat{M_{ij}(S)}$,
and $\deg(M) = \max \{ \deg(M_{ij}) : i \in [p], j \in [q] \}$.
We refer to the book \cite{OD14} for additional
background on boolean Fourier analysis.

\medskip
\noindent
{\bf Quantum information theory.}
The {\em von-Neumann entropy} of a density operator $X$ is denoted $S(X) = -\Tr (X \log X)$.
For two density operators $X$ and $Y$ over the same vector space, the {\em quantum relative entropy} of $X$ with respect to $Y$ is the
quantity $\qe{X}{Y} = \Tr (X\cdot (\log X - \log Y))$.
Here, the operator function $\log$ is defined on positive operators as $\log X = - \sum_{k=0}^\infty \tfrac 1 k (\Id-X)^k$.
In general, for a function $g : I \to \mathbb R$ analytic on an open interval $I \subseteq \mathbb R$ and a symmetric
operator $X \in \cM(H)$, we define $g(X)$ via its Taylor series, with the understanding
that the spectrum of $X$ should lie in $I$.
Finally, we will often use the notation $\uId = \frac{\Id}{\Tr(\Id)}$ to
denote the uniform density matrix (i.e., the maximally mixed state), where the dimension
of the identity matrix $\Id$ is clear from context.
We refer to \cite{Wilde2013} for a detailed account of quantum information theory.

\subsection{Factorizations, quantum learning, and pseudo-densities}
\label{sec:overview}

First, we recall the setup of the main theorem in the paper.
Fix $m \geq 1$, a function $f : \{0,1\}^m \to \Rnn$ and let $d+2 = \sosdeg(f)$.  We define the
matrix $M_n^f$ as in \eqref{eq:restriction}.  Our goal is to show a
lower bound on the positive semidefinite rank of the matrix $M_n^f$.

Suppose we had a psd factorization
\begin{equation}\label{eq:mn-factor}
M_n^f(S,x) = \Tr(P(S) Q(x))
\end{equation}
witnessing $\psdrank(M) \leq r$.
First, we observe that lower bound on $\sosdeg(f)$ already precludes
certain \emph{low degree} psd factorizations.  More precisely, if
$R(x) = Q(x)^{1/2}$ then $\deg(R)$ is constrained to be at least $d/2$.
For the sake of contradiction, let us suppose $\deg(R) < d/2$.
For any row $M_n^f(S, \cdot)$ of the matrix $M_n^f$  we will have,
$$ f(x_S) = \Tr(P(S) R(x)^2) = \| \sqrt{P(S)} R(x)\|_F^2 \mper
$$
This contradicts $\sosdeg(f) = d+2$ since $\|\sqrt{P(S)} R(x)\|_F^2$
is a sum of squares of a polynomials of degree less than $d/2$.
\Pnote{}

\paragraph{Pseudo-densities and low degree psd factorizations}
By appealing to convex duality, it is possible to construct a {\it certificate} that the
matrix $M_n^f$ does not admit {\it low degree} factorizations.
The certificate consists of a linear functional that separates $M_n^f$
from the convex hull of matrices that admit low degree psd
factorizations.  Formally, if we define the convex set $\cC_{d}$ of
non-negative matrices as,
\begin{displaymath}
	\cC_d \defeq \left\{ N: \binom{n}{m} \times \bits^n \to \R
	\suchthat N(S,x) =
\Tr(P(S) R(x)^2), P(S) \succeq 0, \deg(R(x)) < d/2   \right\}
\end{displaymath}
then we will construct a linear functional $L$ on $\binom{n}{m} \times 2^n$ matrices such that
\begin{equation}\label{eq:separator}
  L(M_n^f) < 0, \textrm{ but } L(N) \geq 0 \textrm{ for all } N \in
  \cC_d\,.
\end{equation}
The linear functional is precisely the one derived from what
we refer to as a {\em pseudo-density}.

A {\em degree-$d$ pseudo-density} is a mapping $D\from \{0,1\}^m\to
\R$ such that $\E_x D(x) = 1$ and $\E_x D(x) g(x)^2 \ge 0$ for all
functions $g\from \{0,1\}^m\to \R$ with $\deg(g) \leq
d/2$.\footnote{Note that a degree-$d$ pseudo-density does not
necessarily have degree $d$ as a function on the discrete cube.}
Observe that for any probability distribution over $\{0,1\}^n$, its density function relative to the uniform distribution on $\{0,1\}^n$
satisfies the conditions of a degree-$d$ pseudo-density for every $d\in\N$.
One has the following characterization:
\begin{equation}\label{eq:sos-char}
\sosdeg(f) = \min \left\{ d \geq 0 : \E_x D(x) f(x) \geq 0 \textrm{ for every degree-$d$ pseudo-density $D$} \right\}.
\end{equation}
In other words, the sos degree of a function is larger than $d$ if and only if there exists a degree-$d$ pseudo-density
$D$ such that $\E_x D(x) f(x) < 0$.
To verify this, note that if $\sosdeg(f) > d$, then $f$ is not in the closed, convex
cone generated by the squares
of polynomials of degree at most $d/2$.  Now the required pseudo-density $D$ corresponds exactly
to (the normal vector of) a hyperplane separating $f$ from this cone.

Of course, if $D$ is an actual density (with respect to the uniform measure on $\{0,1\}^m$),
then $\E_x D(x) f(x)$ is precisely the expectation of $f$ under $D$.
For a pseudo-density $D$,
the corresponding linear functional $f \mapsto \E_x D(x) f(x)$ is referred
to as a {\em pseudo-expectation} in previous papers (see, e.g., \cite{DBLP:conf/stoc/BarakBHKSZ12,DBLP:conf/focs/ChanLRS13}),
and the map $D$ is called a {\em pseudo-distribution} in \cite{BarakKS14}.
Over finite domains, these notions are interchangeable.  We use the language
of densities here in anticipation of future applications to infinite domains
and non-uniform background measures (in the context of nonnegative rank,
this occurs already in \pref{sec:corr-lp}).
\Dnote{}
\Pnote{}

\medskip

Now fix a degree-$d$ pseudo-density $D$ with
$\E_x D(x) f(x) < 0.$
We define the following linear functional on matrices $N\from \binom{[n]}{m}\times \{0,1\}^n\to \R$:
\begin{equation}\label{eq:ldn}
  L_D(N)
  \defeq \E_{|S|=m} \E_x D(x_S) \cdot N(S,x)
  \mper
\end{equation}
Consider a matrix $N \in \cC_d$ which admits a low degree psd
factorization given by $N(S,x) = \Tr(P(S) R(x)^2)$.
Then since $D$ is a degree-$d$ pseudo-density,
we would have
\[ \label{eq:ldlowdegree}
L_D(N)  =
\E_{|S|=m} \E_x D(x_S) \Tr\left(P(S) R(x)^2\right) = \E_{|S|=m} \E_x
D(x_S) \|\sqrt{P(S)} R(x)\|_F^2  \geq 0 \mper
\]
However, since $D$ is negatively correlated with $f$,
\begin{equation} \label{eq:ldonf}
L_D(M_n^f) = \E_{|S|=m} \E_x D(x_S) \cdot M_n^f(S,x) = \E_{|S|=m} \E_x
D(x_S) \cdot f(x_S)  < -\epsilon \mper
\end{equation}
for some $\epsilon > 0$.

The core of our psd rank lower bound is to show that the linear functional $L_D$ in fact
separates the matrix $M_n^f$ from {\em all} low
rank psd factorizations, thereby certifying a lower bound on
$\psdrank(M_n^f)$.
Roughly speaking, the idea is to approximate an arbitrary psd
factorization using low degree factorizations with respect to the
linear functional $L_D$, and then appeal to the lower bound
\eqref{eq:ldlowdegree} for low
degree factorizations.

Formally, for a number $r\ge 1$, consider the following set $\cC_r$ of nonnegative matrices,
\begin{displaymath}
	\cC_r \defeq \left\{\vbig N \in \Rnn^{\binom{n}{m} \times
\bits^n} : \psdrank(N)\le r\cdot \lVert  N \rVert_1, ~ \lVert N\rVert _\infty \le 1\right\}\mper
\end{displaymath}
Here, $\lVert  N \rVert_1$ is the average of the entries of $N$ and $\lVert N\rVert _\infty$ is the maximum entry of $N$.
In the rest of the section, we will present an argument that unless
$r$ is very large, every matrix $N \in \cC_r$ satisfies $L_D(N) \geq
-\epsilon$.  Since $L_D(M_n^f) < -\epsilon$, this implies that
the linear functional $L_D$ separates $M_n^f$ from the
convex hull of $\cC_r$, thereby certifying a lower bound on
$\psdrank(M_n^f)$.

Fix a matrix $N \in \cC_r$.  It is
instructive to have the situation $\norm{N}_1,\norm{N}_\infty = \Theta(1)$ in mind for the rest of this
outline.  By definition of $\cC_r$, the matrix $N$ admits a psd
factorization of rank $O(r)$.  In light of the above discussion, our goal is to approximate the
matrix $N$ by a low degree factorization with respect to the
functional $L_D$.
A low degree approximation for $N$ is constructed in two steps.

\medskip
\noindent
{\bf Well-behaved factorizations.}
The first step involves obtaining a nicer
factorization of $N$.
Toward this end, we define the quantity
$$\gamma_r(M) \defeq \sup \left\{ \max_{i,j} \|A_i\| \cdot \|B_j\|_{*}
: N_{ij} = \Tr(A_i B_j), A_i,B_j\in \cS_+^r \quad \forall i \in [p], j
\in [q]\right\}\mcom$$
associated with a matrix $M \in \R_+^{p \times q}$.  The following lemma is proved by Bri\"et,
Dadush, and Pokutta \cite{DBLP:conf/esa/BrietDP13} (see also the discussion in
  \cite{FGPRT14}).

\begin{lemma}[Factorization rescaling]\label{lem:scaling}
For every nonnegative matrix $M$ with $\psdrank(M)\le r$, the following holds:
\[\gamma_r(M) \leq r^2\, \|M\|_{\infty}\,.\]
\end{lemma}

Applying the above lemma to the matrix $N$ at hand, we get a psd
factorization $N(S,x) = \Tr(P(S)Q(x))$ wherein $\norm{P(S)}$ and
$\norm{Q(x)}_*$ are bounded polynomially in $r$.
This analytic control on the factorization will be important for controlling error bounds,
but also---in a more subtle way---for the next step.

\medskip
\noindent
{\bf Learning a low-degree quantum approximation.}
The next step of the argument exploits the following phenomenon
concerning quantum learning.
Fix a $k \geq 1$ and consider a matrix-valued function $Q : \{0,1\}^n \to \cS_+^k$ such that
$\E_x \Tr(Q(x))=1$.
We will try to approximate $Q$ by a simpler mapping
with respect to a certain class of test functionals $\Lambda : \{0,1\}^n \to \cS_+^k$.
If $\tilde Q$ is the approximator, we would like that
\begin{equation}
\label{eq:test}
\left|\bbE_x \Tr\left(\vphantom{\bigoplus} \Lambda(x) (Q(x)-\tilde Q(x))\right)\right| \leq \e
\end{equation}
for some parameter $\e > 0$.  (In this case, $\tilde Q$ and $Q$ are indistinguishable to the test $\Lambda$
up to accuracy $\e$.)
One can set this up as a quantum learning problem in the following way.
We define the density matrix $U_Q = \E_x (e_x e_x^T \otimes Q(x))$
and the PSD matrix $V_{\Lambda} = \sum_x (e_x e_x^T \otimes \Lambda(x))$.\footnote{In
the quantum information literature, these are sometimes called
QC states for ``quantum/classical.''}
Note that $\bbE_x \Tr(\Lambda(x) Q(x)) = \Tr(V_{\Lambda} U_Q)$.

Now, if $\cT$ is a family of test functionals, then a canonical way of finding a ``simple''
approximation to $U_Q$ that satisfies all the tests is via the following
maximum-entropy (convex) optimization problem:
\begin{equation}\label{eq:max-ent}
\max \left\{ S(\tilde U) : \Tr(\tilde U)=1, \tilde U \succeq 0, |\Tr(V_{\Lambda} (U_Q-\tilde U))| \leq \e \,\, \forall \Lambda \in \cT \right\},
\end{equation}
where we recall that $S(\cdot)$ denotes the quantum entropy functional.
Moreover, one can attempt to solve this optimization by some form of projected sub-gradient descent.
Interpretations of this algorithm go by many names, notably the ``matrix multiplicative weights update method''
and ``mirror descent'' with quantum entropy as the regularizer; see, e.g., \cite{MR702836,MR1967286,MR2249846,DBLP:conf/stoc/AroraK07,MR2904324}
and the recent survey \cite{Bubeck2014}.

In our setting, we are not directly concerned with efficiency, but instead simplicity
of the approximator.  A key phenomenon is that when the class of tests $\cT$ is simple,
the approximator inherits this simplicity.  Moreover, one can tailor the nature of the
approximator by choosing the sub-gradient steps wisely.
In \pref{sec:gener-appr-dens} (\pref{thm:general-operator-approx}), we prove a generalization of the following approximation theorem.
(Recall that $\uId = \Id/\Tr(\Id)$ is the uniform density matrix.)

\begin{theorem}[Approximation by a low-degree square]
\label{thm:low-deg-approx}
Let $\kappa \geq 1$ and $\omega > 0$ be given.
Define
\[
\cT_{\kappa,\omega} = \left\{ \Lambda : \{0,1\}^n \to \cS_+^k : \deg(\Lambda) \leq \kappa, \|\Lambda(x)\| \leq \omega\,\,\forall x \in \{0,1\}^n \right\}.
\]
For any $Q : \{0,1\}^n \to \cS_+^k$ with $\E_x \Tr(Q(x))=1$, there is a matrix-valued function $R : \{0,1\}^n \to \cS_+^k$
with $\E_x \Tr(R(x)^2)=1$ satisfying
\begin{equation}\label{eq:deg-bound}
\frac{\deg(R)}{\kappa} \lesssim \left(1+\qe{U_Q}{\uId}\right) \frac{\omega}{\e}\,,
\end{equation}
and for all tests $\Lambda \in \cT_{\kappa,\omega}$,
\[
\left|\E_x \Tr\left(\vphantom{\bigoplus} \Lambda(x) (Q(x)-R(x)^2)\right)\right| \leq \e\,.
\]
\end{theorem}

In other words, the learning algorithm produces a hypothesis with error at most $\e$ for all the tests in $\cT_{\kappa,\omega}$;
moreover, the hypothesis is the square of a polynomial whose degree is not much larger than that of the tests.
The value $\omega$ corresponds to the ubiquitous ``width'' parameter and, as in most applications of the
multiplicative weights method, bounding $\omega$ will be centrally important.
The reader should also take note of the appearance of the relative entropy in the degree bound \eqref{eq:deg-bound}.
It will turn out that low psd rank factorizations will give us functions $Q : \{0,1\}^n \to \cS_+^k$ with high entropy
(and thus small relative entropy with respect to the uniform state);
this is actually a direct consequence of the factorization rescaling in \pref{lem:scaling}.

Notice that the separating functional $L_D$ induces a test of degree
at most $m$.  Therefore, if one takes for granted, as claimed above,
that $\qe{U_Q}{\frac{\Id}{\Tr(\Id)}}$ is small
when $Q$ comes from a low psd rank factorization, then
\pref{thm:low-deg-approx} suggests that we might think of $Q(x)$
as being a low-degree square. %

\medskip
\noindent

\medskip
\noindent
{\bf Proof sketch for \pref{thm:sos-vs-psd}.}
We have all the ingredients to sketch a proof of \pref{thm:sos-vs-psd}.
First, suppose that $\sosdeg(f) > d$ so that by \eqref{eq:sos-char}, there
exists a degree-$d$ pseudo-density $D$ with $\E_x f(x) D(x) < -\e\|f\|_{\infty}$ for some $\e > 0$.
(We do not specify any quantitative bound on $\e$ at the moment,
but we write it this way to indicate how one can get improved bounds
under stronger assumptions.)

Then from the definition of $M_n^f$, we have
\begin{equation}\label{eq:imp}
L_D(M^f_n) < - \e\|M^f_n\|_{\infty}\,.
\end{equation}

On the other hand, we will prove the following theorem in \pref{sec:proof-of-main-theorem}.

\begin{theorem}
  \label{thm:intro-funct}
  For every $m,d \geq 1$, every $\e \in (0,1]$, and
  every degree-$d$ pseudo-density $D : \{0,1\}^m \to \bbR$,
  there exists a number $\alpha > 0$ such that
  for every $n \geq 2m$ and every nonnegative matrix $N\from \binom{[n]}{m}\times \{0,1\}^n\to \R$ satisfying
  \begin{gather*}
    \lVert N\rVert_\infty  \le 1\mcom \textrm{ and}\\
    \tfrac 1 {\lVert  N \rVert_1}\psdrank(N)^2\le \alpha (n/\log n)^{d/2}\,,
  \end{gather*}
we have  $L_D(N)\ge -\e$.
\end{theorem}
Now if we consider the normalized matrix $N=M^f_n/ \lVert M^f_n\rVert_\infty$, we see that it satisfies the first premise $\lVert  N \rVert_\infty \le 1$
but violates the conclusion of the theorem (because of \eqref{eq:imp}).
Therefore we know that the second premise is violated, which gives the lower bound
\[
\psdrank(N)^2 > \alpha (n/\log n)^{d/2} \cdot \|N\|_1 = \alpha (n/\log n)^{d/2} \E_x f(x)\,.
\]

\begin{comment}
  \text{\Dnote{}}

  Use \eqref{eq:imp} and \pref{thm:intro-funct} together to
  conclude that $\apsdrank(M_n^f) = \gtwo(M_n^f)/\|M_n^f\|_{\infty}
  \gtrsim n^{d/2}/\eta$.  Hence \pref{lem:scaling} yields
  \[
  \psdrank(M_n^f) \geq \sqrt{\apsdrank(M_n^f)} \geq \frac{n^{d/4}}{\sqrt{\eta}}\,.
  \]
\end{comment}

Since this achieves our goal, we are left to explain why \pref{thm:intro-funct} should be true,
at least when we apply it with $N=M_n^f/\|M_n^f\|_{\infty}$.
If we apply $L_D$ to the right-hand side of \eqref{eq:mn-factor}---our presumed factorization for $M^f_n$---we arrive at the expression
\begin{equation}\label{eq:Qexp}
L_D(M_n^f) = \E_x \Tr\left(\E_{|S|=m} D(x_S) P(S) Q(x)\right).
\end{equation}
We can view this as a test on $Q$ in the sense of \pref{thm:low-deg-approx}.
Since $\deg(D) \leq m$ (because $D$ is only a function of $m$ variables), this is a low-degree test.
\pref{thm:low-deg-approx} then suggests that we can replace $Q$ by a low-degree approximator
$R^2$, while losing only $\e$ in the ``accuracy'' of the test.

Since the approximation property implies that $Q(x)$ and $R(x)^2$
should perform similarly under the test (up to the ``accuracy'' $\e$), we
would conclude that $L_D(M_n^f) \geq -\e$, yielding
the conclusion of \pref{thm:intro-funct}.

\medskip
\noindent
{\bf Random restriction and degree reduction.}
The one serious issue with the preceding argument is that our supposition
is far too strong:  One cannot expect to have $\deg(R) \leq d/2$.  Indeed,
the guarantee of \pref{thm:low-deg-approx}
tells us that the approximator $R(x)$ has degree at most $K \cdot \deg(D)$ for some (possibly large) number $K$
(which itself depends on many parameters).
To overcome this problem, we use another crucial property of our functional \eqref{eq:ldn}:  It is an expectation
over small sets $S \subseteq [n]$.  If we randomly choose such a subset with $|S|=m \ll n$ and randomly choose
values $y_{\bar S}$ for the variables in $\bar S$, we expect that the
resulting (partially evaluated) polynomial $R(x_S, x_{\bar
S})|_{x_{\bar S}=y_{\bar S}}$ will satisfy $\deg(R(x_S,
x_{\bar S})|_{x_{\bar S}=y_{\bar S}}) \ll \deg(R)$.
(Strictly speaking, this will only be true in an approximate sense.)

It is precisely this degree reduction property of random restriction that saves the preceding sketch.
In the next sections, we perform a more delicate quantitative analysis capable of achieving much stronger lower bounds.
The norm $\|D\|_{\infty}$ of the pseudo-density will play a central role in this study.
Thus in \pref{sec:psd-corr}, we
show that Grigoriev's proof of \pref{thm:grigor}
can be carefully recast in the language of pseudo-densities
such that the resulting pseudo-density has small norm.

\ifnum\stocmode=1
\medskip
\fi

\ifnum\stocmode=1

\newpage

\subsection*{Acknowledgments}

This work was supported, in large part, by NSF grant CCF-1407779.
A significant fraction of the project was completed during a long-term
visit of the authors to the Simons Institute for the Theory of Computing (Berkeley)
for the program on Algorithmic Spectral Graph Theory.
The authors would also like to thank Paul Beame, Siu-On Chan, Daniel Dadush, Troy Lee,
Sebastian Pokutta, Pablo Parrilo, Mohit Singh, Ola Svensson, Thomas Rothvo\ss, and Rekha Thomas for valuable discussions
and comments.

\addreferencesection
\bibliographystyle{amsalpha}
\bibliography{bib/mr,bib/dblp,bib/scholar,bib/lpsize}

\newpage

\thispagestyle{empty}
\vspace*{\fill}
\begingroup
\centering
\begin{center} {\Huge \textbf{APPENDIX: Full paper continued.}}
\end{center}

\endgroup
\vspace*{\fill}

\newpage
\fi

\section{PSD rank and sum-of-squares degree}
\label{sec:theproof}

We now move to proving the main technical theorems of the paper
along the lines of the informal overview presented in \pref{sec:overview}.

\subsection{Analysis of the separating functional}
\label{sec:proof-of-main-theorem}

\Jnote{}
\Pnote{}

Recall that for a pseudo-density $D\from \{0,1\}^m\to \R$ and $n\ge 1$, we define a linear functional $L_D$ on matrices $N\from \binom{[n]}{m}\times \{0,1\}^n\to \R$
\begin{displaymath}
  L_D(N) \defeq \E_{x} \E_S D(x_S) N(S,x)
  \mcom
\end{displaymath}
where the expectation over $S$ is a uniform average over all $S \subseteq [n]$ with $|S|=m$ (as will be the case
throughout this section).
We will use the notation $\lVert  N \rVert_\infty =\max_{S,x} N(S,x)$ and $\lVert  N \rVert_1 = \E_{S,x}N(S,x)$.

\medskip

We prove the following quantitative version of Theorem \ref{thm:psdrank-separation}.
As discussed in \pref{sec:overview}, this theorem implies a lower bound on the $\psdrank(M_n^f)$ in terms of $\sosdeg(f)$.
This implication will be proved formally in \pref{sec:mainthms}.

\begin{theorem}[Strengthening of \pref{thm:intro-funct}]
  \label{thm:gamma-separation}
  \label{thm:psdrank-separation}
  For every $m,d \geq 1$, every $\e \in (0,1]$, and every degree-$d$ pseudo-density $D\from \{0,1\}^m \to \bbR$, there exists a number $\alpha > 0$ such that
  whenever $n \geq 2m$ and a nonnegative matrix $N\from \binom{[n]}{m}\times \{0,1\}^n\to \R$ satisfies
  \begin{gather*}
    \lVert N\rVert_\infty  \le 1\mcom\\
    \tfrac 1 {\lVert  N \rVert_1}\psdrank(N)^2\le  \alpha (n/\log n)^{d/2}\,,
  \end{gather*}
  we have $L_D(N)\ge -\e$.
Moreover, this holds for
\begin{displaymath}
    \alpha = \left (\frac{C\e}{d m^2 \lVert  D \rVert_\infty} \right  )^{d/2} \left(\frac{\e}{\lVert  D \rVert_\infty}\right)^3 \mcom
  \end{displaymath}
  where $C > 0$ is a universal constant.
\end{theorem}

The proof of this theorem consists of two parts.
First, we observe that if $D$ is a degree-$d$ pseudo-density,
 then $L_D(N)$ is nonnegative for all matrices $N$ that admit a factorization in terms of squares of low-degree polynomials, i.e., a factorization $N(S,x)=\Tr(A_S^2 B_x^2)$ for symmetric matrics $\{A_S\}$ and $\{B_x\}$ such that the function $x\mapsto B_x$ has degree at most $d/2$ over $\{0,1\}^n$.\footnote{For the convenience of the reader,
 we recall that the degree of the matrix-valued function $x \mapsto B_x$ is defined
 as the maximum degree of the functions $x \mapsto (B_x)_{ij}$ where $i,j$
 range over the indices of $B_x$.}
Indeed, consider such a factorization.
Then,
\[
L_D(N) = \E_{x} \E_S D(x_S) \Tr(A_S^2 B_x^2)
= \E_{S} \left(\E_x D(x_S) \|A_S B_x\|_F^2 \right) \geq 0\,,
\]
where the inequality used the fact that $D$ is a degree-$d$ pseudo-density (hence $\E_x D(x) g(x)^2 \geq 0$ whenever $\deg(g) \leq d/2$).

As explained in \pref{sec:overview}, this guarantee is not sufficient for us.
The following theorem (proved in Section \ref{sec:degree-reduction})
allows us to analyze $L_D$ even when the degree of the map $x \mapsto B_x$ is much larger than $d/2$.
(For $m\le n^{o(1)}$, it will be the case
that the linear functional $L_D$ is approximately nonnegative on $N$ even when $x\mapsto B_x$ has degree up to $n^{o(1)}$).

\begin{theorem}[Degree reduction]
\label{thm:overview-degree-reduction}
  Consider postive numbers $n \geq 1$ and $d,k,m \leq n$.
  Let $D\from \{0,1\}^m\to R$ be a degree-$d$ pseudo-density.
  Let $N' \from \binom{[n]}{m}\times \{0,1\}^n\to \R$ be a matrix that
  admits a factorization $N'(S,x)=\Tr A_S^2 B_x^2$ for symmetric
  matrices $\{A_S\}$ and $\{B_x\}$ such that the matrix-valued
  function $x\mapsto B_x$ has degree at most $\ell$.
  Then,
  \begin{displaymath}
    L_D(N') \succsim -  \left  (\frac{\ell m}{n-m}\right)^{d/4}  \lVert  D \rVert_{\infty} \cdot \Paren{\Bigparen{ \max_{S}\lVert  A_S^2  \rVert\cdot \E_x \Tr(B_x^2)}\cdot\Bigparen{ \E_x \E_S N'(S,x)}}^{1/2}\mper
  \end{displaymath}
\end{theorem}

With this theorem in place, our goal is to approximate every matrix
$N$ with low psd rank by a matrix $N'$ that satisfies the premise of
\pref{thm:overview-degree-reduction} for a reasonable value of $\ell$.
Here, our notion of approximation is fairly weak.
We only require $L_D(N)\ge L_D(N')-\e$ for sufficiently small $\e>0$.
As a preliminary step, the following general theorem about psd factorizations allows us to assume that the factorization for $N$ is
appropriately scaled.
Recall that $\uId = \Id/\Tr(\Id)$ is the uniform density matrix.
%

\begin{comment}
\begin{theorem}[Normalized psd factorization]
\label{thm:overview-psd-factorization}
  For every nonnegative matrix $M=(M_{i,j})$ with $\max M_{i,j}\le 1$ and every $\eta>0$, there exist psd matrices $\{P_i\}$ and $\{Q_j\}$ with the following properties:
  \begin{enumerate}
  \item $ M_{i,j} \le \Tr P_i Q_j \le  M_{i,j} + \eta $,
    \label{item:psd-factor-approx}
  \item $\E_i P_i = \Id $ and $\lVert  P_i \rVert\precsim \psdrank(M)^2/\eta $ for all $i\in [p]$,
    \label{item:psd-factor-p}
  \item $ Q_j \sleq \psdrank(M)^{O(1)}\tfrac1{\Tr \Id}\Id$ for all $j\in[q]$.
\label{item:psd-factor-q}
  \end{enumerate}
\end{theorem}
\end{comment}

\Jnote{}

\begin{theorem}[psd factorization scaling]
\label{thm:overview-psd-factorization}
  For every nonnegative matrix $M \in \mathbb R^{p \times q}$ and every $\eta \in (0,1]$, there exist psd matrices $\{P_i\}_{i \in [p]}$ and $\{Q_j\}_{j \in [q]}$ with the following properties:
  \begin{enumerate}
  \item $ M_{i,j} \leq  \Tr(P_i Q_j) \le M_{i,j}+\eta \|M\|_{\infty}$,
  \label{item:psd-factor-approx}
  \item $\frac{1}{p} \sum_{i=1}^p P_i = \Id,$
    \label{item:psd-factor-p1}
  \item $\lVert  P_i \rVert\leq
    2 \psdrank(M)^2/\eta$ for all $i\in [p]$,
    \label{item:psd-factor-p2}
    \item $Q_j \sleq \norm{M}_\infty (\eta+\psdrank(M)^{2}) \psdrank(M) \,\uId$ for all $j\in[q]$.
  \label{item:psd-factor-q}
  \end{enumerate}
\end{theorem}

\begin{proof}
Let $r=\psdrank(M)$.
By \pref{lem:scaling}, we have $\gamma \seteq \gamma_r(M) \leq r^2 \lVert  M \rVert_\infty$.
Fix a factorization $M_{i,j} = \Tr(A_i B_j)$ such that $\max_{i,j}
\norm{A_i} \cdot \norm{B_j}_* = \norm{M}_\infty \cdot
\psdrank(M)^2$ and $A_i,B_j\in \R^{r\times r}$.
  By an appropriate
normalization, we may assume $A_i, B_j \succeq 0$
and $\|A_i\| \leq \gamma$, $\|B_j\|_* \leq 1$.
To construct psd matrices $\{P_i\}$ and $\{Q_j\}$ with the desired
properties, make the following definitions:
\begin{align*}
A &= \eta \|M\|_{\infty} \Id + \frac{1}{p} \sum_{i=1}^p A_i \\
P_i &= A^{-1/2} (\eta \|M\|_{\infty} \Id+A_i) A^{-1/2}\\
Q_j &= A^{1/2} B_j A^{1/2}\,.
\end{align*}
Note that \pref{item:psd-factor-p1} holds by construction.
Also observe that \[\Tr(P_i Q_j) = M_{i,j} + \eta\|M\|_{\infty}  \Tr(A^{-1} A^{1/2}
	B_j
	A^{1/2}) = M_{i,j} + \eta \|M\|_{\infty} \Tr(B_j)\,,\]
verifying \pref{item:psd-factor-approx}.  Finally, we have the inequalities for all $i \in [p],j \in [q]$,
\begin{align*}
\|P_i\| & \leq \frac{1}{\eta \|M\|_{\infty}} (\eta \|M\|_{\infty} + \|A_i\|) \leq
1 + \frac{\gamma}{\eta \|M\|_{\infty}} \mcom \\
\|Q_j\|_* & \leq \|A\| \cdot \|B_j\|_* \leq \gamma+\eta \|M\|_{\infty}
 \mper
\end{align*}
The first inequality verifies \pref{item:psd-factor-p2} since $r \geq 1$ and $\eta \leq 1$.
The last inequality implies that \[Q_j \sleq (\gamma+\eta \|M\|_{\infty}) r \frac{\Id}{\Tr(\Id)} \sleq r \|M\|_{\infty} (\eta + r^2) \frac{\Id}{\Tr(\Id)}\] for all $j \in [q]$,
verifying \pref{item:psd-factor-q}.
\end{proof}

Consider a matrix of the form $N : {[n] \choose m} \times \{0,1\}^n \to \Rnn$ with $\|N\|_{\infty} \leq 1$ and
let $\e > 0$ be given.
Apply \pref{thm:overview-psd-factorization} with a value $\eta \in (0,1]$
to be chosen later to obtain a factorization
\[
N(S,x) = \Tr(P_S Q_x)
\]
satisfying the conclusions of the theorem.

We can view the matrix-valued function $x\mapsto Q_x$ as a density matrix $Q=\tfrac 1{\E_x (\Tr Q_x)}\E_x (e_x e_x^T \otimes Q_x)$.
(The first $n$ bits in this density matrix are ``classical'' and their marginal distribution has density $x\mapsto \Tr Q_x$.
If we condition $Q$ on an assignment $x\in \{0,1\}^n$ to the first $n$ bits, the resulting quantum state is $\tfrac 1{\Tr Q_x}Q_x$).
Here, the normalization factor $\tau = \E _x \Tr Q_x$ for the density matrix $Q$ satisfies
\begin{equation}
  \label{eq:tau-normalize}
  \tau = \E_x \Tr Q_x
  \annotaterel{\text{(Thm \ref{thm:overview-psd-factorization}(2))}}= \E_x \E_S \Tr P_S Q_x
  \begin{cases}
    & \annotaterel{\text{(Thm \ref{thm:overview-psd-factorization}(1)}}\ge \E_S \E_x
    N(S,x) =\lVert N \rVert_1\mcom\\
    & \annotaterel{\text{(Thm \ref{thm:overview-psd-factorization}(1))}}\le \E_S \E_x
    N(S,x) +\eta  \le 1 + \eta\mcom
  \end{cases}
\end{equation}
where the last inequality has used $\|N\|_{\infty} \leq 1$.

From \pref{thm:overview-psd-factorization}(4),
the density matrix $Q$ satisfies
\[Q \sleq \frac{(\eta + \psdrank(N)^2) \psdrank(N)}{\tau} \uId\,.\]
Therefore,
\begin{equation}
  \qe{Q}{\uId}\precsim \log (\eta \psdrank(N)/\tau) \le \log \bigl(\psdrank(N)/\lVert  N \rVert_1\bigr)\mper
  \label{eq:entropy-deficit}
\end{equation}
\pref{thm:overview-psd-factorization}(1)
allows us to lower bound $L_D(N)$ in terms of the matrix $(S,x)\mapsto \Tr(P_S Q_x)$ and value $\|D\|_{\infty}$:
\begin{align}
  L_D(N)& = \E_{x} \E_S D(x_S) N(S,x)\notag\\
  & \ge \E_x \E_S D(x_S) \cdot  \Tr (P_S Q_x)   - \eta \lVert  D \rVert_\infty  \quad \text{(by \pref{thm:overview-psd-factorization}(1))}\notag\\
  & = \tau \cdot \Tr (F Q)  -\eta \lVert  D \rVert_\infty
  \mcom
  \label{eq:sos-2}
\end{align}
where $F$ is the symmetric matrix
\begin{equation}
  \label{eq:define-f}
  F=\sum_{x \in \{0,1\}^n}  e_x e_x^T \otimes F_x\text{ with }F_x=\E_S D(x_S) P_S.
\end{equation}

\pref{thm:overview-psd-factorization}(2) allows us to upper bound the spectral norm of $F$ by
\begin{equation}
\label{eq:sos-1}
  \lVert  F \rVert \le \max_x \left\lVert  \E_S D(x_S) P_S \right\rVert \le \lVert  D \rVert_\infty \cdot \left \lVert  \E_S P_S \right \rVert \annotaterel{(Thm \ref{thm:overview-psd-factorization}(2))}= \lVert  D \rVert_\infty.
\end{equation}

The next theorem allows us to lower bound $\Tr(F Q)$
by replacing $Q$ with a simpler density matrix that is a low-degree polynomial in $F$.
(See \pref{thm:operator-approx}, where a slightly more general version is proved.)

\begin{theorem}[Density matrix approximation]
  \label{thm:operator-approx-intro}
  Let $H$ be some finite-dimensional real inner-product space.
  Let $F \in \cM(H)$ be a symmetric matrix and let $Q \in \cD(H)$ be a density matrix.
  Then, for every $\e>0$, there exists a degree-$k$ univariate
  polynomial $p$ with  $k \precsim   (1+\qe{Q}{\uId})\cdot \lVert F\rVert/\e$ such that the density matrix $\tilde Q =  \tfrac 1{\Tr p(F)^2}p(F)^2$ satisfies
  \begin{equation}
    \label{eq:operator-approx-intro}
    \Tr \Paren{\vbig F \tilde Q} \le \Tr (F Q) + \e\mper
  \end{equation}
\end{theorem}

Apply \pref{thm:operator-approx-intro} to the density matrix $Q$ and the symmetric matrix $F$ defined above
with the value $\e$ (which is already fixed).  Let $p$ be the resulting degree-$k$ polynomial, with $k$
satisfying the bounds of the theorem.

Since the function $x \mapsto F_x$ has $\deg(F) \leq \deg(D) \leq m$ (since $D : \{0,1\}^m \to \R$),
the degree of the map $x \mapsto \tilde Q_x = \tfrac 1{\E_x \Tr(p(F_x)^2)}p(F_x)^2$ is at most $\deg (p) \cdot m = k \cdot m$.
Applying \pref{thm:overview-degree-reduction} to the matrix given by
$N'(S,x) = \Tr(P_S \cdot p(F_x)^2),$
we can give a lower bound:
\begin{align}
\left(\E _x\Tr p(F_x)^2\right)\cdot \Tr \Paren{F\cdot \tilde Q}
  &= \E_{S} \E_x D(x_S) \Tr \Paren{P_S \cdot p(F_x)^2} \nonumber \\
& \gtrsim
   - \left(\frac{k m^2}{n-m}\right)^{d/4} \cdot \lVert  D \rVert_\infty \cdot
\Paren{\Bigparen{ \max_{S}\,\lVert  P_S  \rVert\cdot \E_x \Tr(p(F_x)^2) }\cdot\Bigparen{ \E_x \E_S N'(S,x)}}^{1/2} . \nonumber  %
\end{align}
Using the fact that $\E_x \E_S N'(S,x) = \E_S\E_x \Tr P_S \cdot p(F_x)^2=\E_x \Tr p(F_x)^2$
from \pref{thm:overview-psd-factorization}(2)
and the fact that $ \max_S \lVert  P_S \rVert \leq 2 \psdrank(N)^2/\eta$
from \pref{thm:overview-psd-factorization}(3) yields
\begin{equation}
\label{eq:sos-3}
 \Tr \Paren{F\cdot \tilde Q} \gtrsim  - \left(\frac{k m^2}{n-m}\right)^{d/4} \frac{\|D\|_{\infty}}{\sqrt{\eta}} \psdrank(N)\,.
\end{equation}

We have now assembled all components of the proof of \pref{thm:psdrank-separation}.
\begin{proof}[Proof of \pref{thm:psdrank-separation}]
We lower bound the linear functional $L_D(N)$ by
\begin{align*}
  L_D (N)
  & \annotaterel{\pref{eq:sos-2}}\ge
  \tau \cdot \Tr (F Q) - \eta \lVert  D \rVert_\infty
  \\
  & \annotaterel{\pref{eq:operator-approx-intro}}\ge
  \tau \cdot  \Paren{\vbig\Tr  (F \cdot \tilde Q) - \e} -\eta \lVert  D \rVert_\infty
  \\
  &\annotaterel{\pref{eq:sos-3}}\geq
  - c \tau  \cdot \left(\frac{k m^2}{n-m}\right)^{d/4}\cdot \frac{\lVert  D \rVert_\infty}{\sqrt{\eta}} \psdrank(N) - \tau \cdot \e  -\eta \lVert  D  \rVert_\infty\,,
\end{align*}
where $c > 0$ is a universal constant.

Now set $\eta \seteq \min(\e/\|D\|_{\infty}, 1)$ and use \eqref{eq:tau-normalize} to bound $\tau \leq 1+\eta \leq 2$.
This yields
\begin{equation}
\label{eq:functional-lb}
  L_D(N) \ge - 2c \left(\frac{k m^2}{n-m}\right)^{d/4}\cdot \frac{\lVert  D \rVert_\infty^{3/2}}{\sqrt{\eps}} \psdrank(N) - 3\e
\end{equation}

Now recall that our invocation of \pref{thm:operator-approx-intro}
gives us a bound on $k = \deg(p)$:
\begin{equation}
\label{eq:degree-bound}
  k
  \precsim (1+\qe{Q}{U})\cdot \lVert  F \rVert / \e
  \annotaterel{\pref{eq:sos-1}, \pref{eq:entropy-deficit} }
  \precsim  \log \bigl(\psdrank(N)/\lVert N\rVert_1\bigr) \frac{\lVert  D \rVert_\infty}{\e}\,.
\end{equation}

  If $\psdrank(N)^2/\lVert  N \rVert_1$ satisfies the upper bound in the theorem, then the degree bound \pref{eq:degree-bound} above gives
  \begin{displaymath}
    k \precsim \tfrac 1 \e \cdot d \lVert  D \rVert_\infty\cdot \log n\mper
  \end{displaymath}
  Plugging this bound into \pref{eq:functional-lb} yields, for some constant $c' > 0$,
  \begin{displaymath}
    L_D(N) \ge - \left(\frac{c' d \|D\|_{\infty} m^2 \log n}{\e(n-m)}\right)^{d/4}\cdot \frac{\lVert  D \rVert_\infty^{3/2}}{\sqrt{\eps}} \psdrank(N) - 3\e
  \end{displaymath}
Since $\|N\|_1 \leq \|N\|_{\infty} \leq 1$,
  if $\psdrank(N)$ satisfies the upper bound in the theorem (for a sufficiently small constant $\alpha$), this lower bound is $L_D(N)\ge -4\e$ as desired (up to scaling by a factor of $4$).
\end{proof}

\subsection{Degree reduction}
\label{sec:degree-reduction}

\Jnote{Change to $\{0,1\}^m$ domain to be consistent!}

\Dnote{}

\Jnote{}

The next theorem is a restatement of \pref{thm:overview-degree-reduction}.  One should simply note that
for any symmetric matrix $A$, we have $\|A\|_F^2 = \Tr(A^2)$.

\begin{theorem}[Restatement of \pref{thm:overview-degree-reduction}]
  \label{thm:degree-reduction}
  \label{thm:matrix-version}
  Let positive integers $n \geq 1$ and $m,d,\ell \leq n$ be given.
  Suppose $A\from \binom{n}{m}\to \mathbb R^{p \times p}$ and $B\from \{0,1\}^n \to \mathbb R^{p \times p}$ are two functions
  taking {\em symmetric matrices} as values.
  Let $D\from \{0,1\}^m\to \R$ be a degree-$d$ pseudo-density and
  suppose that $\deg(B) \leq \ell$. %
  Then,
  \begin{align*}
    &\!\!\!\!\!\!\!\!\!\E_{S,x} D(x_S) \lVert  A(S) B(x) \rVert_F^2\\
    & \ge - 2  \lVert  D \rVert_{\infty}\,
    \Paren{\frac{\ell m}{(n-m)}}^{d/4} \cdot
    \Paren{\max_S\, \lVert  A(S)^2 \rVert}^{1/2}
    \left( \E_{S,x}  \lVert A(S) B(x) \rVert_F^2 \right)^{1/2} \cdot \left( \E_{x} \norm{B(x)}_F^2 \right)^{1/2} \mcom
  \end{align*}
\end{theorem}
\Dcomment{}

\Pnote{}

\Dnote{this bound should be tight because for degree-$d'$ sum-of-squares we have $R=n^{d'/2}$.
So if $\lVert  D \rVert_{\infty} \le n^{\e d}$ and $d'\le (1-100\e)d$, the theorem gives a lower bound of at least $- n^{(1-100\e)d/4}\cdot n^{\e d}/n^{(1-3\e)d/4}\ge -n^{-\e d}$}

\begin{proof}
For the sake of this lemma, which uses Fourier analysis, we will think of $B$ and $D$ as functions on $\{-1,1\}^n$.
Since this is a linear transformation on the domain, it does not affect their degrees as multilinear polynomials.

For every $S\subseteq [n]$ with $\lvert S \rvert=m$, we decompose $B$ into two parts $B=B_{S,\low}+ B_{S,\hi}$ such that $B_{S,\low}$ is the part of $B$ with degree at most $d/2$ in the variables $S$:
\begin{displaymath}
  B_{S,\low} = \sum_{\substack{\alpha\subseteq [n]\\ \lvert  \alpha\cap S \rvert \le d/2}} \hat B_{\alpha} \chi_\alpha\mper
\end{displaymath}
(Recall \pref{sec:prelims} for the Fourier-analytic definitions.)

The proof consists of two steps that are captured by the following two lemmas.

\begin{lemma}
Let $\tau = \max_S \lVert  A(S)^2\rVert$. Then,
\label{lem:matrix-1}
\begin{displaymath}
  \E_{S,x} D(x_S) \lVert A(S)B(x)\rVert_F^2
  \ge - 2 \sqrt{\tau} \lVert  D \rVert_{\infty}
  \cdot \left(\E_{S,x}  \lVert  B_{S,\hi}(x) \rVert_F^2\right)^{1/2} \cdot \left( \E_{S,x} \norm{A(S)  B(x)}_F^2 \right)^{\nfrac{1}{2}}
\end{displaymath}
\end{lemma}

\begin{proof}
For ease of notation, we will treat $A=A(S)$ and $B=B(x)$ as matrix-valued random variables that are determined by choosing $x\in \{-1. 1\}^n$ and $S\subseteq [n]$ with $\lvert S \rvert=m$ uniformly and independently at random.
In this notation, we are to lower bound the expectation $\E D(x_S) \lVert A B \rVert _F^2$ (over the joint distribution of $x$, $S$, $A$, and $B$).

Let $B_{\low}=B_{S,\low}(x)$ and $B_{\hi}=B_{S,\hi}(x)$ be matrix-valued random variables in the same probability space.
By construction, the Fourier transforms of the functions $x\mapsto  B_{S,\low}(x)$ and $x\mapsto  B_{S,\hi}(x)$ have disjoint support for every subset $S$.
Therefore, the expectation satisfies $\E B_{\low} B_{\hi}^T = 0$.
This fact allows us to control the expectations of $\lVert  A B_{\low} \rVert_F^2$ and $\lVert  A B_{\hi} \rVert_F^2$,
\begin{align*}
   \E\, \lVert  A B_{\low} \rVert_F^2 + \E\, \lVert  A B_{\hi} \rVert_F^2
   = \E\, \lVert  AB \rVert_F^2
   \label{eq:AB-decompose}
\end{align*}
Here, we have used that the quadratic formula $\lVert  AB \rVert_F^2 = \lVert AB_{\low}\rVert _F^2 + \lVert  AB_{\hi} \rVert_F^2 + 2 \langle AB_{\low} , AB_{\hi} \rangle$, where $\langle  \cdot ,\cdot \rangle$ is the inner product that induces $\lVert  \cdot  \rVert_F$, i.e., $\langle  X,Y \rangle=\Tr(X^T Y)$.
Hence, \[\E\left[ \langle  AB_{\low},AB_{\hi} \rangle\mid A
	\right]=\Tr(A^2 \cdot \E B_{\low}B_{\hi}^T)=0\,.\]
Therefore,
\begin{align*}
  &\hspace{-20mm}
    \left  \lvert \E \Brac{ D(x_S) \lVert  A B \rVert_F^2}
    -   \E \Brac{ D(x_S) \lVert  A B_{\low} \rVert_F^2}\right\rvert\\
  & \le \lVert  D \rVert_{\infty} \cdot \E \left[\Abs{\vbig\norm{A B}_F + \norm{AB_{\low}}_F}
    \cdot \Abs{\vbig\norm{A B}_F - \norm{AB_{\low}}_F}\right]\\
  & \le \lVert  D \rVert_{\infty} \cdot \Paren{\E \Abs{\vbig\norm{A B}_F + \norm{AB_{\low}}_F}^2
    \cdot \E \Abs{\vbig\norm{A B}_F - \norm{AB_{\low}}_F}^2}^{1/2}\\
  & \le 2\lVert  D \rVert_{\infty} \cdot \left(\E \norm{AB_{\hi}}^2_F\right)^{1/2}
    \cdot \left(\E \norm{AB}_F^2\right)^{1/2}\mper
\end{align*}
The first step used the identity $\lvert  x^2-y^2 \rvert=\lvert  x+y \rvert\cdot \lvert  x-y \rvert$.
In the second step, we applied Cauchy--Schwarz.
The third step used the triangle inequality,
\begin{math}
  \Abs{\vbig\norm{A B}_F - \norm{AB_{\low}}_F} \le \norm{AB_{\hi}}_F\mper
\end{math}

Since $x\mapsto \lVert  A(S) B_{S,\low}(x) \rVert_F^2$ is a sum of squares of polynomials of degree at most $d/2$ in the variables $S$
and $D$ is a degree-$d$ pseudo-density, the expectation $\E D(x_S) \lVert  A B_{\low} \rVert_F^2$  is non-negative.
It follows that
\begin{displaymath}
  \E \left[D(x_S) \lVert  A B \rVert_F^2\right] \ge - 2 \lVert  D
  \rVert_{\infty}  \cdot \left(\E\, \norm{AB_{\hi}}^2_F\right)^{1/2}
  \cdot \left(\E \norm{AB}_F^2 \right)^{1/2}\mper
\end{displaymath}
We also have
\begin{displaymath}
  \E\, \norm{A B_{\hi}}_F^2
  \le \max_S\, \lVert  A(S)^2 \rVert \cdot \E \lVert  B_{\hi} \rVert_F^2 = \tau \E \lVert  B_{\hi} \rVert_F^2\mper
\end{displaymath}

This bound implies the desired lower bound
\begin{displaymath}
  \E D(x_S) \lVert  A B \rVert_F^2
  \ge - 2 \sqrt{\tau} \lVert  D \rVert_{\infty}  \cdot  \Paren{\E \lVert
	  B_{\hi} \rVert_F^2}^{1/2} \cdot \Paren{\E \norm{AB}_F^2}
  \mper\qedhere
\end{displaymath}
\end{proof}

\begin{lemma}
\label{lem:matrix-2}
  \begin{displaymath}
    \E_{S,x} \lVert  B_{S,\hi}(x) \rVert_F^2 \le
    \frac{\ell^{d/2}m^{d/2}}{(n-m)^{d/2}} \cdot \E_{S,x} \norm{B(x)}_F^2
  \end{displaymath}
\end{lemma}

\begin{proof}
  By construction the Fourier transform of $B_{S,\hi}$ satisfies
  \begin{displaymath}
    B_{S,\hi}=\sum_{\substack{\alpha\subseteq [n]\\ \lvert  \alpha\cap S \rvert > d/2}} \hat B(\alpha) \chi_\alpha\mper
  \end{displaymath}
  Therefore,
  \begin{displaymath}
    \E_x \lVert  B_{S,\hi}(x) \rVert_F^2 = \sum_{\substack{\alpha\subseteq [n]\\ \lvert \alpha\cap S  \rvert>d/2}} \lVert \hat B(\alpha)\rVert_F^2\mper
  \end{displaymath}
  The expectation satisfies
  \begin{displaymath}
    \E_S \E_x \lVert  B_{S,\hi}(x) \rVert_F^2 = \sum_{\alpha\subseteq [n]} \lVert  \hat B(\alpha) \rVert_F^2 \cdot \Pr\left\{\vbig \card{\alpha\cap S}>d/2 \right\}\mper
  \end{displaymath}
  Since $B$ has degree at most $\ell$, we can upper bound the probability of the event $\{ \card{\alpha\cap S}>d/2 \}$,
  \begin{displaymath}
    \Pr\left\{\vbig \card{\alpha\cap S}>d/2 \right\} \le \binom{\ell}{d/2} \binom{n}{m-d/2} / \binom{n}{m} \le \frac{\ell^{d/2}m^{d/2}}{(n-m)^{d/2}}\mper
  \end{displaymath}
  Together with $ \sum_\alpha \lVert \hat B_\alpha \rVert_F^2= \E_x
  \lVert  B(x) \rVert_F^2$, the desired bound on the expected norm of $B_{S,\hi}$ follows:
  \begin{displaymath}
       \E_S \E_x \lVert  B_{S,\hi}(x) \rVert_F^2
       \le\frac{\ell^{d/2}m^{d/2}}{(n-m)^{d/2}} \cdot \E_x \norm{B(x)}_F^2\,.\qedhere
  \end{displaymath}
\end{proof}

  We combine the previous two lemmas to lower bound the correlation between the pseudo-density $D(x_S)$ and the norms $\lVert  A(S) B(x)\rVert_F^2$,
  \begin{align*}
   \E_{S,x} D(x_S) \lVert  A(S) B(x) \rVert_F^2
    & \ge - 2 \sqrt{\tau} \lVert  D \rVert_{\infty}
      \left( \E_{S,x} \norm{A(S)B(x)}_F^2 \right)^{1/2}
      \left(\E_{S,x} \lVert  B_{S,\hi}(x) \rVert_F^2\right)^{1/2} \\
      & \text{(using \pref{lem:matrix-1})}\\
   & \ge - 2 \sqrt{\tau} \lVert  D \rVert_{\infty} \frac {\ell^{d/4}m^{d/4}}{(n-m)^{d/4}}
     \left( \E_{S,x} \norm{A(S)B(x)}_F^2\right)^{1/2}
     \left(\E_{x} \norm{B(x)}_F^2\right)^{1/2} \\
   &\text{(using \pref{lem:matrix-2})}\mper
   \qedhere
  \end{align*}
\end{proof}

\subsection{Proof of the main theorem}
\label{sec:mainthms}

For a function $f\from \{0, 1\}^m \to[0,1]$ and an integer $n \geq m$,
let $M^f_n\from \binom{n}{m} \times \bits^n \to [0,1]$ be the matrix,
\begin{displaymath}
  M^f_n(S,x) \defeq f(x_S)\mper
\end{displaymath}

\begin{theorem} \label{thm:gtwo-pseudodist}
For any $m,d \geq 1$, the following holds.
  Let $f : \{0,1\}^m \to [0,1]$ be a nonnegative function with $d+2=\sosdeg(f)$.
  Then for $n \geq 2m$,
  \begin{equation} \label{eq:limitthm}
	  1+ n^{1+ d/2} \geq \psdrank\left(M_n^f\right) \geq C_f \left(\frac{n}{\log n}\right)^{\frac{d}{4}} \mcom
  \end{equation}
  where $C_{f} > 0$ is a constant depending only on $f$.

  Moreover, if there exists an $\e \in (0,1]$, and a degree-$d$ pseudo-density $D : \{0,1\}^m \to \R$
  with $\E_x D(x) f(x) < -\e$, then for every $n \geq 2m$, we have
  \begin{equation} \label{eq:quantitative}
    \psdrank(M_n^f)
    \ge  \left ( \frac{ c \e  n }{d m^2 \lVert  D \rVert_\infty \log n} \right  )^{d/4} \left(\frac{\e}{\lVert  D \rVert_\infty}\right)^{3/2} \sqrt{\E_x f(x)}\,,
  \end{equation}
where $c > 0$ is a universal constant.
\end{theorem}

\Dnote{P:DONE:  TODO: show matching lower bound here, i.e., show that if $f$ has sos degree $d$ then the corresponding matrix has $\gtwo$ at most $n^{d/2}$.
show be fairly easy computation.
the $Q_x$ matrices in the sos relaxation are rank $1$ and average to the $n^{d/2}$ dimensional identity matrix.
the $P_S$ matrices have trace $O(1)$ assuming that $\E f =O(1)$.}

\begin{proof}
Let $d+2 = \sosdeg(f)$ and consider a degree-$d$ pseudo-density with $\E D f < -\e$ for some $\e>0$.
Recall the linear functional $L_D$ defined in \pref{sec:proof-of-main-theorem}.
One observes that $L_D(M^f_n)<-\e$.

By (the contrapositive of) \pref{thm:psdrank-separation}, it follows that $\psdrank(M^f_n)^2\ge \alpha (n/\log n)^{d/2}\cdot \lVert  M^f_n \rVert_1$,
where $\alpha$ is a constant depending only on the parameters $\e,m,d$, and the pseudo-density $D$.
Note that $\lVert  M^f_n \rVert_1=\E f$.
This immediately implies \eqref{eq:limitthm}.
Likewise, \eqref{eq:quantitative} follows directly from \pref{thm:gamma-separation}.

\medskip

Let us now prove that $\psdrank(M_n^f) \leq 1+n^{1+ d/2}$ by exhibiting an
explicit factorization of $M_n^f$.
Let $\cF \defeq \{ A \subseteq [n] : |A| \leq 1+ d/2 \}$ and set $r =
|\cF|$.
 For $x \in
\bits^n$,
we use the notation $x^A \seteq \prod_{i \in A} x_i$.
Suppose $f = \sum_{j=1}^t g_j^2$ for some $\{g_j : \bits^m \to \R\}$ such
that $\deg(g_j) \leq 1+ d/2$ for $j \in [t]$.

For each function $j \in [t]$ and subset $S \sse [n]$ with
$|S| = m$, define the function $g_{S,j} : \bits^n \to \R$ by
$g_{S,j}(x) = g_j(x_S)$.  We associate the coefficient vector $\hat{g}_{S,j} : \cF \to \R$
associated to $g_{S,j}$ by letting $\hat{g}_{S,j}(A)$ be the coefficient of the monomial
$\prod_{i \in A} x_i$ in $g_{S,j}$.
Finally, for every $|S|=m$ and $x \in \bits^m$,
we define $r \times r$ PSD matrices indexed by $\cF$ as follows:
 $(Q_x)_{A,B} \seteq x^A x^B$ and $(P_S)_{A,B} \seteq \sum_{j = 1}^t
\hat{g}_{S,j}(A) \hat{g}_{S,j}(B)$.  It is easy to check that
$$ \Tr(P_S Q_x) = \sum_{j = 1}^t \left(\sum_{A \in \cF} x^A \hat{g}_{S,j}(A)\right)^2 = \sum_{j=1}^t g_{S,j}(x)^2 = f(x_S) = M_n^f(S,x)
\mcom$$
which yields an explicit psd factorization of $M^f_n$ with matrices
$\{ P_S \} , \{Q_x\}$
of dimension $r = \sum_{i \leq 1+d/2} \binom{n}{i} \leq 1+n^{1+d/2}$.
\end{proof}

\begin{comment}
Now assume that $n \geq \left(m^2 d \norm{D}_\infty/\epsilon\right)^{1+\delta}$.
By (the contrapositive of) \pref{thm:separation-technical}, we get the following lower bound on $\psdrank(M^f_n)$,
\begin{displaymath}
    \tfrac 1 {\lVert  M_n^f \rVert_1}\psdrank(M_n^f)^2
    \ge  \left ( \frac{ \alpha \e  n }{d m^2 \lVert  D \rVert_\infty \log n} \right  )^{d/2}
    \left(\frac{\e}{\lVert  D \rVert_\infty}\right)^3
    \ge  \left ( \frac{ \alpha n^{\delta/(1+\delta)}}{ \log n} \right  )^{d/2}
    \left(\frac{\e}{\lVert  D \rVert_\infty}\right)^3 \ge \Omega(n^{\delta d / 16})\mper
  \end{displaymath}
\Dnote{}
\end{comment}
%
%
%
%
%
%
%
%
%
%
%
%
%
%
%
%
%
%
%
%
%
%
%
%
%
%
%
%
%
%
%
%
%
%

%
%
%
%

\section{Approximations for density operators}
\label{sec:learning}

We turn now to a central theme of our approach:  High-entropy states can be approximated
by ``simple'' states if the approximation is only with respect to ``simple'' tests.
In our setting, ``simple'' will mean low-degree.
In \pref{sec:single-test}, we present a basic version
of this principle with respect to a single test functional.
This suffices for essentially all our applications to psd rank lower bounds.

We believe that the maximum-entropy approximation framework is a powerful one,
so \pref{sec:gener-appr-dens} is devoted to a more general exploration of the principle.
In particular, we state and prove approximation theorems for density operators
with respect to families of tests.
In the rest of this section, we fix a finite-dimensional real inner product space $H$.

\subsection{Approximation against a single test}
\label{sec:single-test}

The following theorem shows that a linear functional over density matrices with
high entropy is approximately minimized at a density matrix that is the square of a low-degree polynomial in the linear functional.  We recall that $U = \frac{\Id}{\Tr(\Id)}$ is the uniform density matrix.

\begin{theorem}[Density matrix approximation]
  \label{thm:operator-approx}
  Let $F \in \cM(H)$ be a symmetric matrix and let $Q \in \cD(H)$ be a density matrix.
  Then, for every $\e \in (0,\frac12)$, there exists a degree-$k$ univariate polynomial $p$ with  $k \le   O(\lVert F\rVert/\e) \cdot \qe{Q}{U} + O\Paren{\frac{\log 1/\e}{\log \log 1/\e}}$ such that
  \begin{displaymath}
    \Tr \Paren{\vbig F \cdot \tfrac 1{\Tr(p(F)^2)}p(F)^2} \le \Tr (F Q) + \e\mper
  \end{displaymath}
  Moreover, the polynomial $p$ depends only on $\e$, the operator norm $\lVert  F \rVert$, and the relative entropy $\qe{Q}{U}$.)
\end{theorem}

The proof consists of two steps.
First, we will show that the theorem holds with $\tfrac 1{\Tr(p(F)^2)}p(F)^2$ replaced by $e^{-\lambda F}/\Tr(e^{-\lambda F})$ for $\lambda \le (1/\e)\cdot \qe{Q}{U}$.
Then, we will approximate the matrix exponential by the square of a low-degree polynomial.

\begin{lemma}\label{lem:singletest}
  For every symmetric matrix $F$ and every density matrix $Q$,
  \begin{displaymath}
    \Tr \left (F \cdot \tfrac 1 {\Tr e^{-\lambda F}}e^{-\lambda F} \right) \le \Tr (F Q)  + \e
    \mcom
  \end{displaymath}
  as long as $\lambda \ge 1/\e \cdot \qe{Q}{U}$.
\end{lemma}

\begin{proof}
  By the duality formula for quantum entropy (see, e.g., \cite[Thm. 2.13]{MR2681769}),
  the function $f: X\mapsto \lambda \Tr (F X) + \qe{X}{\uId}$ over the the set of density matrices is minimized at $X^\star=e^{-\lambda F}/\Tr (e^{-\lambda F})$.
  Therefore, using the fact $\qe{X^\star}{\uId}\ge 0$, we get
  \begin{displaymath}
    \lambda \Tr (F X^\star) \le f(X^\star)\le f(Q)= \lambda \Tr (F Q) + \qe{Q}{\uId}
    \mcom
  \end{displaymath}
  which implies that $\Tr (F X^\star) \le \Tr (F Q) + \qe{Q}{\uId}/\lambda\le \Tr(F Q) + \e$, as desired.
\end{proof}

Next we observe that one can pass from univariate approximations of $e^x$ to approximations of $e^F$ in the trace norm.

\begin{lemma}
\label{lem:uniapprox}
Let $\delta \in (0,1]$ and $\tau > 0$ be given.
Suppose there exists a univariate polynomial $p(x)$ such that
for every $x \in [-\tau/2,\tau/2]$,
\begin{equation}\label{eq:uniapprox}
\left|e^x - p(x)\right| \leq \delta e^x\,.
\end{equation}
Then for every $F \in \mathcal M(H)$ with $\|F\| \leq \tau$, we have
\begin{equation}\label{eq:deg-goal-1}
\left\|\frac{e^F}{\Tr(e^F)} - \frac{p(F/2)^2}{\Tr(p(F/2)^2)}\right\|_* \leq 6 \delta\mper
\end{equation}
\end{lemma}

\begin{proof}
Under the assumptions, for every $x \in [-\tau,\tau]$, one has
\begin{align}
\left|e^x - p(x/2)^2\right| &= \left|e^{x/2} - p(x/2)\right| \cdot \left|e^{x/2} + p(x/2)\right| \nonumber \\
&\leq\nonumber
e^{x/2}(2+\delta)\left|e^{x/2} - p(x/2)\right| \\
&\leq \delta e^{x} (2+\delta) \nonumber \\
&\leq 3 \delta e^x\,,\label{eq:appg}
\end{align}
where the last line follows from $\delta \leq 1$.
Note the elementary equality:  For all $x,y,x',y' > 0$,
\begin{equation}\label{eq:element}
\frac{x}{y} - \frac{x'}{y'} = \frac{x-x'}{y} + \frac{y-y'}{y y'} x'\,.
\end{equation}
Let $\lambda_1, \lambda_2, \ldots, \lambda_n \in [-\tau,\tau]$ denote the eigenvalues of $F$.
We conclude that
\begin{align*}
\sum_{i=1}^n \left|\frac{e^{\lambda_i}}{\sum_{i=1}^n e^{\lambda_i}} - \frac{p(\lambda_i/2)^2}{\sum_{i=1}^n p(\lambda_i/2)^2}\right|
&\stackrel{\eqref{eq:element}}{\leq}
\sum_{i=1}^n \frac{\left|e^{\lambda_i} - p(\lambda_i/2)^2\right|}{\sum_{i=1}^n e^{\lambda_i}} + \frac{p(\lambda_i/2)^2 \left|\sum_{i=1}^n e^{\lambda_i} - p(\lambda_i/2)^2\right|}{\left(\sum_{i=1}^n e^{\lambda_i}\right)\left(\sum_{i=1}^n p(\lambda_i/2)^2\right)}
\\
&\stackrel{\eqref{eq:appg}}{\leq}
3\delta +
\sum_{i=1}^n
\frac{p(\lambda_i/2)^2 \left(3 \delta \sum_{i=1}^n e^{\lambda_i}\right)}{\left(\sum_{i=1}^n e^{\lambda_i}\right)\left(\sum_{i=1}^n p(\lambda_i/2)^2\right)} \\
&\,\,\leq 6 \delta\mper
\end{align*}
Since $e^F$ and $p(F)$ are simultaneously diagonalizable, the preceding inequality is precisely our goal \eqref{eq:deg-goal-1}.
\end{proof}

The following corollary of \pref{lem:uniapprox} follows by checking that the Taylor expansion of $e^x$ satisfies the approximation guarantee \pref{eq:uniapprox}.

\begin{corollary}\label{cor:taylor-exp}
For every $\e \in (0,\frac12)$ and every
symmetric matrix $F \in \mathcal M(H)$,
there is a number $k \leq 3e\left(\|F\|_{\infty} + \frac{\log (1/\e)}{\log \log (1/\e)}\right)$ and a univariate degree-$k$ polynomial $p_k$
with non-negative coefficients such that
\begin{equation}%
\left\|\frac{e^F}{\Tr(e^F)} - \frac{p_k(F/2)^2}{\Tr(p_k(F/2)^2)}\right\|_* \leq \e\mper
\end{equation}
\end{corollary}

\begin{proof}
Let $p_k(x) = \sum_{t=0}^k \frac{x^k}{k!}$.  By Taylor's theorem, we have
\[
\left|e^x - p_k(x)\right| \leq e^x \frac{x^{k+1}}{(k+1)!}\,.
\]
Define $\tau = \|F\|_{\infty}$ and
choose $k = \left\lfloor 3e\left(\tau + \frac{\log (1/\e)}{\log \log (1/\e)}\right)\right\rfloor$ so that for $x \in [-\tau/2,\tau/2]$, we have
$\frac{x^{k+1}}{(k+1)!}\leq \e/6$.  Finally, apply \pref{lem:uniapprox}.
\end{proof}

The proof of the main theorem in this section follows by combining \pref{lem:singletest} and \pref{cor:taylor-exp}.

\begin{proof}[Proof of \pref{thm:operator-approx}]
  Fix $\e \in (0,\frac12)$ and $F \in \cM(H)$, $q \in \cD(H)$.
  Choose $\lambda = (2/\e) \cdot \qe{Q}{\uId}$ and $F'=-\lambda F$.
  Let $p_k$ be the polynomial from \pref{cor:taylor-exp} for $k=3e\left( \lVert F'\rVert + \frac{\log (1/\e')}{\log \log (1/\e')}\right)$ and $\e'=\e / (2\lVert  F \rVert)$.
  Note that $k\le O(\lVert  F  \rVert/\e)\cdot \qe{Q}{\uId} + O\left(\frac {\log 1/\e}{\log \log 1/\e}\right)$.
  Moreover,
  \begin{align*}
    \Tr \Paren{F \cdot \tfrac 1 { \Tr (p_k(F'/2)^2)}  p_k(F'/2)^2}
    &\le \Tr \left(F \cdot \tfrac 1 {\Tr (e^{F'})} e^{F'}\right ) + \e' \cdot \lVert  F \rVert
    & \text{(by \pref{cor:taylor-exp})}\\
    &\le \Tr (F Q) +  \tfrac \e 2  + \e' \cdot \lVert  F \rVert
    & \text{(by \pref{lem:singletest})}
  \end{align*}
  Since $\e /2 + \e' \cdot \lVert  F \rVert\le \e$, the polynomial $p(x)=p_k(-\lambda x/2)$ satisfies the desired bound
  \[\Tr \left(F \cdot \tfrac 1 {\Tr(p(F)^2)}p(F)^2\right)\le \Tr (F Q) + \e\,.\]
\end{proof}

\subsection{Approximation against a family of tests}
\label{sec:gener-appr-dens}

Let $\cT \subseteq \cM(H)$ denote a compact set of matrices, and
set $\Delta(\cT) \seteq \sup_{A \in \cT} \|A\|$.
For $A \in \cM(H)$, we define the associated dual gauge
\[
[A]_{\cT}  \defeq \sup_{B \in \cT} \Tr(BA)\,.
\]

One should think of $\cT$ as a set of test functionals; for $A,A' \in \cM(H)$, the value
$[A-A']_{\cT}$ measures the extent to which $A$ and $A'$ are distinguishable
using tests from $\cT$.
It is important to note that if $\cT$ is not centrally symmetric, then $[\cdot]_{\cT}$
might also fail to be symmetric.
For future reference, we observe that fact that for any $A \in \cM(H)$,
\begin{align}
\left|[A]_{\cT}\right| &\leq \Delta(\cT) \|A\|_* \label{eq:gauge-holder} \\
[A+A']_{\cT} &\leq [A]_{\cT} + [A']_{\cT} \label{eq:gauge-triangle}\,.
\end{align}

Our main approximation theorem asserts that,
with respect to tests from a convex set $\cT$,
a high-entropy density operator can be well-approximated by
the square of a low-degree polynomial in some element of $\cT$.

\begin{theorem}[Approximation by a low-degree square]
  \label{thm:general-operator-approx}
  For every $\e \in (0,\frac12)$, the following holds.
  Let $\cT \subseteq \cM(H)$ be compact and convex, and
  let $Q \in \cD(H)$ be a density matrix.
  Then there exists a number
  \[
  k \lesssim \left(1+\qe{Q}{\uId}\right) \frac{\Delta(\cT)}{\e}\,,
  \]
  a univariate degree-$k$ polynomial $p$, and an element $F \in \cT$ such that $\Tr(p(F)^2)=1$ and
    \begin{displaymath}
  \left[Q - p(F)^2\right]_{\cT} \leq \e
        \mper
  \end{displaymath}
\end{theorem}

Just as for \pref{thm:operator-approx}, this is proved in two steps:  First we find
an initial approximator of a simple form, and then we construct from that a low-degree approximator.
In the next argument, it is helpful to have the following fact:
If $X(t)$ is continuously differentiable matrix-valued function, then for any $\beta \in \mathbb R$,
we have the Duhamel formula:
\begin{equation}\label{eq:integral-rep}
\frac{d}{dt} e^{\beta X(t)} = \int_0^{\beta} e^{\alpha X(t)} \frac{d X(t)}{dt} e^{(\beta-\alpha) X(t)}\,d\alpha\,.
\end{equation}
This can be verified immediately by showing that both sides satisfy
the differential equation
\[\frac{\partial F}{\partial \beta} = e^{\beta X} \frac{dX}{dt} + X(t) F(\beta,t)
\]
with $F(0,t)=0$ for all $t$.  (This argument is taken from \cite{Wilcox67}.)

We will only require \eqref{eq:integral-rep} for $\beta=1$.
For example, \eqref{eq:integral-rep} and cyclicity of the trace
yields
\begin{equation}\label{eq:trace-deriv}
\Tr\left(\frac{d}{dt} e^{X(t)}\right) = \Tr\left(e^{X(t)} \frac{dX(t)}{dt}\right)\,.
\end{equation}

\begin{comment}
For an operator $\Lambda$, let us now define
\[
\langle \Lambda \rangle_t \defeq \frac{\Tr(e^{X(t)} \Lambda)}{\Tr(e^{X(t)})}\,,
\]
and $\Delta(\Lambda) \seteq \Lambda - \langle \Lambda \rangle_t$.
Differentiating with respect to $t$ and using \eqref{eq:integral-rep} yields
\[
\frac{d}{dt} \langle \Lambda \rangle_t = \int_0^1 \left\langle e^{\alpha X(t)} \Delta\left(\frac{d X(t)}{dt}\right) e^{-\alpha X(t)} \Delta(\Lambda)\right\rangle\,d\alpha\,.
\]

If we plug in $\Lambda = \frac{d X(t)}{dt}$, we get
\[
\frac{d}{dt} \langle \frac{d X(t)}{dt} \rangle_t = \int_0^1 \left\langle e^{\alpha X(t)} \Delta\left(\frac{d X(t)}{dt}\right) e^{-\alpha X(t)} \Delta\left(\frac{d X(t)}{dt}\right)\right\rangle\,d\alpha\,.
\]
\end{comment}

Denote $X'(t) = \frac{dX(t)}{dt}$.
If we know that $X(t)$ is symmetric,
and its eigenvalues are $\{\lambda_i\}$, then
by diagonalizing in the basis of $X(t)$,
we can also derive
\begin{align*}
\Tr\left(X'(t) \frac{d}{dt} e^{X(t)}\right)  \nonumber
&=  \int_0^1 \Tr\left(X'(t) e^{\alpha X(t)} X'(t) e^{(1-\alpha) X(t)}\right)\,d\alpha  \\
&=  \sum_{i,j} \left(X'(t)\right)_{ij}^2 e^{\lambda_i} \int_0^{1} e^{\alpha (\lambda_j - \lambda_i)}\,d\alpha \\
&= \sum_{i,j} \left(X'(t)\right)_{ij}^2 \frac{e^{\lambda_i}-e^{\lambda_j}}{\lambda_i-\lambda_j}
\qquad\, \textrm{using } \int_0^1 e^{\alpha x}\,d\alpha = \frac{e^x-1}{x} \\
&\leq \sum_{i,j} \left(X'(t)\right)_{ij}^2 e^{\max(\lambda_i,\lambda_j)}\,,
\end{align*}
where in the final line we have used the fact that if $a \geq b$, then
\[
\frac{e^a - e^b}{a-b} = \frac{e^a(1-e^{b-a})}{a-b} \leq e^a\,,
\]
since $e^{b-a} \geq 1+(b-a)$.
Thus we have
\begin{align}
\Tr\left(X'(t) \frac{d}{dt} e^{X(t)}\right) &\leq  2 \sum_i e^{\lambda_i} \sum_{j} \left(X'(t)\right)_{ij}^2 \nonumber \\
&\leq 2 \Tr(e^{X(t)}) \max_i \sum_{j} \left(X'(t)\right)_{ij}^2 \nonumber \\
&=2  \Tr(e^{X(t)}) \left\|X'(t)\right\|^2_{2\to \infty} \nonumber\\
&\leq
2 \Tr(e^{X(t)}) \left\|X'(t)\right\|^2\,.\label{eq:magic-1}
\end{align}

Together these imply
\begin{align}
	\Tr\left(X'(t) \frac{d}{dt} \frac{e^{X(t)}}{\Tr(e^{X(t)})} \right)  \nonumber
&  = \Tr
	\left( \frac{X'(t)}{\Tr(e^{X(t)})} \frac{d e^{X(t)}}{dt}\right) -
	\frac{\Tr(X'(t) e^{X(t)})}{\Tr(e^{X(t)})^2}
	\Tr\left(\frac{d e^{X(t)}}{dt} \right)  \\
	& = \Tr\left( \frac{X'(t)}{\Tr(e^{X(t)})} \frac{d e^{X(t)}}{dt}\right) -
	\frac{\Tr(X'(t) e^{X(t)})}{\Tr(e^{X(t)})} \cdot \frac{\Tr(e^{X(t)} X'(t))}{\Tr(e^{X(t)})}  \nonumber \\
	& = \frac{1}{\Tr(e^{X(t)})}\Tr\left( X'(t)\frac{d e^{X(t)}}{dt}\right) -
\left(\frac{\Tr(X'(t) e^{X(t)})}{\Tr(e^{X(t)})}\right)^2 \nonumber \\
	& \stackrel{\eqref{eq:magic-1}}{\leq} 2 \lVert X'(t) \rVert^2 \label{eq:magic}\,.
\end{align}

We will use this for the following lemma.

\begin{lemma}[Sparse approximation by mirror descent]
\label{lem:mirror}
  For every $\e > 0$, the following holds.
  Let $\cC \subseteq \cM(H)$ be a compact set, and
  let $Q,Q_0 \in \cD(H)$ be density matrices.
    If one defines $h = \lceil \frac{8}{\e^2} \qe{Q}{Q_0} \Delta(\cT)^2\rceil$ then there exist $A_1, A_2, \ldots, A_h \in \cT$ such that
  \begin{equation}\label{eq:prescribed}
  \tilde Q \defeq \frac{\exp\left(\log Q_0 - \frac{\e}{4\Delta(\cT)^2} \sum_{i=1}^h A_i\right)}{\Tr\left(\exp\left(\log Q_0 - \frac{\e}{4\Delta(\cT)^2} \sum_{i=1}^h A_i\right)\right)} \in \cD(H)
  \end{equation}
satisfies
  \begin{equation}\label{eq:mirror-goal}
  \left[Q-\tilde Q\right]_{\cT} \leq \e.
  \end{equation}
\end{lemma}

\begin{proof}
Consider for $t \geq 0$, the density matrix
\[
Q_t = \frac{\exp\left(\log Q_0 -\int_0^t \Lambda_s\,ds\right)}{\Tr\left(\exp\left(\log Q_0-\int_0^t \Lambda_s\,ds\right)\right)}\,,
\]
where $s \mapsto \Lambda_s \in \cT$ is any measurable function.

First, one calculates
\begin{equation}\label{eq:derivQt}
\frac{d}{dt} \log Q_t = -\Lambda_t - \Id \frac{d}{dt} \log \Tr\left(\exp\left(\log Q_0-\int_0^t \Lambda_s\,ds\right)\right)
\end{equation}
Now, we have
\[
\frac{d}{dt} \log \Tr\left(\exp\left(\log Q_0 -\int_0^t \Lambda_s\,ds\right)\right)
= \frac{\frac{d}{dt} \Tr\left(\exp\left(\log Q_0 - \int_0^t \Lambda_s\,ds\right)\right)}{\Tr\left(\exp\left(\log Q_0-\int_0^t \Lambda_s\,ds\right)\right)}
\stackrel{\eqref{eq:trace-deriv}}{=} -\Tr(\Lambda_t Q_t)\,,
\]
and thus
\begin{align}
\frac{d}{dt} \qe{Q}{Q_t} = \Tr\left(Q \frac{d}{dt} \log Q_t\right)
&= - \Tr(\Lambda_t Q) + \Tr(Q) \Tr(\Lambda_t Q_t)
= - \Tr(\Lambda_t (Q-Q_t))\,, \label{eq:deriv}
\end{align}
where in the final line we have used $\Tr(Q)=1$.

Let $T = \frac{2}{\e} \qe{Q}{Q_0}$.
Suppose the map $t \mapsto \Lambda_t \in \cT$ is such that
\begin{equation}\label{eq:ensure}
\Tr(\Lambda_t (Q-Q_t)) > \frac{\e}{2}\quad\forall\,t \in [0,T]\,.
\end{equation}
Then from \eqref{eq:deriv} and \eqref{eq:ensure}, we arrive at
\[
\qe{Q}{Q_T} < \qe{Q}{Q_0} - \frac{\e}{2} T =  0\,,
\]
which contradicts the fact that $\qe{Q}{Q_T} \geq 0$.

Finally, we define the elements $A_1, \ldots, A_h \in \cT$
and corresponding approximators $\tilde Q_0, \tilde Q_1, \ldots, \tilde Q_{h} \in \cD(H)$ inductively.
Define, for $i=0,1,2,\ldots,h$, the times $t_i = i \frac{\e}{4\Delta(\cT)^2}$.
We will choose the map $t \mapsto \Lambda_t$ and put
$\tilde Q_i = Q_{t_i}$.

We begin by setting $\tilde Q_0 = \uId$ and $\Lambda_0 = 0$.
Now if $\left[Q-\tilde Q_i\right]_{\cT} \leq \e$, then we are done.
Otherwise, let $A_{i+1} \in \cT$ be such that
\begin{equation}\label{eq:unsat}
\Tr(A_{i+1} Q - A_{i+1} \tilde Q_i) > \e\,,
\end{equation}
and define
$\Lambda_t = A_{i+1}$ for $t \in (t_i, t_{i+1}]$.

Finally, observe that for $t \in (t_i, t_{i+1})$, we have
\begin{align*}
  \frac{d}{dt} \Tr(A_{i+1} Q - A_{i+1} Q_t)
  = \frac{d}{dt} \Tr\left(\Lambda_t  Q_t\right)
  \stackrel{\eqref{eq:magic}}{\geq} -2 \|\Lambda_t\|^2\,.
  \end{align*}
  where we have used the fact that $\Lambda_t = - \frac{d}{dt} \log \left(e^{\log Q_0-\int_0^t \Lambda_s\,ds}\right)$.
We conclude that
\begin{align*}
\Tr(A_{i+1} Q - A_{i+1} Q_{t_{i+1}}) &\geq \Tr(A_{i+1} Q - A_{i+1} Q_{t_{i}}) -2\|\Lambda_t\|^2 (t_{i+1}-t_i)
\\
&\geq \Tr(A_{i+1} Q - A_{i+1} Q_{t_{i}}) - \frac{\e}{2} \\
&> \frac{\e}{2}\,,
\end{align*}
using \eqref{eq:unsat}.  Thus we either find an approximator $\tilde Q_i$ for some $i=0,1,\ldots,h$
satisfying \eqref{eq:mirror-goal} or \eqref{eq:ensure} holds.
But we have already seen that the latter possibility cannot happen.
Observe that the approximators $\tilde Q_i$ are all of the desired form \eqref{eq:prescribed}.
\end{proof}

\begin{proof}[Proof of \pref{thm:general-operator-approx}]
First, we apply \pref{lem:mirror} with $Q_0=\uId$ to
obtain an approximation $\tilde Q$ of the form
\[
\tilde Q = \frac{e^{\lambda F}}{\Tr(e^{\lambda F})}
\]
with $|\lambda| \lesssim 1+\frac{1}{\delta} \qe{Q}{U}$
and which satisfies
\begin{equation}\label{eq:app1}
\left[Q-\tilde Q\right]_{\cT} \leq \e/2\mper
\end{equation}
Note here that since $\cT$ is assumed to be convex, we have $F \in \cT$ (see the form of \eqref{eq:prescribed}).

Then we apply Corollary \ref{cor:taylor-exp} to $\lambda F$ to obtain a degree-$k$ polynomial
$p_k$ such that
\begin{equation}\label{eq:app2}
\left\|\tilde Q - \frac{p_k(\lambda F/2)^2}{\Tr(p_k(\lambda F/2)^2)} \right\|_* \leq \frac{\e}{2\Delta(\cT)}\,,
\end{equation}
where \[k \lesssim |\lambda| \Delta(\cT) + \frac{\log(\Delta(\cT)/\e)}{\log \log (\Delta(\cT)/\e)} \lesssim (1+\qe{Q}{\uId}) \frac{\Delta(\cT)}{\e} \,.\]

Thus we conclude that
\begin{eqnarray*}
\left[Q-\frac{p_k(\lambda F/2)^2}{\Tr(p_k(\lambda F/2)^2)}\right]_{\cT} &\stackrel{\eqref{eq:gauge-triangle}}{\leq}&
\left[Q-\tilde Q\right]_{\cT} +
\left[\tilde Q - \frac{p_k(\lambda F/2)^2}{\Tr(p_k(\lambda F/2)^2)}\right]_{\cT}
\\
&\stackrel{\eqref{eq:gauge-holder}}{\leq}&
\left[Q-\tilde Q\right]_{\cT} +
\Delta(\cT) \left\|\frac{p_k(\lambda F/2)^2}{\Tr(p_k(\lambda F/2)^2)}-\tilde Q\right\|_{*} \\
&\stackrel{\eqref{eq:app1} \wedge \eqref{eq:app2}}{\leq}& \e\mper \hfill \qedhere
\end{eqnarray*}
\end{proof}

\ifnum\entmaxmode=1
\input{content/ent-max.tex}
\fi

\subsubsection{Junta approximation}

We record here the following application to ``classical'' functions by
restricting \pref{lem:mirror} to the diagonal case.
If $X$ is a finite set, and $\cT$ is a collection of real-valued functions on $X$,
we extend the notation $\Delta(\cT) = \sup_{g \in \cT} \|g\|_{\infty}$.
If $\mu$ is a measure on $X$, and $f : X \to \Rnn$ satisfies $\E_{\mu} f = 1$,
we abuse notation by writing
\[
\ce{f}{\mu} = \E_{\mu} [f \log f]\,.
\]
for the relative entropy between $f \mu$ and $\mu$.
We will also allow ourselves to conflate $\mu$ with the
corresponding density by writing $\mu(x)$ for $\mu(\{x\})$ and $x \in X$.
One should note that an analog of \pref{lem:mirror} for the special
case of probability distributions
(instead of density matrices) can be proved exactly along the same lines,
but without the use of matrix inequalities.

\begin{corollary}[Sparse approximation of functions by mirror descent]
\label{cor:mirror}
  For every $\e > 0$, the following holds.
  Let $X$ be a finite set equipped with a probability measure $\mu$.
  Let $\cT \subseteq L^2(X,\mu)$ be a compact set of functions,
  and let $f :X \to \Rnn$ be such that $\E_{\mu} f = 1$.
  If one defines $h = \lceil \frac{8}{\e^2} \ce{f}{\mu} \Delta(\cT)^2\rceil$ then there exist functions $g_1, g_2, \ldots, g_h \in \cT$ such that
  \begin{equation}\label{eq:prescribed-cor}
  \tilde f \defeq \frac{\exp\left(\frac{\e}{4\Delta(\cT)^2} \sum_{i=1}^h g_i\right)}{\sum_{x \in X} \exp\left(\frac{\e}{4\Delta(\cT)^2} \sum_{i=1}^h g_i(x)\right) \mu(x)}
  \end{equation}
  so that $\E_{\mu} \tilde f = 1$, and for every $g \in \cT$,
  \begin{equation}\label{eq:mirror-cor-goal}
  \E_{x \sim \mu} g(x) \left(f(x)-\tilde f(x)\right) \leq \e.
  \end{equation}
\end{corollary}

\begin{proof}
$H$ the Euclidean space $\mathbb R^{\{0,1\}^n}$, and let $\{e_x : x \in \{0,1\}^n\}$
be an orthornormal basis of $H$.
We will represent $f$ by the diagonal matrix $M(f) \in \cD(H)$ defined by
\[
M(f) = \sum_{x \in \{0,1\}^n} f(x) e_x e_x^T \mu(x)\,.
\]
We also lift each test $g$ to a matrix $M(g) = \sum_{x \in \{0,1\}^n} g(x) e_x e_x^T$
and the set $M(\cT)$ now denotes a class of matrix tests.

Furthermore, we write
\[
Q_0 = \sum_{x \in \{0,1\}^n} e_x e_x^{T} \mu(x)\,
\]
so that $\qe{M(f)}{Q_0} = D(f\,\|\,\mu)$.

Applying Lemma \ref{lem:mirror} yields an approximation $\tilde Q$ to $M(f)$ of the form
\[
\tilde Q = Q_0 \cdot M\,,
\]
where $M$ is a diagonal matrix.  Furthermore, by construction,
the function $\tilde f : \{0,1\}^n \to \mathbb R$ given by $\tilde f(x) = \langle e_x, M e_x\rangle$
has the form \eqref{eq:prescribed-cor}.
Finally, the approximation guarantee $[M(f) - \tilde Q]_{M(\cT)} \leq \e$ is precisely \eqref{eq:mirror-cor-goal}.
\end{proof}

We now apply the preceding corollary to prove an approximation-by-juntas theorem.
An essentially equivalent result for Boolean domains is proved in \cite{DBLP:conf/focs/ChanLRS13},
but it is instructive to see that it falls easily out of the learning framework.
Fix $n \geq 1$ and a finite set $X$.
We recall that for a subset $S \subseteq \{1,\ldots,n\}$, a function $f : X^n \to \mathbb R$ is called an {\em $S$-junta}
if $f$ only depends (at most) on the coordinates in $S$.  In other words, for all $x,x' \in X^n$, if $x|_S = x'|_S$ then $f(x)=f(x')$.
We say that $f$ is a {\em $k$-junta} if it is an $S$-junta for a set with $|S|=k$.

\begin{theorem}[Junta approximation]
\label{thm:junta-approx}
Let $X$ be an arbitrary finite set, and
let $\mu$ denote a probability measure on $X^n$.
Consider a non-negative function $f : X^n \to \mathbb R_+$ with $\E_{\mu} f = 1$,
and let $\cT$ be a collection of $k$-juntas.
Then for every $\e > 0$, there exists a non-negative $k'$-junta $\tilde f : X^n \to \mathbb R_+$
with $\E_{\mu} \tilde f = 1$, where
\[
k' \lesssim \frac{k}{\e^2} D(f \,\|\,\mu) \Delta(\cT)^2\,,
\]
and such that for every $g \in \cT$,
\begin{equation}\label{eq:mu-approx}
\E_{x \sim \mu} g(x) \left(f(x)-\tilde f(x)\right) \leq \e\mper
\end{equation}
\end{theorem}

\begin{proof}
Applying Corollary \ref{cor:mirror} yields an approximation $\tilde f$.
One simply notes that from \eqref{eq:prescribed-cor}, $\tilde f$
is an $hk$-junta where $h \lesssim \frac{1}{\e^2} D(f \,\|\,\mu) \Delta(\cT)^2$
\end{proof}

\section{The correlation polytope}
\label{sec:corr}

Recall the correlation polytope  $\corr_n \subseteq \mathbb R^{n^2}$ given by
\[
\corr_n = \mathrm{conv}\left(\{ xx^T : x \in \{0,1\}^n \}\right).
\]
This polytope is also known as the {\em Boolean quadric polytope} \cite{Padberg89} for
the following reason.

\begin{proposition}[Restatement of \pref{prop:quadric-intro}]
\label{prop:quadric}
If $f : \{0,1\}^m \to \Rnn$ is a nonnegative quadratic function over $\{0,1\}^m$,
then for any $n \geq m$, $M_n^f$ is a submatrix of some slack matrix associated to $\corr_n$.
\end{proposition}

\begin{proof}
Let $\langle A,B\rangle = \Tr(A^T B)$ denote the Frobenius inner product on $\mathbb R^{n^2}$.
Suppose that $f(x) = \sum_{i \leq j} a_{ij} x_i x_j + a_0 \geq 0$
for all $x \in \{0,1\}^n$.
We claim that this gives a valid linear inequality
for $\corr_n$ as follows:  For all $x \in \{0,1\}^n$,
\[
f(x) = \langle A, xx^T\rangle + a_0 \geq 0\,,
\]
where $A$ is the matrix $A = (a_{ij})$.
Since this inequality holds at the vertices, it holds for all of $\corr_n$.
\end{proof}

We now recall the relationship between the correlation, cut, TSP, and stable set polytopes.  The first fact is from \cite{DeSimone90}, while the second two
are taken from \cite{FMPTW12}.

\begin{proposition}
\label{prop:polytope-relations}
For every $n \geq 1$, the following hold:
\begin{enumerate}
\item $\corr_n$ is linearly isomorphic to $\cut_{n+1}$.
\item There exists a number $a_n \leq O(n^2)$ such that
some face of $\tsp_{a_n}$ linearly projects to $\corr_n$.
\item There exists a graph $H_n$ on $b_n \leq O(n^2)$ vertices
such that some face of $\stab_{b_n}(H_n)$ linearly projects to $\corr_n$.
\end{enumerate}
\end{proposition}

\ifnum\nnrmode=1
\subsection{Positive semidefinite rank}
\fi
\label{sec:psd-corr}

We will now prove a lower bound on the psd rank of $\corr_n$.
Our first goal is to construct a suitable family of pseudo-densities.
We will employ Grigoriev's work \cite{DBLP:journals/cc/Grigoriev01}
on degree lower bounds for
Positivstellensatz calculus refutations.
The primary difficulty will be in expressing
Grigoriev's lower bound using a pseudo-density of small norm.

\begin{theorem} \label{thm:knapsack-pseudo}
	Fix an odd integer $m \geq 3$.  There exists a degree-$m$
	pseudo-density $D : \bits^m \to \R$ such that
$$  \E_x D(x) \left( \sum_{i=1}^m x_i - \frac{m}{2}
 \right)^2 = 0 \mcom$$
 and $$ \norm{D}_\infty \leq m^{3/2}\,.$$
\end{theorem}

\begin{proof}
\begin{comment}
Equivalently, the work of Grigoriev \cite{DBLP:journals/cc/Grigoriev01} showed that low
degree SoS relaxations cannot refute the following system of equations
over $\{X_i\}_{i \in [m]}$ associated with the knapsack problem.
$$ X_i^2 - X_i = 0 \quad \forall i \in [m] \qquad \text{ and } \qquad
\sum_{i \in [m]} X_i =
\frac{m}{2} \mper $$
\end{comment}

Grigoriev constructs a linear functional $\mathcal{G}$ on the space of $m$-variate real polynomials
modulo the ideal $\mathcal{I}$ generated by $\{X_i^2-X_i : i=1 \in [m]\}$:
\[
\mathcal{G} : \mathbb R[X_1, \ldots, X_m]/\mathcal{I} \to \bbR\,.
\]
His functional satisfies
\begin{equation} \label{eq:grigfunctional2}
 \mathcal{G}\left(p(X)^2\right) \geq 0 \qquad  \forall p \in \mathbb R[X_1, \ldots, X_m]/\mathcal I,\,\,  \deg(p) \leq m/2\,,
\end{equation}
and
\begin{equation}\label{eq:grigfunctional1}
 \mathcal{G}\left(\left( \sum_{i=1}^m X_i - \tfrac{m}{2}
 \right)^2\right) = 0\,.
\end{equation}

The functional is uniquely defined by the values
$$\mathcal{G} (X^S) \defeq
\frac{\binom{m/2}{|S|}}{\binom{m}{|S|}} \mcom$$
for each multilinear monomial $X^S = \prod_{i \in S} X_i$ with $S \subseteq [m]$.
Observe that $m/2$ is not an
integer and the (generalized) binomial coefficient $\binom{m/2}{k}$ is defined using
the formal expression
$$ \binom{r}{k} = \frac{ r \cdot (r-1)  \cdots  (r-k+1) }{k \cdot
(k-1)  \cdots 1} \mper$$
It is easy to check that $\mathcal{G}$ satisfies
\eqref{eq:grigfunctional1}:  \Jnote{}
\begin{align*}
 \mathcal{G} \left(\left(\sum_{i=1}^m X_i - \tfrac{m}{2}\right)^2\right)  &=
 \sum_{i=1}^m \mathcal{G} (X_i^2) + 2 \sum_{i \neq j}
\mathcal{G} (X_i X_j)  - m \sum_{i=1}^m \mathcal{G}(X_i) + \frac{m^2}{4} \\
&= \frac{m}{2} +  m(m-1) \frac{m-2}{4(m-1)} - m \frac{m}{2} + \frac{m^2}{4}
= 0
\mper
\end{align*}
Grigoriev shows that $\mathcal{G}$ satisfies
\eqref{eq:grigfunctional2} \cite[Lem. 1.4]{DBLP:journals/cc/Grigoriev01}.

We will construct a pseudo-density $D : \{0,1\}^m \to \R$ such that $\E_x D(x)
p(x) = \mathcal{G}(p(X_1, \ldots, X_m))$ for every multilinear polynomial $p$.
Observe that $\mathcal{G}$ is invariant under
permutation of variables $\{X_1, \ldots, X_m\}$.
For $w=0,1,\ldots, m$, let $c_w$ denote the unique degree $m$ polynomial
such that,
$$ c_w(t) = \begin{cases}  1 & \text{ if } t = w \\ 0 & \text{ if } t
\in \{0,1,\ldots, m\} \setminus \{w\} \end{cases}$$
We claim that for any univariate real polynomial $p$ with $\deg(p) \leq m$,
	\begin{equation} \label{eq:unipoly}
	 \sum_{w = 0}^m p(w) \cdot c_{w}(t)  = p(t) \mper
 \end{equation}
Both sides of the claimed identity are univariate polynomials in $t$ of degree at most
$m$ and agree with each other on the $m+1$ points given by $t \in \{0,1,\ldots, m\}$.
Hence, the two polynomials are identically equal.

For each $x \in \bits^m$,
let $|x|$ denote its hamming weight, and
define
$$ D(x) \defeq 2^m  \cdot
\frac{c_{|x|}\left(\nfrac{m}{2}\right)}{\binom{m}{|x|}} \mper$$

We claim that $D$ satisfies $\E_x D(x)
p(x) = \mathcal{G} \left(p(X)\right)$ for every polynomial multilinear real polynomial $p$.
To see this, consider any monomial $x^S=\prod_{i \in S} x_i$ with $S \subseteq [m]$.  Put $\ell = |S|$.  Then we have:
\begin{align*}
\E_x D(x) x^S & =  \E_x D(x)\cdot \frac{1}{\binom{m}{\ell}} \left(
\sum_{T \sse [m], |T| = \ell} x^T \right) \qquad (\text{symmetry
of } D)\\
& =  \E_x D(x) \cdot \frac{1}{\binom{m}{\ell}} \binom{ |x|}{\ell}
\\
& = \sum_{w = 0}^m \frac{\binom{m}{w}}{2^m} \E_{x}\left[ D(x) \cdot
\frac{1}{\binom{m}{\ell}} \binom{ |x|}{\ell} \,\,\Big|\,\, |x| = w
\right]\\
& = \sum_{w = 0}^m c_w(\nfrac{m}{2})
\frac{\binom{w}{\ell}}{\binom{m}{\ell}} \\
& = \frac{ \binom{ \nfrac{m}{2}}{\ell}}{\binom{m}{\ell}} \\
& =
\mathcal{G}\left(X^S\right)  \qquad\qquad\qquad ( \text{using } \eqref{eq:unipoly}
\text{ with } p(t) = \textrm{\small $\binom{t}{\ell}$} )\,.  \\
\end{align*}
Finally, in order to bound $\norm{D}_\infty$, observe that the
polynomials $c_w(t)$ are given by the interpolation formula
$$ c_w(t) = \frac{\prod_{a = 0, a \neq w}^m (t-a) }{ \prod_{a = 0, a
\neq w}^m (w-a)} \mper$$
For an $x \in \bits^m$ with $|x| = w$ we have,
\begin{align*}
|D(x)|  = \left| 2^m \cdot
\frac{c_w(\nfrac{m}{2})}{\binom{m}{w}} \right|
       & = \left| \frac{\prod_{a = 0, a\neq w}^m (m - 2a)}{\prod_{a = 0, a
\neq w}^m (w - a)} \cdot \frac{1}{\binom{m}{w}} \right| \\
	& =\left| \frac{\prod_{a = 0, a\neq w}^m (m - 2a)}{(w-1)! (m-w)!}
\cdot  \frac{1}{\binom{m}{w}} \right| \\
& = m \cdot \frac{1}{2^{m-1}}  \cdot \binom{m-1}{\nfrac{(m-1)}{2}}   \cdot
\frac{w}{|m-2w|} \\
&\leq \frac{m}{\sqrt{\frac32 (m-1)+1}} \left(\frac{m+1}{2}\right) \\
&\leq m^{3/2}\,,
 \end{align*}
 where in the last step we have used Sterling's approximation
 for the inequality $\binom{m-1}{(m-1)/2}
 \leq 2^{m-1}/\sqrt{\frac32 (m-1)+1}$, valid for $m \geq 3$.
\end{proof}

\begin{theorem}  There is a constant $\alpha > 0$ such that for every $n \geq 1$,
	$$\psdrank(\corr_n) \geq 2^{\alpha n^{2/13}}\,.$$
\end{theorem}
\begin{proof}
For $m \geq 1$,
define $f: \bits^m \to \Rnn$ by
	\begin{equation}\label{eq:f-form}
f(x) =  \frac{1}{m^2} \left(\left(\sum_{i=1}^m x_i -
	\frac{m}{2}\right)^2 - \frac{1}{4} \right)\mcom
\end{equation}
and let $M^f_n : \binom{n}{m} \times \bits^n \to \Rnn$ be given by
$M^f_n(S,x) = f(x_S)$ as in \eqref{eq:restriction}.

By \pref{thm:knapsack-pseudo}, there exists a degree-$m$
pseudo-density $D: \bits^m \to \R$ such that $\E_x
D(x) f(x) = -\frac{1}{4m^2}$ and $\norm{D}_\infty \leq m^{3/2}$.
Fix $\e = 1/(4m^2)$ and $d=m$ and
apply \pref{thm:gtwo-pseudodist} to conclude that there is a constant $\alpha' > 0$
such that for $n \geq 2m$, we have
\[
\psdrank(M_n^f) \geq \left(\frac{\alpha' n}{m^{13/2} \log
n}\right)^{m/4} \cdot m^{-21/4} \cdot m^{-1/2}\]
Choosing $n \geq \frac{2}{\alpha'} m^{13/2} \log n$, an easy calculation shows that
\[
\psdrank(M_n^f) \geq 2^{\Omega(n^{2/13})}\,.
\]
By \pref{prop:quadric}, we have
$\psdrank(\corr_n) \geq \psdrank(M_n^f)$, completing the proof.
\end{proof}

\section{Optimality of low-degree sum-of-squares for max CSPs}
\label{sec:maxcsp}

Constraint satisfaction problems form a broad class of discrete
optimization problems that include, for example, \maxcut and \maxthreesat.
For simplicity of presentation, we will focus on constraint satisfaction problems with a boolean alphabet, though
similar ideas extend to larger domains (see an analogous
generalization in Section \ref{sec:nnr}).
We begin our presentation with a formal definition of semidefinite
programming relaxations for max-CSPs.

\subsection{The SDP approximation model}
\label{sec:sdp-model}

\iffalse
For a finite collection $\Pi=\set{P}$ of $k$-ary predicates $P\from
\{0,1\}^k\to \bits$, we let \maxpi denote the following optimization
problem:
%
An instance $\inst$ consists of boolean variables $X_1,\ldots,X_n$ and
a collection of $\Pi$-constraints $P_1(X)=1,\ldots,P_M(X)=1$ over
these variables.
%
A $\Pi$-constraint is a predicate $P_0\from \{0,1\}^n\to \bits$ such
that $P_0(X)=P(X_{i_1},\ldots,X_{i_k})$ for some $P\in\Pi$ and {\em distinct} indices
$i_1,\ldots,i_k\in[n]$.
%
The objective is to find an assignment $x\in \{0,1\}^n$ that satisfies as
many of the constraints as possible, that is, which maximizes
\begin{displaymath}
  \inst(x)\defeq \tfrac1M\sum_{i=1}^M P_i(x)\mper
\end{displaymath}
%
We denote the optimal value of an assignment for $\inst$ as
$\opt(\inst)=\max_{x\in\{0,1\}^n}\inst(x)$.

\medskip

\emph{Examples:} \maxcut corresponds to the case where $\Pi$ consists of
the binary inequality predicate.
%
For \maxthreesat, $\Pi$ contains all eight 3-literal disjunctions,
e.g., $X_1\vee \bar X_2\vee \bar X_3$.
%
\fi
\newcommand{\tinst}{\tilde\inst}
\newcommand{\tx}{\tilde x}

In order to write an SDP relaxation for a max-CSP,
one needs to {\it linearize} the objective function.
For $n\in \N$, let $\maxpi_n$ be the set of \maxpi instances on $n$ variables.
An {\em SDP-relaxation of size $r$} for $\maxpi_n$ consists of
the following.
\begin{description}
\item[Linearization:]
  Let $r$ be a natural number.
  For every $\inst\in\maxpi_n$, we associate a vector $\tinst \in
  \R^{r \times r}$ and for every assignment $x\in\{0,1\}^n$, we associate a point
  $\tx\in\R^{r \times r}$, such that $\inst(x)=\iprod{\tinst,\tx}$ for all
  $\inst\in \maxpi_n$ and all $x\in\{0,1\}^n$.
\item[Feasible region:] The feasible region is a closed, convex (possibly unbounded)
	spectrahedron $ \cS \sse \R^{r \times r}$ described as the intersection
of the cone of $r \times r$ PSD matrices with an affine linear
subspace:
	$$ \cS = \{ y \in \R^{r \times r} | Ay = b, y \in
\cS_r^+\} \mcom$$
 such that $\tx\in \cS$ for all assignments
  $x\in\{0,1\}^n$.  Note that the spectrahedron $\cS$ is independent of the instance
  $\inst$ of $\maxpi_n$.
\end{description}

Given an instance $\inst\in \maxpi_n$, the SDP relaxation $\mathcal S$
has value \[\mathcal S(\inst)\defeq \max_{y\in \cS} \iprod{\tinst,y}\,.\]
Since $\tx\in \cS$ for all assignments $x\in\{0,1\}^n$ and
$\iprod{\tinst,\tx}=\inst(x)$, we have $\mathcal S(\inst)\ge
\opt(\inst)$ for all instances $\inst\in\maxpi_n$.

\paragraph{Low-degree sum-of-squares relaxations}
We will now describe the low-degree sum-of-squares relaxation as
it applies to a max-CSP.
\newcommand{\degdcone}{\cC^{\mathrm{sos}}_d}
Let $\Pi$ be a max-CSP with arity $k$.  Given an instance $\inst$ of
$\maxpi_n$, we recall
that we think of it as a function $\inst: \bits^n \to \R$
given by $\inst(x) = \frac{1}{m} \sum_{i = 1}^m P_i(x)$ where
$\{P_i\}_{i \in [m]}$ are the constraints in $\inst$.
Define the cone $\degdcone \sse \R^{\bits^n}$ as the cone generated by
squares of polynomials of degree at most $d/2$, i.e.,$$\degdcone =
\mathrm{Cone}\left(\{g^2 \mid g: \bits^n \to \R, \deg(g) \leq d/2\}\right) \mper $$

The {\em degree-$d$ sos relaxation} for $\inst$ is given by
\begin{equation} \label{eq:sosrelaxation}
 \sosub_d(\inst) \defeq \min \left\{c \mid c - \inst \in \degdcone\right\}
 \end{equation}
We will now write the dual formulation of the above semidefinite
program to expose the underlying spectrahedron and linearization.
The dual of \eqref{eq:sosrelaxation} can be written as,
\begin{align}
\sosub_d(\inst) =   \max_{D : \bits^n \to \R}   &\ \ \ \iprod{D, \inst}
  \label{eq:sosdual}\\
  \text{ subject to } & \ \ \   \iprod{D, 1} = 1 \mcom \nonumber \\
	& \ \ \  \iprod{D, h} \geq 0 \ \ \ \forall h \in \degdcone
	\mper \nonumber
\end{align}
The function $D : \bits^n \to \R$ is referred to as a pseudo-density
over $\bits^n$, since it satisfies that for every degree $d/2$
function $g$,  $\E_x D(x) g^2(x) \geq 0$.

Notice that all the constraints on the pseudodensity $D : \bits^n \to \R$ correspond to
inner products with functions of degree at most $d$.  Hence, without
loss of generality, we may assume $\deg(D) \leq d$.
Alternately, the convex program \eqref{eq:sosdual} can be written
succinctly in terms of the low-degree part of $D$.  We will now carry
this out explicitly and thereby identify the feasible region
associated with the degree-$d$ sos relaxation.

To this end, set $\cF \seteq \{ A : A \sse [n], |A| \leq d/2\}$ and let $r = | \cF | \leq
\sum_{i = 0}^{d/2} \binom{n}{i}$.
Recall that $\cS_r^+ \sse \R^{r \times r}$ is the cone of $r \times r$ PSD
matrices.  We will index the matrices in $\cS^+_r$ using elements of $\cF$
in the natural way.
Define a matrix $Y : \cF \times \cF \to \R$ as follows,
$$ Y(A,B) = \left\langle D, \prod_{i \in A \cup B} x_i\right\rangle \mper $$
By definition of $Y$, it is clear that $Y(A,B) = Y(B,A) = Y(A
\cup B, \emptyset)$ for all $A,B \in \cF$.  Moreover, we have $Y(\emptyset,
\emptyset) = \iprod{D,1} = 1$.
Furthermore, the matrix $Y$ is PSD since, for all
$\hat{g}: \cF \to \R$, we have
\begin{align*}
 \langle \hat{g}, Y \hat{g}\rangle & = \sum_{A,B \in \cF} \hat{g}_A \hat{g}_B Y(A,B) \\
& = \left \langle D, \sum_{A,B \in \cF} \hat{g}_A \hat{g}_B \prod_{i
\in A \cup B} x_i \right \rangle \\
& = \left \langle D, \left(\sum_{A \in \cF} \hat{g}_A
\prod_{i \in A} x_i\right)^2 \right \rangle  &  \text{ using } x_i^2 =
x_i \quad \forall i  \in [n], x\in \bits^n
\\
&\geq 0
\end{align*}
where the final inequality used the fact that $\iprod{D, g^2} \geq 0$
for  all functions $g$ with $\deg(g) \leq d/2$.

From the above discussion, it is clear that the {\it feasible region} of the
degree $d$-sos relaxation \eqref{eq:sosdual} corresponds to
the spectrahedron,
$$ \cS \defeq \{ Y \in \R^{r \times r} \mid Y \in \cS_r^+, Y(\emptyset,\emptyset) =
1\text{ and
}Y_{A,B} = Y_{B,A} = Y_{A \cup B, \emptyset}\ \ \forall
A,B \in \cF \}$$

Now we describe the {\it linearization} associated with the degree-$d$
sos relaxation.
For every assignment $x \in \bits^n$, associate the matrix
$\tilde{x} : \cF \times \cF \to \R$
given by
\begin{equation} \label{eq:cspsos1}
\tilde{x} (A,B) \defeq \prod_{i \in A \cup B}
x_i.
\end{equation}
By definition, we have $\tilde{x}(A,B) = \tilde{x}(B,A) = \tilde{x}(A \cup B,
\emptyset)$ and $x(\emptyset,\emptyset) = 1$.  Moreover,  the matrix $\tilde{x}$ is positive semidefinite
since it can be written as $\tilde{x} = X X^T$ wherein $X: \cF \to
\R$ is given by $X(A) = \prod_{i \in A} x_i$.  Therefore, for each
assignment $x$, we have $\tilde{x} \in \cS$.

Finally, given an instance $\inst \in \maxpi_n$ its linearization
$\tinst$ is written as follows.  Fix $d \geq 2\lceil k/2\rceil$,
and for every subset $S \subseteq [d]$ with $|S| \leq k$,
define a disjoint union $S = A_S \cup B_S$
where $A_S$ contains (up to) the $\lceil k/2 \rceil$
smallest elements of $S$, and $B_S$ contains the rest
(or is empty).

Each constraint $P_0$ in $\inst$
is of the form $P_0(X) = P(X_{i_1},\ldots, X_{i_k})$ for a predicate
$P: \bits^n \to \bits$ in $\Pi$.  Therefore the function $\inst :
\bits^n \to \R$ given by $\inst(x) = \frac{1}{m} \sum_{i =1}^m P_i(x)$
can be expressed as a degree-$k$ multilinear polynomial in $x$,i.e.,
$$ \inst(x) = \sum_{A \sse [n], |A| \leq k} \hat{\inst}_A \prod_{i \in A}
x_i \mper $$
The linearization $\tinst : \cF \times \cF \to \R$ is given by,
\begin{equation} \label{eq:cspsos2}
\tinst (A,B) \defeq \begin{cases} \hat{\inst}_S & \text{ if } A=A_S, B=B_S
	\\
	0 & \text{ otherwise}	
	\end{cases}
\end{equation}
From \eqref{eq:cspsos1} and \eqref{eq:cspsos2}, for every instance
$\inst \in \maxpi_n$ and every assignment $x \in \bits^n$ we have
$$ \iprod{\tinst, \tilde{x}} = \sum_{A, |A| \leq d/2} \hat{\inst}_A
\prod_{i \in A} x_i = \inst(x) \mper $$

Now the degree-$d$ sos relaxation corresponding to an instance $ \inst \in
\maxpi_n$ in \eqref{eq:sosrelaxation} and \eqref{eq:sosdual}
can be equivalently formulated as
$$ \sosub_d(\inst) \defeq \max_{y \in \cS}\, \iprod{ \tinst, y} \mper
$$

\Pnote{}
\paragraph{$(c,s)$-approximations}
For $0 \leq s \leq c \leq 1$,
a sequence of SDP relaxations $\{ \cS_n \}_{n=1}^{\infty}$ for
$\maxpi$ is said to achieve a {\em $(c,s)$-approximation to
$\maxpi$} if for each $n \in \N$ and every instance $\inst$ of $\maxpi_n$ with $\opt(\inst)
\leq s$, we have $\cS_n(\inst) \leq c$.  In order to study
$(c,s)$-approximations for $\maxpi$, we recall (from \pref{prop:intro-csp2})
the set of matrices $\{M^{n,\Pi}_{c,s}\}_{n=1}^{\infty}$ associated with it, defined as:
\[
M^{n,\Pi}_{c,s}(\inst, x) = c - \inst(x)\,,
\]
where the first index of $M^{n,\Pi}_{c,s}$ ranges over all instances on $n$ variables satisfying $\opt(\inst) \leq s$.
A simple consequence of \pref{prop:gpt} is the following.
\begin{proposition} \label{prop:gptreboot}
	There exists a sequence of SDP relaxations $\cS_n$ of size $r_n$
	achieving a $(c,s)$-approximation to $\maxpi_n$ if and only if
	$\psdrank(M^{n,\Pi}_{c,s}) \leq r_n$.
\end{proposition}

\subsection{General SDPs vs. sum-of-squares}

Our main theorem is that general SDP relaxations for max-CSPs are no more powerful
than low-degree sum-of-squares relaxations in the polynomial-size regime.

\begin{theorem}
  \label{thm:maincsp}
    Fix a positive number $d \in \mathbb N$,
    and a $k$-ary CSP $\maxpi$ with $d \geq 2 \lceil k/2\rceil$.
  Suppose that the degree-$d$ sos relaxation cannot achieve a $(c,s)$-approximation for $\maxpi$.
  Then no sequence of SDP relaxations of size at most
  $o\left(\left(\frac{n}{\log n}\right)^{d/4}\right)$ can achieve a $(c,s)$-approximation
  for $\maxpi$.
\end{theorem}

\begin{proof}
	Given that the degree-$d$ sos relaxation cannot achieve
a $(c,s)$-approximation, there exists an instance $\inst$ of
$\maxpi_m$ for some $m$ such that $\opt(\inst) \leq s$ but
$\sosdeg(c-\inst) > d$.

	By \pref{prop:gptreboot}, it is sufficient to lower bound the
	psd rank of the matrix $M^{n,\Pi}_{c,s}$.
Fix $f = c-\inst$ and define the matrix $M_{n}^f : \binom{n}{m} \times \{0,1\}^n \to [0,1]$ as
$ M_n^f(S,x) \defeq f(x_S)$.  Since $M_n^f$ is a submatrix of
$M^{n,\Pi}_{c,s}$, we have $\psdrank(M^{n,\Pi}_{c,s}) \geq
\psdrank(M_n^f)$.
By \pref{thm:sos-vs-psd}, for some constant $C \geq 1$, we have
$$ \psdrank(M_n^f) \geq \left(\frac{n}{C\log
n}\right)^{d/4} \mper$$
This implies that no sequence of SDP relaxations of size at most
$o( (\frac{n}{\log n})^{d/4})$ can achieve a $(c,s)$-approximation for $\maxpi$.
\end{proof}

For a stronger quantitative bound, we require the following simple fact.

\begin{fact} \label{fact:pseudodist-fourier}
For every positive even integer $d$, every
degree-$d$ pseudo-density $D : \{0,1\}^m \to
	\R$ and every subset $\alpha \subseteq [m], |\alpha| \leq d$, we have
	$$ | \E_x D(x) \chi_{\alpha}(x) | \leq 1\mper$$
\end{fact}
\begin{proof}

Write $\chi_{\alpha} = \chi_{A} \chi_{B}$ for some
$A, B$ with $|A|, |B| \leq d/2$ and observe that
$$ \E_x D(x) \chi_\alpha(x) =  \E_x D(x)
\left(1 -
\frac{(\chi_A-\chi_{B})^2}{2}\right) \leq \E_x D(x) \cdot 1 =
1 \mcom$$
where we used the fact that $\E_x D(x) p(x)^2 \geq 0$ whenever
$\deg(p) \leq d/2$.  Using $\chi_{\alpha} = \frac{1}{2} (\chi_A +
\chi_B)^2 - 1$, we get the other direction of the inequality.
\end{proof}

\begin{theorem}
  \label{thm:csplogn}
Fix a $k$-ary CSP $\maxpi$ and a
   monotone increasing function $d: \N \to \N$
such that the following three conditions are true:  $d(1) \geq 2 \lceil k/2\rceil$, and $d(n) \leq n$ for all $n \geq 1$, and
and $\lim_{n \to \infty} d(n) = \infty$.
   Fix $\epsilon > 0$ and $0 < s < c \leq 1$.  There is a constant $K > 0$ such that the following holds.

   Suppose that for every $n \geq 1$, the degree-$d(n)$ sos
  relaxation cannot achieve a $(c + \e,s)$-approximation for
  $\maxpi_n$.  Then for all $n \geq 1$, no SDP relaxation
  of size at most $K n^{d(n)^2/8}$ can achieve a $(c,s)$-approximation
  for $\maxpi_N$ for every $N > n^{4d(n)}$.
\end{theorem}

\begin{proof}
Without loss of generality, we may assume that $d(n) \geq 24$ is always an even integer.
Given that the degree-$d(n)$ sos relaxation cannot achieve
a $(c + \e,s)$-approximation for $\maxpi_n$, there exists an instance $\inst$ of
$\maxpi_n$ such that $\opt(\inst) \leq s$, along with a
degree-$d(n)$ pseudo-density $D(x)$ such that
$ \E_x D(x) (c-\inst(x)) < -\e  $.
The pseudo-density $D(x)$ can be written as
$$ D(x) = \sum_{\alpha \sse [n], |\alpha| \leq d(n)}
\E_x[D(x) \chi_{\alpha}(x)] \cdot \chi_{\alpha}(x) \mcom$$
where for each $\alpha$, $|\E_x[D(x) \chi_{\alpha}(x)]| \leq
1$ by \pref{fact:pseudodist-fourier}.  Hence,
\[ \norm{D}_\infty \leq \sum_{i=0}^{d(n)} \binom{n}{i} \leq 1+n^{d(n)}\,.\]

Fix $f = c-\inst$ and define the matrix $M_{N}^f : \binom{N}{n} \times
\{0,1\}^n \to [0,1]$ as
$ M_N^f(S,x) \defeq f(x_S)$.  By \pref{thm:gtwo-pseudodist} \eqref{eq:quantitative},
whenever $N > n^{4 d(n)}$, we have
we have
$\psdrank(M_n^f) \geq n^{d(n)^2/8}$.
As $M^{N,\Pi}_{c,s}$ contains $M_N^f$ as a submatrix, the
same lower bound applies to $M^{N,\Pi}_{c,s}$.
This implies that no SDP relaxation of size at most
$n^{d(n)^2/4}$ can achieve a $(c,s)$-approximation for $\maxpi_N$ when
$N > n^{4d(n)}$.
\end{proof}

Using known lower bounds for low degree sum-of-squares relaxations
for max-CSPs \cite{Grigoriev2001, schoenebeck2008linear, DBLP:conf/stoc/Tulsiani09},
\pref{thm:csplogn} implies lower bounds against general SDP
relaxations for a range of specific max-CSPs.  For instance, the lower bounds of
Grigoriev \cite{Grigoriev2001} and Schoenebeck
\cite{schoenebeck2008linear} imply a lower bound for \maxthreesat
(see \pref{thm:maxthreesat}).

\begin{theorem}[\cite{Grigoriev2001,schoenebeck2008linear}]
For every $\e > 0$, there exists a constant $c_{\e}$ such that the following holds.
For every $n \geq 1$, there is a $\maxpi_n$ instance $\inst_n$ such that
$\opt(\inst_n) \leq 7/8 + \e$, but $\sosub_{c_\e n}(\inst)=1$.
\end{theorem}

Observe that one can obtain the bound of \pref{thm:maxthreesat}
using the preceding result as follows.
In \pref{thm:csplogn}, choose $n \asymp \log N$ and $d(n) \asymp \frac{\log N}{\log \log N}$
so that $n^{4 d(n)} \asymp N$.
In that case, the lower bound obtained is of the order $N^{d(n)/32} \asymp N^{\Omega(\log N/\log \log N)}$.

\section{Nonnegative rank}
\label{sec:nnr}

\pref{thm:gtwo-pseudodist} exhibits
a connection between psd rank and sos degree.
There is a similar connection between nonnegative rank and junta-degree.
The results of \pref{sec:nnr-vs-junta} generalize those of \cite{DBLP:conf/focs/ChanLRS13},
while the method of proof is closely related.  As opposed to \cite{DBLP:conf/focs/ChanLRS13},
we use the learning approach of \pref{sec:learning}
to approximate by juntas.
In \pref{sec:corr-lp}, we demonstrate
an application to the correlation polytope.

\subsection{Nonnegative rank vs. junta degree}
\label{sec:nnr-vs-junta}

We recall that the {\em nonnegative rank} of a matrix $M \in \Rnn^{p \times q}$ is the
smallest integer $r \geq 1$ such that there exist $v_1, \ldots, v_p, u_1, \ldots, u_q \in \mathbb R_+^r$
satisfying $M_{ij} = \langle u_i, v_j\rangle$ for all $i \in [p], j \in [q]$.
We denote the minimal value $r$ by $\nnrank(M)$.

\medskip
\noindent
{\bf Junta degree and pseudo-densities.}
Fix a finite set $X$.
For a nonnegative function $f : X^n \to \Rnn$, we say that $f$ has a {\em nonnegative junta certificate
of degree $d$}
if there exist nonnegative $d$-juntas $g_1, g_2, \ldots, g_k : X^n \to \Rnn$ such that
$f = \sum_{i=1}^k g_i$ (as functions on the discrete cube).  The {\em junta degree of $f$}, denoted
$\juntadeg(f)$, is the minimal $d$ such that $f$ has a nonnegative junta certificate of degree $d$.

Consider an arbitrary measure $\mu$ on $X^n$.
A function $D : X^n \to \mathbb R$ is called a {\em $d$-local pseudo-density (with respect to
the measure $\mu$)} if $\E_{\mu} D=1$
and furthermore $\E_{x \sim \mu} D(x) g(x) \geq 0$ for all nonnegative $d$-juntas $g$.
If a measure $\mu$ is unspecified, we always refer to the uniform measure by default.
The following characterization
is immediate from the fact that the set of functions satisfying $\juntadeg(f) \leq d$ is a closed convex cone.

\begin{lemma}
For every $f : X^n \to \mathbb R_+$ and $d \geq 0$, we have
$\juntadeg(f) > d$ if and only if there exists a $d$-local pseudo-density such that $\E_x D(x) f(x) < 0$.
\end{lemma}

We also define a more quantitative notion:  {\em Approximate} junta degree with respect
to an arbitrary measure.
Given $\e > 0$ and a measure $\mu$ on $X^n$, we define
\[
\juntadeg^{\e}(f;\mu) \defeq 1+\max \left\{ d : \exists \text{ a $d$-local pseudo-density $D$ wrt $\mu$ and } \E_{x \sim \mu} D(x) f(x) < -\e \|D\|_{\infty} \E_{\mu} f \right\}\,,
\]
where we take the maximum to be equal to $-1$ if no such pseudo-density exists.
(See \pref{sec:corr-lp} for an example where a biased measure $\mu$
is used to analyze the nonnegative rank of the lopsided disjointness matrix.)
We can now state our main theorem on nonnegative rank.

For any measure $\mu$ on $X$, we use $\mu^n$ to denote
the corresponding product measure on $X^n$.
In the following theorem, we write
$\|f\|_1 \seteq \E_{\mu^m} f$.

\begin{theorem}
\label{thm:nnr-main}
For any finite set $X$, any measure $\mu$ on $X$,
and any $\e > 0$, the following holds.
For any $f : X^m \to \Rnn$ and
all $n \geq 2m$,
\begin{equation}\label{eq:nnr-eq}
	1+  n^{d+1} \geq \nnrank(M_n^f) \geq
\left(\frac{c \e^2 n}{m^2 (d \log n + \log (\|f\|_{\infty}/\|f\|_1))}\right)^{d}\,,
\end{equation}
where $d + 1= \juntadeg^{\e}(f;\mu^m)$ and $c > 0$ is a universal constant.
\end{theorem}
\Pnote{Changed $1+ dn^d$ to $n^d$ since $\sum_{i=0}^d \binom{n}{d}
\leq 1+ n^d$}

\begin{proof}
The left-hand-side inequality of \eqref{eq:nnr-eq}
follows from the fact that the cone of nonnegative $d$-juntas
is spanned by $\sum_{i=0}^{d+1} \binom{n}{i} \leq 1+ n^{d+1}$ nonnegative $d$-juntas.
We move on to right-hand inequality.

We will use $\langle \cdot,\cdot\rangle$ for the inner product on $L^2(X^n;\mu^n)$, i.e.
$\langle g,h\rangle = \E_{\mu^n} [gh]$.
Consider a rank-$r$ nonnegative factorization
\begin{equation}\label{eq:nn-fact}
M_n^f(S,x) = \sum_{i=1}^r \lambda_i(S) q_i(x)\,.
\end{equation}
By rescaling, we may assume that $\E_{\mu^n} q_i=1$ for each $i \in [r]$.
Observe that, by taking expectation on both sides with respect to $\mu^n$,
for any fixed $S$ we have
\begin{equation}\label{eq:lambdas}
\sum_{i=1}^r \lambda_i(S) = \E_{x\sim \mu^n} M_n^f(S,x) = \E_{\mu^m} f = \|f\|_1\,.
\end{equation}

Let $\Lambda_{\tau} = \{ i : \|q_i\|_{\infty} \leq \tau \}$.
Note that for $i \notin \Lambda_{\tau}$, we must have
\begin{equation}\label{eq:badlam}
\lambda_i(S) \leq \frac{\|f\|_{\infty}}{\tau} \quad \forall |S|=m\,.
\end{equation}

Let $D : X^m \to \mathbb R$ be a $d$-local pseudo-density witnessing
$\juntadeg^{\e}(f;\mu^m) > d$.
For $S \subseteq [n]$ with $|S|=m$, we define
a function $D_S : X^n \to \mathbb R$ by $D_S(x)=D(x_S)$.
Note that each $D_S$ is clearly an $m$-junta.

For some $\delta > 0$ to be chosen later,
for each $q_i$ with $i \in \Lambda_{\tau}$, we apply
\pref{thm:junta-approx} to obtain a density $\tilde q_i$
that is a $k'$-junta for $k'=O(\|D\|_{\infty}^2 \frac{m}{\delta^2} \log \tau)$, and such
that for every $S \subseteq [n]$ with $|S|=m$, we have
\begin{equation}\label{eq:approx-q}
\langle D_S, q_i\rangle \geq \langle D_S, \tilde q_i\rangle - \delta\,.
\end{equation}
Let $J_i$ denote the set of coordinates on which $q_i$ depends, so that $|J_i| \leq k'$.

We now take the inner product of both sides of \eqref{eq:nn-fact} with the function $\E_{|S|=m} D_S$.
On the left-hand side, using our assumption on the pseudo-density $D$,
\begin{equation}\label{eq:lhs-q}
\E_S \E_{x\sim \mu^n} D_S(x) M_n^f(S,x) < -\e \|D\|_{\infty} \|f\|_1\,.
\end{equation}
We break the right-hand side of \eqref{eq:nn-fact} into two parts.  First, using \eqref{eq:badlam},
\begin{equation}\label{eq:partone}
\sum_{i \notin \Lambda_{\tau}} \E_S \E_{x\sim \mu^n} D_S(x) \lambda_i(S) q_i(x) \geq -\frac{\|f\|_{\infty}}{\tau} \|D\|_{\infty} r\,.
\end{equation}

For the second part, we use \eqref{eq:approx-q} so that for every $i \in \Lambda_{\tau}$ and $|S|=m$,
\begin{align*}
\E_{x\sim \mu^n} D_S(x)  q_i(x) & \geq
-\delta + \E_{x \sim \mu^n} D_S(x) \tilde q_i(x) \\
&=
-\delta + \E_{\stackrel{y \in X^S}{y \sim \mu^m}} D_S(y) \E_{x \sim \mu^n} [\tilde q_i(x) \mid x_S=y] \\
&\geq -\delta - \|D\|_{\infty} \1_{\{|S \cap J_i| > d\}}\,,
\end{align*}
where in the final line we have used the facts that the function $y \mapsto \E_{x \sim \mu^n} [\tilde q_i(x) \mid x_S=y]$
is a nonnegative $(S \cap J_i)$-junta, and that $D_S : X^m \to \R$ is a $d$-local pseudo-density.

This implies that for $i \in \Lambda_{\tau}$,
\begin{align*}
\E_S \lambda_i(S) \E_{x\sim \mu^n} D_S(x)  q_i(x)
&\geq -\delta \E_S \lambda_i(S) - \|D\|_{\infty} \|\lambda_i\|_{\infty} \Pr\left(|S \cap J_i| > d\right) \\
&\geq -\delta \E_S \lambda_i(S) - \|D\|_{\infty} \|\lambda_i\|_{\infty} \frac{\binom{|J_i|}{d} \binom{n}{m-d/2}}{\binom{n}{m}} \\
&\geq -\delta \E_S \lambda_i(S) - \|D\|_{\infty} \|\lambda_i\|_{\infty} \frac{|J_i|^d m^d}{(n-m)^d} \\
&\geq -\delta \E_S \lambda_i(S) - \|D\|_{\infty} \|f\|_1 \frac{(k')^d (2m)^d}{n^d}\,,
\end{align*}
where in the final line we have used $\|\lambda_i\|_{\infty} \leq \|f\|_1 $ from \eqref{eq:lambdas}
and also $|J_i| \leq k'$ and $n \geq 2m$.

Combining this with \eqref{eq:partone} and \eqref{eq:lhs-q}, we conclude that
\begin{align*}
-\e \|D\|_{\infty} \|f\|_1
&>
\sum_{i=1}^r \E_S \lambda_i(S) \E_{x\sim \mu^n} D_S(x) q_i(x)   \\
&\geq
- \frac{\|f\|_{\infty}}{\tau} \|D\|_{\infty} r
- |\Lambda_{\tau}| \cdot \|D\|_{\infty} \|f\|_1  \frac{(k')^d (2m)^d}{n^d}
-\delta \sum_{i \in \Lambda_{\tau}} \E_S \lambda_i(S) \\
&\geq
- r \|D\|_{\infty} \left(\frac{\|f\|_{\infty}}{\tau}
+ \|f\|_1 \frac{(k')^d (2m)^d}{n^d}\right)
-\delta \|f\|_1 \,.
\end{align*}
Let us now set $\tau = 3 r \frac{\|f\|_{\infty}}{\|f\|_1}$ and $\delta = \e \|D\|_{\infty}/3$, yielding
\begin{equation}\label{eq:penultimate}
r  \geq \frac{1}{3} \left(\frac{n}{2 k' m}\right)^{d}
\geq  \left(c \frac{\e^2 n}{m^2 \log \tau}\right)^{d}
\end{equation}
for some universal constant $c > 0$.

Now, if $r \geq n^d$, then we are done.  Otherwise,
\eqref{eq:penultimate} yields
\[
r
\geq \left(\frac{c \e^2 n}{m^2 (d \log n + \log (\|f\|_{\infty}/\|f\|_1))}\right)^{d}\,,
\]
completing the proof.
\end{proof}

\subsection{The correlation polytope and lopsided disjointness}
\label{sec:corr-lp}

We now illustrate a particularly simple application of our method to nonnegative rank.

\begin{lemma}
\label{lem:lopsided}
There is a constant $\e_0 > 0$ such that for all $m \geq 3$,
the following holds.
Define $f : \{0,1\}^m \to \Rnn$ by
\begin{equation}\label{eq:fdef}
f(x) = \left(1-\sum_{i=1}^m x_i\right)^2\,,
\end{equation}
and let $\mu$ be the measure on $\{0,1\}$ satisfying $\mu(0)=1-2/m$ and $\mu(1)=2/m$.
Then,
\[
\juntadeg^{\e_0}(f;\mu^m) \geq \frac{m}{2}+1\,.
\]
\end{lemma}

Plugging this result into \pref{thm:nnr-main} yields the following.

\begin{theorem}
\label{thm:lopside}
There is a constant $c > 0$ such that
for every $m \geq 3$ and $n \geq 2m$, we have
\[
\nnrank(M_n^f) \geq \left(\frac{cn}{m^3 \log n}\right)^{m/2}\,.
\]
\end{theorem}

In particular, by setting $m=m(n)$ appropriately,
\pref{prop:quadric} implies that $\nnrank(\corr_n) \geq 2^{\Omega(n^{1/3})}$.
One should note that this is somewhat weaker than the lower bound $\nnrank(\corr_n) \geq 2^{\Omega(n)}$
proved in \cite{FMPTW12}.

\begin{proof}[Proof of \pref{lem:lopsided}]

\newcommand{\lf}{\cG}
\newcommand{\zero}{\mathbf{0}}
\newcommand{\one}{\mathtt{1}}
\newcommand{\ei}[1]{\mathbf{e_{#1}}}

For $x \in \bits^m$, let $|x|$ denote the hamming weight of $x$.
Define the pseudo-density $D : \{0,1\}^m \to \R$ with respect to $\mu^m$ by
\begin{align*}
	D(x) \defeq \begin{cases}
		-\frac{1}{\mu^m(\mathbf{0})} &  |x|= 0\\
		\frac{2}{m \mu^m(x)} & |x| = 1 \\
		0    & |x| > 1
	\end{cases}
\end{align*}

We now verify that $D$ is a $d$-local pseudo-density (with respect to $\mu^m$)
for $d=\frac{m}{2}+1$.
Observe first that
\[
\E_{x \sim \mu^m} D(x) = -1 + m \cdot \frac{2}{m} = 1\,.
\]

Let $\beta = \frac{2}{m} \left(1-\frac{2}{m}\right)^{m-1}$ denote $\mu^m(1,0,\ldots,0)$.
Consider a subset $S \subseteq [m]$ and some fixed string $b \in \{0,1\}^S$.
Let $\1_b : \{0,1\}^m \to \{0,1\}$ denote the indicator of whether $x_S=b$.
If $b=\mathbf{0}$, then
\[
\E_{x \sim \mu^m} D(x) \1_b(x) = \beta \left(m-|S|\right) - \left(1-\frac{2}{m}\right)^m
= \left(1-\frac{2}{m}\right)^{m-1} \left(1-2 \frac{|S|-1}{m}\right)
\]
The latter quantity is nonnegative as long as $|S| \leq \frac{m}{2}+1$.

If $|b| > 1$, then $\E_{x \sim \mu^m} D(x) \1_b(x) = 0$, and if $|b|=1$, then
$\E_{x \sim \mu^m} D(x) \1_b(x) \geq 0$ since $D(x) \geq 0$ on the support of $\1_b$.
But any nonnegative $S$-junta is a nonnegative combination of the functions $\1_b$
as $b$ ranges over $\{0,1\}^S$.  We conclude that as long as $d \leq \frac{m}{2}+1$,
$D$ is a $d$-local pseudo-density.

Moreover we have
\[
\E_{x \sim \mu^m} D(x) f(x) = \E_{x \sim \mu^m} D(x) \left(1- 2 \sum_{i=1}^m x_i + \sum_{i=1}^m x_i^2 + 2 \sum_{i \neq j} x_ix_j\right) = -1\,.
\]
Also observe that since $m \geq 3$,
\[
\|D\|_{\infty} = |D(\mathbf{0})| = \left(1-\frac{2}{m}\right)^{-m} \leq 27\,.
\]
Lastly, it is easy to see that $\E_{\mu^m} f \geq \Omega(1)$.
These facts together imply that for some universal constant $\e_0 > 0$, we have
$\juntadeg^{\e_0}(f;\mu^m) \geq m/2+1$, as desired.
\end{proof}
	An interesting feature of the pseudodensity $D$ is that it is
	supported only on $x \in \bits^m$ with $|x| \leq 1$.
	Therefore, the lower bound on the approximate junta degree
	established in \pref{lem:lopsided} applies to any function $f:
	\bits^m \to \R_+$ that satisfies
	$$ f(x) = \begin{cases} 1 & \text{ for } |x| = 0\\
				0 & \text{ for } |x| = 1\mper
		\end{cases}$$
	Moreover, the lower bound on $\nnrank(M_n^f)$ also applies in
	this general setting.  To restate this generalization of \pref{thm:lopside}, let us interpret an element of $\bits^n$ as a subset
of $\{1,2,\ldots,n\}$.
		\begin{corollary}[Lopsided unique disjointness]
	There is a fixed constant $c > 0$ such that for every $m \geq
	3$ and $n \geq 2m$, given a matrix $M : \binom{[n]}{m} \times 2^{[n]} \to \R_+$ satisfying
	$$M(S,T) = \begin{cases} 1 & \text{ if } |S \cap T| = 1
		\\
			0 & \text{ if } |S \cap T| = 0 \mcom
		\end{cases} $$
	we have $\nnrank(M) \geq \left( \frac{cn}{m^3 \log{n}}
	\right)^{m/2}$.	
\end{corollary}
	In other words, the lower bound of
	\pref{thm:lopside} applies to all matrices that have a subset
	of entries corresponding to the {\it unique disjointness}
	problem.

\subsection{Unique games hardness for LPs}
\label{sec:unique-games}

As an illustrative application of the relation between nonnegative
rank and junta-degree (\pref{thm:nnr-main}), we present
an LP hardness result for the Unique Games problem.\footnote{We thank Ola
Svensson for the suggestion to make this explicit.}

Fix an integer $q \geq 1$.  An instance $\inst$ of unique games $\ug^{q}$ consists
of variables $X_1,\ldots, X_n$ taking values in $[q]$ and a
collection of predicates $P_1,\ldots, P_M$ over these variables.
Each constraint $P_i$ is over a pair of distinct variables
$\{X_{a_i},X_{b_i}\}$ and is specified by a bijection $\pi_i: [q] \to
[q]$ as follows:
$$P_i (X_{a_i},X_{b_i}) \defeq \Ind[ \pi_i(X_{a_i}) = \pi_{i}(X_{b_i})
] \mper$$
The goal is to find an assignment that maximizes, over $x \in [q]^n$, the number of
satisfied constraints:
$$ \inst(x) \defeq \frac{1}{M} \sum_{i=1}^M P_i(x)$$
Recall that $\opt(\inst) = \max_{x \in [q]^n} \inst(x)$.
Let $\ug^q_n$ denote the family of Unique Games instances on $n$ variables.
The authors of
\cite{charikar2009integrality} exhibit a
strong integrality gap for Sherali-Adams linear programming
relaxations of $\ug^q$
\begin{theorem}[\cite{charikar2009integrality}]
\label{thm:cmm}
	Fix a number $t \geq 1$ and let $q = 2^t$.  Then for every
	$\delta >
	0$, there exist $\gamma, \e > 0$, an $m \geq 1$, and
	an instance $\inst \in \ug^q_m$ such that
	$$\opt(\inst) \leq \frac{1}{q} + \delta\mcom$$
	but $\juntadeg^{\eps}(1-\delta-\inst) \geq m^{\gamma}$.
\end{theorem}
In fact, the authors of \cite{charikar2009integrality}
construct a lower bound for the $d$-round Sherali-Adams LP 
relaxation (where $d \asymp m^{\gamma})$.  But there is an equivalence between such lower bounds
and the existence of a $d$-local pseudo-density; we refer to 
\cite{DBLP:conf/focs/ChanLRS13} for a discussion.
Applying \pref{thm:nnr-main} (with $X=[q]$ and $\mu$ as the uniform measure on $[q]$),
we obtain the following corollary.

Let $M^{n,\ug^q}_{c,s}$
denote the matrix with entries
\[
M^{n,\ug^q}_{c,s}(\inst,x) = c - \inst(x)\,,
\]
where $\inst$ runs over all $\ug^q_n$ instances with $\opt(\inst) \leq s$, and
all values $x \in [q]^n$.

\begin{corollary}
For every $t \geq 1$, $\delta > 0$, and $d \geq 1$,
there exists a constant $c > 0$ such that for all $n \geq 1$,
\[
\nnrank\left(M^{n,\ug^{q}}_{1-\delta, 1/q+\delta}\right) \geq c n^d\,,
\]
where $q=2^t$.
\end{corollary}

In the language of 
\cite{DBLP:conf/focs/ChanLRS13} (see also \pref{sec:maxcsp} for related
definitions in the SDP setting), this shows that polynomial-size families of
LP relaxations cannot achieve a $(1-\delta,\frac{1}{q}+\delta)$-approximation
for the Unique Games problem.

\ifnum\stocmode=0

\subsection*{Acknowledgments}

This work was supported, in large part, by NSF grant CCF-1407779.
A significant fraction of the project was completed during a long-term
visit of the authors to the Simons Institute for the Theory of Computing (Berkeley)
for the program on Algorithmic Spectral Graph Theory.
The authors would also like to thank Paul Beame, Siu-On Chan, Daniel Dadush, Troy Lee,
Sebastian Pokutta, Pablo Parrilo, Mohit Singh, Ola Svensson, Thomas Rothvo\ss, and Rekha Thomas for valuable discussions
and comments.

\addreferencesection
\bibliographystyle{alpha}
\bibliography{bib/mr,bib/dblp,bib/scholar,bib/lpsize}
\fi

%\clearpage
%\appendix

%\input{content/together}

%\input{content/pseudo-distr}

%\input{content/matchings}

\end{document}